\documentclass[11pt]{article}
\usepackage{tikz}
\usepackage{amsmath,amssymb}
\usepackage{amsthm} 
\usepackage{bbm}  
\usepackage{fancyhdr}
\usepackage[utf8]{inputenc}    
\usepackage{fullpage}
\usepackage{float}
\usepackage{authblk}
\usepackage{color}
\usepackage{enumitem}
\usepackage{ytableau}

\usepackage{color}

\renewenvironment{thebibliography}[1]{
  \begin{oldthebibliography}{#1}
    \setlength{\itemsep}{0em}
    \setlength{\parskip}{0em}
}
{
  \end{oldthebibliography}
}

\title{Spontaneous symmetry breaking in 2D supersphere sigma models and applications to intersecting loop soups}
\author[1,2]{Etienne Granet}
\author[1,2,3]{Jesper Lykke Jacobsen}
\author[1,4]{Hubert Saleur}
\affil[1]{Institut de Physique Th\'eorique, Paris Saclay, CEA, CNRS, 91191 Gif-sur-Yvette, France}
\affil[2]{Laboratoire de Physique Th\'eorique, D\'epartement de Physique de l'ENS,
  \'Ecole Normale Sup\'erieure, Sorbonne Universit\'e, CNRS, PSL Research University, 75005 Paris, France}
\affil[3]{Sorbonne Universit\'e, \'Ecole Normale Sup\'erieure, CNRS, Laboratoire de Physique Th\'eorique (LPT ENS), 75005 Paris, France}
\affil[4]{USC Physics Department, Los Angeles CA 90089, USA}
\date{}

\newtheorem{lemma}{Lemma}
\newtheorem{theorem}{Theorem}

\newcommand{\tr}{\,\text{tr}\,}

\newcommand{\cotan}{\,\text{cotan}\,}

\newcommand{\rvac}{|0\rangle}
\newcommand{\lvac}{\langle 0|}
\newcommand{\normord}[1]{:\mathrel{#1}:}
\newcommand{\expvalue}[2]{\langle #1| #2 |#1\rangle}

\newcommand{\tikTra}{ 
            \begin{tikzpicture}[scale=0.3, baseline=-3.5pt]
   \draw [purple,thick,domain=0:180] plot ({cos(\x)}, {-sin(\x)});
   \draw [purple,thick,domain=0:180] plot ({cos(\x)+1}, {-sin(\x)});
   \draw [purple,thick,domain=0:1] plot ({3}, {-\x});
   \draw [purple,thick,domain=0:1] plot ({4}, {-\x});
            \end{tikzpicture}
}

\newcommand{\tikTrb}{ 
            \begin{tikzpicture}[scale=0.3, baseline=-3.5pt]
   \draw [purple,thick,domain=0:1] plot ({-1}, {-\x});
   \draw [purple,thick,domain=0:180] plot ({cos(\x)*1.5+1.5}, {-sin(\x)*1.5});
   \draw [purple,thick,domain=0:180] plot ({cos(\x)/2+1.5}, {-sin(\x)/2});
   \draw [purple,thick,domain=0:1] plot ({4}, {-\x});
            \end{tikzpicture}
}

\newcommand{\tikfour}[4]{ 
            \begin{tikzpicture}[scale=0.3, baseline=-3.5pt]
                \draw[black,line width=1pt](0,-0.5) -- (#1,0.5); 
                \draw[black,line width=1pt](0.5,-0.5) -- (#2,0.5); 
                \draw[black,line width=1pt](1,-0.5) -- (#3,0.5); 
                \draw[black,line width=1pt](1.5,-0.5) -- (#4,0.5); 
                    \node at (0,-0.5){\tiny $\bullet$};
                    \node at (0,0.5){\tiny $\bullet$};
                    \node at (0.5,-0.5){\tiny $\bullet$};
                    \node at (0.5,0.5){\tiny $\bullet$};
                    \node at (1,-0.5){\tiny $\bullet$};
                    \node at (1,0.5){\tiny $\bullet$};
                    \node at (1.5,-0.5){\tiny $\bullet$};
                    \node at (1.5,0.5){\tiny $\bullet$};
            \end{tikzpicture}
}

\newcommand{\tikthree}[3]{ 
            \begin{tikzpicture}[scale=0.3, baseline=-3.5pt]
                \draw[black,line width=1pt](0,-0.5) -- (#1,0.5); 
                \draw[black,line width=1pt](0.5,-0.5) -- (#2,0.5); 
                \draw[black,line width=1pt](1,-0.5) -- (#3,0.5); 
                    \node at (0,-0.5){\tiny $\bullet$};
                    \node at (0,0.5){\tiny $\bullet$};
                    \node at (0.5,-0.5){\tiny $\bullet$};
                    \node at (0.5,0.5){\tiny $\bullet$};
                    \node at (1,-0.5){\tiny $\bullet$};
                    \node at (1,0.5){\tiny $\bullet$};
            \end{tikzpicture}
}

\newcommand{\tiktwo}[2]{ 
            \begin{tikzpicture}[scale=0.3, baseline=-3.5pt]
                \draw[black,line width=1pt](0,-0.5) -- (#1,0.5); 
                \draw[black,line width=1pt](0.5,-0.5) -- (#2,0.5); 
                    \node at (0,-0.5){\tiny $\bullet$};
                    \node at (0,0.5){\tiny $\bullet$};
                    \node at (0.5,-0.5){\tiny $\bullet$};
                    \node at (0.5,0.5){\tiny $\bullet$};
            \end{tikzpicture}
}

\newcommand{\tikc}{ 
            \begin{tikzpicture}[scale=0.3, baseline=-3.5pt]
                \draw[black,line width=1pt](0,-0.5) to[out=90, in=90] (0.5,-0.5); 
                \draw[black,line width=1pt](0,0.5) to[out=-90, in=-90] (0.5,0.5); 
                    \node at (0,-0.5){\tiny $\bullet$};
                    \node at (0,0.5){\tiny $\bullet$};
                    \node at (0.5,-0.5){\tiny $\bullet$};
                    \node at (0.5,0.5){\tiny $\bullet$};
            \end{tikzpicture}
}

\begin{document}
\maketitle

\begin{abstract}

Two-dimensional sigma models on superspheres $S^{r-1|2s}\cong OSp(r|2s)/OSp(r-1|2s)$ are known to flow to weak coupling $g_\sigma\to 0$ in the IR when $r-2s<2$. Their long-distance properties are described by a free ``Goldstone" conformal field theory (CFT) with $r-1$ bosonic and $2s$ fermionic degrees of freedom,  where the $OSp(r|2s)$ symmetry is spontaneously broken. This behavior is made possible by the lack of unitarity. 

The purpose of this paper is to study logarithmic corrections to the free theory at small but non-zero coupling $g_\sigma$. We do this in two ways. On the one hand, we perform perturbative calculations with the sigma model action, which are of special technical interest since the perturbed theory is logarithmic. On the other hand, we study an integrable lattice discretization of the sigma models provided by vertex models and spin chains with $OSp(r|2s)$ symmetry. Detailed analysis of the Bethe equations then confirms and completes the field theoretic calculations. Finally, we apply our results to physical properties of dense loop soups with crossings.

%
%

\end{abstract}

\newpage

\tableofcontents


\newpage

\section{Introduction}

The Mermin-Wagner theorem---which forbids spontaneous breaking of continuous symmetries in dimensions $D\leq 2$---relies in particular on the assumption of unitarity. It has long been known that when unitarity is broken---typically, in statistical mechanics systems where some ``Boltzmann-weights" can be negative---stable massless Goldstone phases can indeed appear. A simple and beautiful example of this phenomenon is provided by 2D superspin models with orthosymplectic symmetry. These models involve ``spins" with bosonic and fermionic components (more detail below), and enjoy symmetry under a supergroup, generalizing  the usual orthogonal symmetry group to superspins. With $r$ bosonic and $2s$ fermionic degrees of freedom, the symmetry is described by the orthosymplectic group $OSp(r|2s)$ which behaves, in many ways, like the ordinary group $O(N\equiv r-2s)$. The spontaneously broken symmetry phase can be described by a sigma model with target space $OSp(r|2s)/OSp(r-1|2s)\cong S^{r-1|2s}$---a supersphere---and a single coupling constant, $g_\sigma$. The perturbative $\beta$ function to leading order is $\beta={dg_\sigma\over d\log l}\propto (r-2s-2)g_\sigma^2$. For the physically relevant sign of positive $g_\sigma$ and $r-2s>2$ (which includes in particular the ordinary sphere sigma model), the flow is towards  large coupling as usual: fluctuations grow at large distance, and the symmetry is eventually restored. Meanwhile, if $r-2s<2$, the flow is towards weak coupling, and symmetry remains broken. In this case, the fixed point theory in the infra-red is a very simple Goldstone theory made of free massless scalars with $r-1$ bosonic and $2s$ fermionic components. It is a conformal field theory with central charge $c=r-1-2s$ and $r-1$ non-compact directions.

Microscopic realizations of supersphere sigma models were proposed in \cite{denseloops} in terms of a loop soup where loops cover all the vertices of the square lattice, with every edge visited exactly once, every vertex visited exactly twice, and with possibility of crossing. Each crossing is given a special Boltzmann weight $w$, which is the only parameter in the problem, apart from the loop weight, taken to be $r-2s\equiv N$. 
Why this model provides a realization of the spontaneously broken $OSp(r|2s)$ symmetry  phase was discussed in detail in \cite{denseloops}. In that reference, it was in particular checked numerically that the IR properties of the model are given by those of the Goldstone theory with central charge $c=r-1-2s$ indeed. 

The lattice model of dense loops is expected to provide a  physical realization of the broken symmetry phase for all finite, non-zero values of the coupling $w$. Different values of $w$ correspond to different bare values of the coupling constant $g_\sigma$. Since the bare $g_\sigma$ is usually finite,  there will be logarithmic corrections to scaling, appearing for instance in powers of $\log x$ corrections to the pure conformal behavior expected in the fixed point Goldstone theory.  The main goal of this paper is to study these corrections. This is interesting for several reasons. From a practical perspective,  these corrections affect directly the correlations measured in dense loop models, which have a variety of interesting applications.  From a more fundamental perspective, the dense loop model provides a regularization of the supersphere sigma model which is rather easy to study both analytically (via the Bethe-ansatz, see below) and  numerically, since it involves lattice models with a small number of discrete, compact degrees of freedom. A lot of weak-coupling properties can then be investigated, which are much harder to get in ordinary sigma models such as $O(3)$ \cite{niedermayer}. Finally, we recall that 2D sigma models on super-targets appear in a variety of contexts from string theory \cite{schomerusreview} to the study of phase transitions in non-interacting disordered 2D quantum electron gases \cite{Zirnbauer,Zirnbauer1,gruzberg} to the study of many statistical mechanics models \cite{parisisourlas,readsaleursigma}. Our understanding of these models remains sketchy, largely because the loss of unitarity leads to unpleasant features, such as indecomposable action of the conformal symmetry \cite{LCFTreview}. 

The supersphere sigma models provide an interesting ground where to gain experience in these matters---here, for instance, by investigating perturbation of logarithmic conformal field theory. 

 \begin{figure}
 \begin{center}
\includegraphics[scale=0.5]{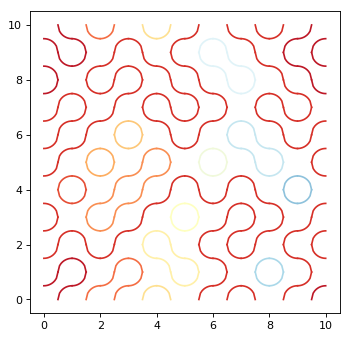} 
\includegraphics[scale=0.5]{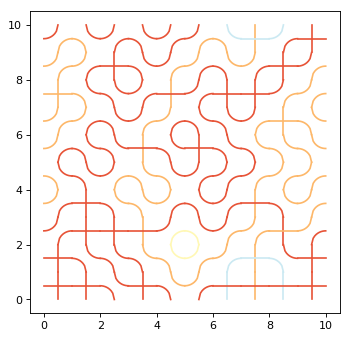} 
\end{center}
\caption{A sample configuration of a dense loop soup where crossings are forbidden (left) or allowed (right). For clarity, each loop has been given a different color: these colors do not form part of the definition of the model.}
\label{fig:loopsoup}
\end{figure}

Remarkably, there exists a series of integrable vertex models known to be in the universality class of the $OSp(r|2s)$ supersphere sigma models for $r-2s<2$ \cite{nienhuisrietman,martinsnienhuisrietman}. These models have edge degrees of  freedom taking values in the $OSp(r|2s)$ fundamental representations, and are closely associated with special points in the phase diagram of the dense loop model \cite{denseloops}. In the  limit of small spectral parameter, the transfer matrix of these vertex models gives rise to a spin chain Hamiltonian acting on the tensor product of $r+2s$ dimensional representations. 
 
 Our strategy in this paper will be to study the weak-coupling physics of the supersphere sigma models mostly by analyzing their integrable lattice regularizations using the Bethe ansatz technique. We will also compare these results with those of (logarithmic) conformal perturbation theory.

 It is important to stress that integrable spin chains are not usually associated with free Goldstone theories such as the ones we will encounter below.%
  \footnote{Note in particular that for $r>1$ the associated CFTs will have non-compact degrees of freedom, even though the spin chains  have a finite number of states per site.}
  The well known $SU(2)$ spin-${1\over 2}$ chain for instance can  be considered as a lattice regularization of the $O(3)$ sigma model but with topological angle $\theta=\pi$, and bare coupling of order unity \cite{affleck},\cite{haldane}. It 
flows to strong coupling in the IR, with long-distance conformal properties described by a  level-one Wess-Zumino theory (thus with left and right $SU(2)$ current algebra symmetries). Here in contrast, the  spin chains provide regularizations of  sigma models with no topological term, small bare coupling constants, flowing to weak coupling in the IR, with long-distance properties described by free Goldstone CFTs where the $OSp$ symmetry is spontaneously broken.

 The plan of the paper is as follows. We focus in the first section on  $OSp(1|2)$, a case which has already been partly studied in \cite{birgit2}.  We compute logarithmic corrections to the gaps directly from the   $\sigma$-model Hamiltonian, and match them with the Bethe-ansatz results. Such a field theory calculation of the low-lying spectrum directly from the sigma model Hamiltonian, and its comparison to the spin chain discretization, has not been done before to our knowledge.  The sigma model analysis involves perturbation of a logarithmic CFT, which we compare with  a discussion by  Cardy of  logarithmic corrections in ordinary CFTs perturbed by marginal operators, and explain to which extent the results on perturbed CFTs apply to the logarithmic case. We then move on to study of the $OSp(r|2)$ chains with $r=2,3,4$ in the following three  sections with a similar approach. In the last section, we apply our results to the discussion of physical observables in dense loop soups and derive new predictions for the power-law and logarithmic decay of intersecting loops correlations.  Some technical aspects are studied in the appendices. In particular, in Appendix~\ref{sec:appc}, we  give formulas for logarithmic corrections directly from the Bethe ansatz equations, and study how these are perturbed when there is an additional source term at one site in the Bethe equations (that occurs in case of modified boundary conditions, or inclusions of strings).

We note that our work is not the first to address the topic of logarithmic corrections in orthosymplectic spin chains. 
In a previous paper \cite{frahmmartins} the logarithmic corrections for some states in the  $OSp(3|2)$ chain (technically, those with spins $j=1/2$, $q$ arbitrary, see below) were  obtained numerically and interpreted as the value of the Casimir. A similar exercise was carried out in  \cite{FM} for similar states and other $OSp(r|2s)$ models. Our paper extends and in some cases significantly corrects the results of these two references. It also  provides analytical results from the Bethe-ansatz, as well as a detailed study of the relationship with the sigma models  (including in particular a perturbative calculation of their  properties) and the dense loop soups.

Before entering the subject, we briefly give a few definitions related to orthosymplectic spin chains.

\subsection{Definitions/Reminders}

\subsubsection{The orthosymplectic symmetry}

We give here a brief description of the orthosymplectic symmetry. 

We first remind that the superspace $\mathbb{R}^{r|2s}$ is parametrized by $r$ 'bosonic' variables $\phi_1,...,\phi_r$ and $2s$ 'fermionic' variables $\eta^1_1,\eta^2_1,...,\eta^1_s\eta^2_s$ that satisfy $[\phi_i,\phi_j]=0$, $[\phi_i,\eta^k_l]=0$, $\{ \eta^i_j, \eta^k_l \}=0$. The scalar product between two vectors $x,y\in \mathbb{R}^{r|2s}$ is defined by
\begin{equation}
\langle x,y\rangle = x^t\cdot  J_{r|2s}\cdot y\,,
\end{equation}
where
\begin{equation}
J_{r|2s}=\left(
\begin{matrix}
I_r &O_{r\times 2s}\\
O_{2s\times r} & I_{s}\otimes J_{0|2}
\end{matrix}
\right)\,, \quad J_{0|2}=\left(
\begin{matrix}
0&1\\
-1&0
\end{matrix}
\right)\,,
\end{equation}
with $I_r$ the identity matrix of size $r\times r$, and $O_{a\times b}$ the zero matrix of size $a\times b$.

The group $OSp(r|2s)$ is the set of linear transformations on $\mathbb{R}^{r|2s}$ that leave the norm invariant:
\begin{equation}
OSp(r|2s)=\{ M \in \mathcal{M}_{r+2s,r+2s}\,, \quad \forall x\in \mathbb{R}^{r|2s}\,, \langle x, M x\rangle = \langle x, x\rangle\}\,.
\end{equation}

The Lie superalgebra $osp(r|2s)$ of such a group can be represented as \cite{scheunert}
\begin{equation}
osp(r|2s)=\left\lbrace
\left(
\begin{matrix}
A & X & B\\
-Y & E & X^t \\
C & Y^t & -A^t
\end{matrix}
\right); A,B,C \in \mathcal{M}_{s,s}; X\in \mathcal{M}_{s,r}; Y\in \mathcal{M}_{r,s}; E\in \mathcal{M}_{r,r}; E^t=-E
\right\rbrace\,.
\end{equation}
In this representation the generators for $osp(r|2)$ will be denoted $J_z,J_+,J_-$, $F_1,...,F_r$, $G_1,...,G_r$, $Q_{ij}$, $i,j=1,...,r$, $i>j$ and are such that
\begin{equation}
\label{eq:generators}
j_zJ_z+j_+J_++j_-J_-+\sum_{k=1}^r (f_k F_k + g_k G_k) + \sum_{i>j}q_{ij}Q_{ij} =\left(
\begin{matrix}
j_z/2 & f_1 &...& f_n & j_+\\
g_1 & 0 &...& q_{n1}& f_1\\
... &...  &...&... & ..\\
g_n & -q_{n1} &...& 0 & f_n\\
j_- & -g_1 &...& -g_n & -j_z/2\\
\end{matrix}\right)\,.
\end{equation}
The commutation relations of the generators can be read off directly from this matrix representation. For example, for $osp(1|2)$ one has $5$ generators $J_z,J_+,J_-,F_+,F_-$ that satisfy the following relations
\begin{equation}
\label{eq:osp12}
\begin{aligned}
&[J_z,J_\pm ]=\pm J_\pm\,, \quad [J_+,J_-]=2J_z\\
&[J_z,F_1]=\frac{1}{2}F_1\,,\quad [J_z,G_1]=-\frac{1}{2}G_1\,, \quad [J_+,G_1]=-F_1\,,\quad [J_-,F_1]=-G_1\\
&\{F_1,F_1\}=2J_+\,,\quad \{G_1,G_1\}=-2J_-\,, \quad \{F_1,G_1\}=2J_z\,.
\end{aligned}
\end{equation}
As for the Casimir, we normalize it in such a way that it is the inverse of the Killing form of the generators \eqref{eq:generators} (i.e., the quadratic term in the eigenvalue $j$ of $J_z$ is $2j^2$). For example for $osp(1|2)$ it reads
\begin{equation}
{\cal C}=2J_z^2+(J_+ J_-+J_-J_+)+\frac{1}{2}(G_1F_1-F_1G_1)\,.
\end{equation}



The representation theory of the $osp(r|2s)$ algebras is a bit complicated, and involves (except when $r=1$) issues of typicality.  Rather than give generalities at this stage, we will recall necessary features in our case by case analysis below. 

\subsubsection{The supersphere $\sigma$-model}
A simple but non-trivial field theory with $OSp(r|2s)$ symmetry is the  theory of ``free'' fields constrained to lie on the supersphere of dimensions $(r-1|2s)$, i.e. the subset $S^{r-1|2s}\subset \mathbb{R}^{r|2s}$ such that $\forall x\in S^{r-1|2s}, \langle x, x\rangle=1$. This is the  non-linear $\sigma$-model with target space the supersphere $S^{r-1|2s}$. We will restrict in the following to the case $s=1$, and thus symmetries $OSp(r|2)$. The action is
\begin{equation}
S(\phi_1,...,\phi_r,\eta^1,\eta^2)=\frac{\kappa}{4\pi g_\sigma}\int  dxdt  \left( \sum_{i=1}^r \partial_\mu\phi_i\partial_\mu\phi_i+2\partial_\mu\eta^2\partial_\mu\eta^1 \right)\,,
\end{equation}
with  $\partial_\mu X\partial_\mu X \equiv-(\partial_x X)^2+(\partial_t X)^2$, and $\kappa$ a normalization factor to make matching with existing literature easier.


The constraint then translates into
\begin{equation}
\label{supersphere}
\sum_{i=1}^r\phi_i^2 +2\eta^2\eta^1=1\,.
\end{equation}
The model for $r<4$ is known to flow to a Goldstone free theory. From integration of the leading order in the $\beta$ function ${dg_\sigma\over d\log L}\propto (r-2s-2)g_\sigma^2$ we find, after properly adjusting normalizations, replacing the RG scale by the size of the system, and setting $g\equiv g_\sigma$ for simplicity
 \cite{wegner1989,floratospecther}
\begin{equation}
g\approx \frac{\kappa}{(4-r)\log L/L_0} \,, \quad \mbox{for } L\to\infty\,,
\label{scalingofg}
\end{equation}
where $L_0$ is, at the order we are working,  an irrelevant  length scale we shall take equal to unity in the following.

\subsubsection{The orthosymplectic spin chain \label{sec:chain}}

We will study spin chains built from an $R$-matrix that satisfies the Yang-Baxter equation and that belongs to the fundamental representation of the $osp(r|2s)$ superalgebra. Periodic boundary conditions will be considered in this paper, although the open boundary case is briefly addressed in appendix \ref{sec:appc}.

Consider a spin chain on $L$ sites with periodic boundary conditions, each site $i$ being described by a $\mathbb{Z}_2$-graded vector space $V_i$ of dimension $D=r+2s$. The Grassman parity $p_\alpha$ of the $\alpha$-th degree of freedom is defined as
\begin{equation}
\label{eq:grading}
\begin{aligned}
p_\alpha&=1\quad \text{if }\alpha=1,...,s \quad \text{ or }\quad \alpha=r+s+1,...,r+2s\\
&=0\quad \text{if }\alpha=s+1,...,r+s\,.
\end{aligned}
\end{equation}
The Boltzmann weight of the chain is given by a matrix $R_{ab}$ that acts on the (graded) tensor product of two spaces $V_a\times V_b$. It is thus a square matrix of size $D^2\times D^2$. The transfer matrix $t(\lambda)$ of the model at spectral parameter $\lambda$, that acts on the tensor product $V_L\otimes ...\otimes V_1$ of $L$ spaces, is given by the supertrace of the monodromy matrix
\begin{equation}
\label{eq:T}
\begin{aligned}
t(\lambda)&=\sum_{i=1}^D (-1)^{p_i}T_{ii}(\lambda)\\
\text{with }\quad T(\lambda)&=R_{aL}(\lambda)R_{aL-1}(\lambda)...R_{a1}(\lambda)\,.
\end{aligned}
\end{equation}
 $T_{ij}(\lambda)$ denotes the $(i,j)$ component of $T$ in the auxiliary vector space $V_a$. The notation $R_{a i}$ means that $R$ acts on the whole tensor product $V_a\otimes V_L\otimes ...\otimes V_1$, but non-trivially only on the spaces $V_a$ and $V_i$. However the graded tensor product introduces signs everywhere, and for clarity we give the explicit expression of the components of the transfer matrix
\begin{equation}
\label{tm}
\begin{aligned}
t(\lambda)^{\alpha_1...\alpha_L}_{\beta_1...\beta_L}=\sum_{c_1,...,c_L=1}^D R(\lambda)^{\alpha_L c_L}_{c_1 \beta_L}R(\lambda)^{\alpha_{L-1} c_{L-1}}_{c_L \beta_{L-1}}...R(\lambda)^{\alpha_1 c_1}_{c_2 \beta_1}
 (-1)^{p_{c_1}}(-1)^{\sum_{j=2}^L(p_{\alpha_j}+p_{\beta_j})\sum_{i=1}^{j-1}p_{\alpha_i}}\,.
\end{aligned}
\end{equation}
Such an explicit expression can be found in \cite{essler}.

Let us now choose a particular matrix $R$. We impose the two following conditions
\begin{enumerate}[label=(\roman*)]
\item $R$ satisfies the graded Yang-Baxter equation
\begin{equation}
\label{yb}
R_{12}(\lambda)R_{13}(\lambda+\mu)R_{23}(\mu)=R_{23}(\mu)R_{13}(\lambda+\mu)R_{12}(\lambda)\,,
\end{equation}

\item $R$ has $osp(r|2s)$ symmetry:  that is, for each generator $A$ of $osp(r|2s)$ (in a certain representation, not necessarily \eqref{eq:generators}), $t(\lambda)$ commutes with $A_{\rm tot}=\sum_{i=1}^L A_i$ where $A_i$ acts on the vector space $V_i$.

\end{enumerate}
We will write $R_{ik}^{lj}$ the component along $e_j\otimes e_l$ of the action of $R$ on $e_i\otimes e_k$, where the $e_i$'s denote the basis vector of $V$.

The following $R$-matrix is known to satisfy these two properties \cite{arnaudonavan2}
\begin{equation}
\label{Rmatrix}
R_{ab}(\lambda)=\lambda I_{ab}+P_{ab}+\frac{2\lambda}{2-r+2s-2\lambda}E_{ab}\,,
\end{equation}
where $I_{ab}$ is the $D^2\times D^2$ identity matrix, $P_{ab}$ is the graded permutation operator
\begin{equation}
(I)^{lj}_{ik}=\delta_{ij}\delta_{kl}\,,\qquad (P)^{lj}_{ik}=(-1)^{p_i p_j}\delta_{il}\delta_{jk}\,,
\label{idpermop}
\end{equation}
and $E_{ab}$ the matrix given by
\begin{equation}
(E)^{lj}_{ik}=(-1)^{i>r+s}(-1)^{j\leq s} \delta_{ik'}\delta_{jl'}\,,
\end{equation}
with $i'=D+1-i$, and $(-1)^{x>y}$ is $-1$ if $x>y$, $1$ if $x\leq y$. The Hamiltonian of the chain is then defined as
\begin{equation}
\mathcal{H}=\frac{d}{d\lambda} \log t(\lambda)\rvert_{\lambda=0}\,.
\end{equation}

The Yang-Baxter equation ensures that the transfer matrices $t(\lambda)$ and $t(\mu)$ at different spectral parameters commute, and it has been shown that the eigenvalues of $t$ can be determined with the Bethe ansatz \cite{galleasmartins,arnaudonavan} (completeness is not proven; moreover some eigenvalues may be a bit singular, like the ground state of the $osp(3|2)$ spin chain for example that is obtained with coinciding Bethe roots that should be normally excluded) .

\subsubsection{The critical exponents from the spin chain\label{sec:cftspectrum}}
If a model is described by a Conformal Field Theory (CFT) in the thermodynamic limit, then the content of the theory (central charge, conformal dimensions, structure constants; and in case of LCFT, the logarithmic couplings) can be seen from the finite-size corrections to the low-lying excited energies for large system sizes \cite{blotecardy,affleckc,cardyoperator}. More precisely, denoting $E_L^0$ the energy of the ground state of a periodic chain, and $e_L^0=E_L^0/L$ its intensive form, one has
\begin{equation}
e_L^0=e_\infty-\frac{\pi v_F }{6L^2}c+o(L^{-2})\,,
\end{equation}
where $e_\infty$ is the thermodynamic value of the energy, $c$ is the central charge and $v_F$ the Fermi velocity. The finite-size corrections to the energy of the excited states $e_L$ contain the conformal dimensions $h+\bar{h}$ of the theory
\begin{equation}
\label{eq:defsmallE}
e_L-e_L^0=\frac{2\pi v_F }{L^2}(h+\bar{h})+o(L^{-2})\,,
\end{equation}
Some structure constants can be seen in the next-order finite-size corrections, see \cite{cardyoperator}. In case of logarithmic finite-size corrections
\begin{equation}
e_L-e_L^0=\frac{2\pi v_F }{L^2}(h+\bar{h})+\frac{2\pi v_F}{L^2\log L} \alpha+o(L^{-2}(\log L)^{-1})\label{scalegaps}\,,
\end{equation}
the correlation functions of the field $\phi$ associated to $e_L$ are expected to satisfy at large $x$
\begin{equation}
\langle \phi(x)\phi(0)\rangle \sim \frac{1}{x^{2(h+\bar{h})}(\log x)^{2\alpha}}\label{afflecklogeq}\,,
\end{equation}
see \cite{afflecklog} for a derivation.

In the following we will often refer to the quantity ${L^2\over 2\pi v_F}(e_L-e_L^0)$ as a ``scaled gap''. 
We will also sometimes denote $e_L-e_L^0$ by $\Delta e_L$.

\section{$OSp(1|2)$}
\subsection{The spectrum from field theory}
We start by deriving the expected low-lying spectrum of the non-linear $\sigma$-model with $OSp(1|2)$ symmetry at first order in $(\log L)^{-1}$.

\subsubsection{General strategy}
There is a common strategy for studying the different models. The first step is to derive the Hamiltonian in terms of the modes of the fields from the lagrangian. 
The expectation values of the Hamiltonian within modes are naively divergent: to regularize them, we express them in terms of their normal-ordered versions and isolate infinite sums. Every regularized value for these gives a Hamiltonian with the desired $osp$ symmetry, and they have to be fixed by an additional condition. Once the expression of the states in term of the modes are derived, the eigenvalues of the Hamiltonian at order $g$ can then be obtained.
This gives access to the logarithmic corrections by using \eqref{scalingofg}.

\subsubsection{The action}
In the $osp(1|2)$ case the constraint \eqref{supersphere} can be satisfied by imposing
\begin{equation}
\phi_1=1-\eta^2\eta^1\,.
\end{equation}
The action becomes then
\begin{equation}
S(\eta^1,\eta^2)=\frac{\kappa}{2\pi g}\int dxdt(\partial_\mu \eta^2\partial_\mu \eta^1 -\eta^1\eta^2\partial_\mu\eta^1\partial_\mu\eta^2)\,.
\end{equation}
Upon the change of variable $\eta^{1,2}\to \sqrt{g}\eta^{1,2}$ it reads
\begin{equation}
\label{osp12action}
S(\eta^1,\eta^2)=\frac{\kappa}{2\pi}\int dxdt(\partial_\mu \eta^2\partial_\mu \eta^1 -g\eta^1\eta^2\partial_\mu\eta^1\partial_\mu\eta^2)\,.
\end{equation}
Here $g$ shall be treated at order $1$. Recall that for large $L$ we have from \eqref{scalingofg}
\begin{equation}
g=\frac{\kappa}{3\log L}\,.
\end{equation}

\subsubsection{The Hamiltonian}
The first task is to derive the Hamiltonian corresponding to the action \eqref{osp12action}. For calculation convenience we write \eqref{osp12action} as
\begin{equation}
S=\frac{\kappa}{2\pi}\int dxdt(-\partial_x\eta^2\partial_x\eta^1+\dot{\eta}^2\dot{\eta}^1+g\mathcal{V}(\eta^1,\eta^2))\,,
\end{equation}
with $\mathcal{V}(\eta^1,\eta^2)$ a generic potential. 
We take the system to be defined on a cylinder of unit radius, so that $x$ is integrated between $0$ and $2\pi$, and $t$ between $0$ and an arbitrary final time $T$. The derivatives with respect to fermions are always considered from the right (this means for example that $(d/d\dot{\eta}^2) (\dot{\eta}^2\dot{\eta}^1)=-\dot{\eta}^1$). In terms of modes
\begin{equation}
\eta^{1,2}(x,t)=\sum_k\eta^{1,2}_k(t)e^{ikx}\,,
\end{equation}
the action reads
\begin{equation}
S=\kappa\int dt\left(\sum_k-k^2\eta^2_k\eta^1_{-k}+\dot{\eta}^2_k\dot{\eta}^1_{-k}+gV\right)=\int dt \mathcal{L}\,,
\end{equation}
with $\mathcal{L}$ the lagrangian density, and
\begin{equation}
V(t)=\frac{1}{2\pi}\int dx\mathcal{V}(x,t)\,.
\end{equation}
The conjugate momenta to $\eta^1_k$ and $\eta^2_k$ are
\begin{equation}
\pi^1_k=\kappa\dot{\eta}^2_{-k}+\kappa g\frac{dV}{d\dot{\eta}^1_k},\quad \pi^2_k=-\kappa\dot{\eta}^1_{-k}+\kappa g\frac{dV}{d\dot{\eta}^2_k}\,.
\end{equation}
The quantization procedure imposes the following anticommutators 
at equal times at all orders in $g$:
\begin{equation}
\label{canonical}
\{\eta^{1,2}_k,\pi^{1,2}_p\}=i\delta_{k,p}\,.
\end{equation}
The corresponding Hamiltonian is then defined as
\begin{equation}
H=\sum_k (\pi^1_k\dot{\eta}^1_k-\dot{\eta}^2_k\pi^2_k)-\mathcal{L}\,.
\end{equation}
It gives, neglecting terms of order $O(g^2)$:
\begin{equation}
H=\sum_k (\kappa k^2\eta^2_k\eta^1_{-k}+\kappa^{-1}\pi^2_k\pi^1_{-k})-\kappa g  V+O(g^2)\,.
\end{equation}
We now use the following expression of the time derivative for a quantity $X$
\begin{equation}
\dot{X}=i[H,X]\,,
\end{equation}
and the relation valid for all (commuting or anticommuting) quantities $a_i$ and $b$ ($n\geq 2$)
\begin{equation}
[a_n...a_1,b]=\sum_{i=1}^n (-1)^{i-1} a_n...a_{i+1}\{a_i, b\}a_{i-1}...a_1\,,
\end{equation}
to compute
\begin{equation}
\label{derivatives}
\begin{aligned}
\dot{\eta}^1_k&=-\kappa^{-1}\pi^2_{-k}+\kappa g\frac{dV}{d\pi^1_k}\,,\qquad \dot{\eta}^2_k=\kappa^{-1}\pi^1_{-k}+\kappa g\frac{dV}{d\pi^2_k}\\
\dot{\pi}^1_k&=-\kappa k^2\eta^2_{-k}+\kappa g\frac{dV}{d\eta^1_k}\,,\qquad \dot{\pi}^2_k=\kappa k^2\eta^1_{-k}+\kappa g\frac{dV}{d\eta^2_k}\,.
\end{aligned}
\end{equation}
Now let us define the following charges
\begin{equation}
\begin{aligned}
J_z&=\sum_k \frac{\eta^1_k\pi^1_k-\eta^2_k\pi^2_k}{2i}\,,\qquad J_+=-\sum_ki\eta^1_k\pi^2_k\,,\qquad J_-=-\sum_k i\eta^2_k\pi^1_k\\
F_1&=\sum_{k+l+m=0}\pi^2_{-k}(\delta_{l,0}\delta_{m,0}+g\eta^1_l\eta^2_m)\,,\qquad G_1=\sum_{k+l+m=0}\pi^1_{-k}(\delta_{l,0}\delta_{m,0}-g\eta^2_l\eta^1_m)\,.
\end{aligned}
\end{equation}
Using \eqref{canonical} one can check that they satisfy the $osp(1|2)$ relations \eqref{eq:osp12}---remember that before \eqref{osp12action} we made the replacement $\eta^{1,2}\to\sqrt{g}\eta^{1,2}$. With the formulas \eqref{derivatives} one has:
\begin{equation}
\label{charges_deriv}
\begin{aligned}
\partial_t J_z&=\sum_k \frac{\kappa g}{2}\left(\frac{dV}{d\pi^1_k}\pi^1_k+\eta^1_k\frac{dV}{d\eta^1_k}-\eta^2_k\frac{dV}{d\eta^2_k}-\frac{dV}{d\pi^2_k}\pi^2_k \right)\\
\partial_t J_+&=\sum_k \kappa g\left(\frac{dV}{d\pi^1_k}\pi^2_k+\eta^1_k\frac{dV}{d\eta^2_k} \right)\,,\qquad \partial_t J_-=\sum_k \kappa g\left(\frac{dV}{d\pi^2_k}\pi^1_k+\eta^2_k\frac{dV}{d\eta^1_k} \right)\\
\partial_t F_1&=\sum_{k+l+m=0}g\left( \kappa\frac{dV}{d\eta^2_0}\delta_{l,0}\delta_{m,0}-\kappa km\eta^1_l\eta^2_m\eta^1_k+\kappa^{-1}\eta^1_k\pi^1_{-l}\pi^2_{-m}\right)\\
\partial_t G_1&=\sum_{k+l+m=0}g\left(\kappa\frac{dV}{d\eta^1_0}\delta_{l,0}\delta_{m,0}+\kappa km\eta^2_l\eta^2_m\eta^1_k-\kappa^{-1}\eta^2_k\pi^1_{-l}\pi^2_{-m} \right)\,,
\end{aligned}
\end{equation}
where the following relation has been used (it is an integration by part)
\begin{equation}
\begin{aligned}
\sum_{k+l+m=0}k^2\eta^1_k\eta^1_l\eta^2_m=\sum_{k+l+m=0}k(k-k-l-m)\eta^1_k\eta^1_l\eta^2_m=\sum_{k+l+m=0}-km\eta^1_k\eta^1_l\eta^2_m\,.
\end{aligned}
\end{equation}

With the equations \eqref{charges_deriv} it can be checked that the following potential implies the conservation of all these charges:
\begin{equation}
\mathcal{V}(\eta^1,\eta^2)=\eta^2\eta^1\partial_x\eta^2\partial_x\eta^1-\eta^2\eta^1\partial_t\eta^2\partial_t \eta^1\,,
\end{equation}
or in terms of $V$
\begin{equation}
V=\sum_{k+l+m+n=0}\left(-mn\eta^2_k\eta^1_l\eta^2_m\eta^1_n+\kappa^{-2}\eta^2_k\eta^1_l\pi^1_{-m}\pi^2_{-n}\right)\,.
\end{equation}
Define now the following modes for all $k$
\begin{equation}
\begin{aligned}
\psi^1_k&=\frac{-ik\eta^1_k\kappa^{1/2}-\pi^2_{-k}\kappa^{-1/2}}{\sqrt{2}},\quad \psi^2_k=\frac{-ik\eta^2_{k}\kappa^{1/2}+\pi^1_{-k}\kappa^{-1/2}}{\sqrt{2}}\\
\bar{\psi}^1_k&=\frac{-ik\eta^1_{-k}\kappa^{1/2}-\pi^2_{k}\kappa^{-1/2}}{\sqrt{2}},\quad \bar{\psi}^2_k=\frac{-ik\eta^2_{-k}\kappa^{1/2}+\pi^1_{k}\kappa^{-1/2}}{\sqrt{2}}\,.
\end{aligned}
\end{equation}
They satisfy
\begin{equation}
\{\psi^1_k,\psi^2_p\}=k\delta_{k+p,0}\,,\qquad \{\bar{\psi}^1_k,\bar{\psi}^2_p\}=k\delta_{k+p,0}
\end{equation}
for $k,p\neq 0$, the other anticommutators being zero. The original modes for $k\neq 0$ read in terms of the $\psi$'s:
\begin{equation}
\begin{aligned}
\eta^{1,2}_k&= \frac{i}{k\sqrt{2\kappa}}(\psi^{1,2}_k-\bar{\psi}^{1,2}_{-k})\\
\pi^{1,2}_{-k}&=\pm\sqrt{\kappa/2}(\psi^{2,1}_k+\bar{\psi}^{2,1}_{-k})\,.
\end{aligned}
\end{equation}

The potential then reads
\begin{equation}
V=\sum_{k+l+m+n=0}\frac{1}{2\kappa^2}\left(-i\sqrt{2\kappa}\eta^2_0+\frac{\psi^2_k-\bar{\psi}^2_{-k}}{k}\right)\left(-i\sqrt{2\kappa}\eta^1_0+\frac{\psi^1_l-\bar{\psi}^1_{-l}}{l}\right)(\psi^2_m\bar{\psi}^1_{-n}+\bar{\psi}^2_{-m}\psi^1_n)\,.
\end{equation}

\subsubsection{Normal order}
Up to now no normal order has been put on the fields. The elementary annihilation operators are set to be the $\psi^{1,2}_m$, $\bar{\psi}^{1,2}_m$ with $m\geq 0$, and the elementary creation operators the same modes but for $m<0$, as well as $\eta^{1,2}_0$. The normally ordered version of an operator $X$ is denoted $\normord{X}$ and is defined, for every product of elementary operators that appear in the expression of $X$, by putting all the annihilation operators to the right of their creation operators, multiplied by the corresponding fermionic sign. Equivalently (this is much more convenient for practical purposes), it amounts to forbidding contractions between modes that compose the Hamiltonian. Indeed, the only difference between the Hamiltonian and its normally-ordered version is the contractions that appear whenever an annhilation operator is moved to the right of a creation operator. For example if one computes the expectation value
\begin{equation}
\begin{aligned}
\lvac \psi_{1}^2 \left( \sum_k \normord{\psi^2_k\psi^1_{-k}}\right) \psi^1_{-1}\rvac&= \lvac \psi_{1}^2 \left( \sum_{k<0}\psi^2_k\psi^1_{-k}-\sum_{k>0}\psi^1_{-k}\psi^2_{k}+\psi^2_0\psi^1_0\right) \psi^1_{-1}\rvac\\
&= \lvac \psi_{1}^2 \left(- \psi^1_{-1}\psi^2_{1}\right) \psi^1_{-1}\rvac\\
&=-1\,,
\end{aligned}
\end{equation}
one actually contracts the $\psi^2_1$ inside the Hamiltonian with the $\psi^1_{-1}$ outside the Hamiltonian, without touching the $\psi^1_{-1}$ inside the Hamiltonian (but counting the $-$ sign that comes when going through it).

The normal ordering is known to remove 'infinite quantities' from the expression of the fields. These are actually sums of anticommutators $\{\eta^1_k,\pi^1_k\}=i$. While no expectation value is taken, one may equally consider that $\{\eta^1_k,\pi^1_k\}=i\cdot 1_{|k|}$ where $1_k=1_{-k}$ is a bosonic variable that commutes with everything, and that could be treated on the same footing as the fermionic variables $\eta$'s. This way these 'infinite quantities' are of the form $\sum_{k>0}1_k$ which are regular elements of the algebra we are using. The only point is then to define a vacuum expectation value for this element of the algebra. Let us define thus
\begin{equation}
\xi_0=\sum_{m>0}1_m,\quad \xi_{-1}=\sum_{m>0}\frac{1_m}{m}\,.
\end{equation}
The bosonic charges are not altered by the normal order:
\begin{equation}
\begin{aligned}
J_z&=\normord{J_z}\,,\qquad J_+=\normord{J_+}\,,\qquad J_-=\normord{J_-}\,.
\end{aligned}
\end{equation}
However the fermionic charges change. With an implicit sum over $k+l+m=0$ (that is explicitly written for $m$ in case of constraints), we have
\begin{equation}
\begin{aligned}
F_1&=\pi^2_{0}+g\pi^2_{-k}\eta^1_l\eta^2_m\\
&=\pi^2_0+g\pi^2_{-k}\eta^1_{-k}\eta^2_0+\frac{i g}{\sqrt{2\kappa}}\pi^2_{-k}\eta^1_l\left(\frac{\psi^2_m}{m}-\frac{\bar{\psi}^2_{-m}}{m} \right)\\
&=\normord{F_1}-ig\eta^1_0+\frac{i g}{\sqrt{2\kappa}}\left(\sum_{m<0}[\pi^2_{-k}\eta^1_l,\tfrac{\psi^2_m}{m}] -\sum_{m>0}[\pi^2_{-k}\eta^1_l,\tfrac{\bar{\psi}^2_{-m}}{m}]\right)\\
&=\normord{F_1}-ig\eta^1_0-\frac{2ig}{\sqrt{2\kappa}}\sum_{m>0}\left(\frac{i}{m\sqrt{2\kappa}}\pi^2_0+\sqrt{\kappa/2}\eta^1_0\right)1_m\\
&=\normord{F_1}-ig\eta^1_0-ig\eta^1_0 \xi_0+g\kappa^{-1}\pi^2_0 \xi_{-1}\,.
\end{aligned}
\end{equation}
To go from the second line to the third line, we noticed that $\pi^2_{-k}\eta^1_l$ only involves $\psi^1,\bar{\psi}^1$ that anticommute, so that the only 'ill-ordered' case that can occur is when a creation operator $\psi^2_{m},\bar{\psi}^2_m$ with $m<0$ is at the right. Similarly:
\begin{equation}
\begin{aligned}
G_1=\normord{G_1}+ig\eta^2_0+ig\eta^2_0 \xi_0+g\kappa^{-1}\pi^1_0 \xi_{-1}\,.
\end{aligned}
\end{equation}
As for the potential, it reads with an implicit summation over $k+l+m+n=0$
\begin{equation}
\begin{aligned}
V&=-\kappa^{-1}\eta^2_k\eta^1_l(\psi^2_m\bar{\psi}^1_{-n}+\bar{\psi}^2_{-m}\psi^1_n)\\
&=\normord{V}-\kappa^{-1}\sum_{\substack{m< 0\\n\leq 0}}[\eta^2_k\eta^1_l,\psi^2_m]\bar{\psi}^1_{-n}+\kappa^{-1}\sum_{\substack{m\geq 0\\n> 0}}[\eta^2_k\eta^1_l,\bar{\psi}^1_{-n}]\psi^2_{m}-\kappa^{-1}\sum_{\substack{m< 0\\n> 0}}\left([\eta^2_k\eta^1_l,\psi^2_m]\bar{\psi}^1_{-n}+\psi^2_m[\eta^2_k\eta^1_l,\bar{\psi}^1_{-n}]\right)\\
&+\kappa^{-1}\sum_{\substack{n< 0\\m\leq 0}}[\eta^2_k\eta^1_l,\psi^1_n]\bar{\psi}^2_{-m}-\kappa^{-1}\sum_{\substack{n\geq 0\\m> 0}}[\eta^2_k\eta^1_l,\bar{\psi}^2_{-m}]\psi^1_{n}+\kappa^{-1}\sum_{\substack{n< 0\\m> 0}}\left([\eta^2_k\eta^1_l,\psi^1_n]\bar{\psi}^2_{-m}+\psi^1_n[\eta^2_k\eta^1_l,\bar{\psi}^2_{-m}]\right)\\
&=\normord{V}-\frac{i}{\kappa\sqrt{2\kappa}}\sum_k\sum_{m>0}(\eta^2_k\bar{\psi}^1_{k}-\eta^1_{-k}\psi^2_k-\eta^1_k\bar{\psi}^2_{k}+\eta^2_{-k}\psi^1_k)1_m \\
&=\normord{V}+i\kappa^{-2}\sum_k(\eta^2_k\pi^2_k+\eta^1_k\pi^1_k)\xi_0\,.
\end{aligned}
\end{equation}
Here, to get to the second line we had to treat separately the 6 different cases where at least one of the fields in $(\psi^2_m\bar{\psi}^1_{-n}+\bar{\psi}^2_{-m}\psi^1_n)$ is a creation operator. The $\eta^2_k\eta^1_l$ part is then already normally ordered, since the only possible ill-ordered case is when $k=-l$, but then the ordering of the $\psi$ part inside $\eta^2_k\eta^1_l$ for $k>0$ cancels out with the ordering of the $\bar{\psi}$ part for $k<0$.

The whole Hamiltonian is
\begin{equation}
\label{eq:hami}
\begin{aligned}
H=&\sum_{k<0}(\psi^2_k\psi^1_{-k}-\psi^1_{k}\psi^2_{-k}+\bar{\psi}^2_k\bar{\psi}^1_{-k}-\bar{\psi}^1_{k}\bar{\psi}^2_{-k})+2\psi^2_0\psi^1_0-\frac{c}{12}-\kappa g \normord{V}\\
&-i\kappa^{-1}g\sum_k (\eta^2_k\pi^2_k+\eta^1_k\pi^1_k)\xi_0\,,
\end{aligned}
\end{equation}
with the central charge $c=-2$ obtained with a usual zeta regularization, namely assigning the value $\zeta(-1)=-1/12$ to the formal sum ``$\sum_{k>0}k$'' by analytically continuing the zeta function $\zeta(x)=\sum_{n>0}n^{-x}$ to the negative axis. The normally ordered Hamiltonian corresponds to the first line. Although the total Hamiltonian commutes with the $osp(1|2)$ charges, this is not the case anymore when normal order is imposed. One thus has to work with the total Hamiltonian, taking into account the $\xi$'s for which a prescription of expectation value has to be given. Note that any (complex) value of the $\xi$'s preserves the commutation relations.\\

\subsubsection{Building the states \label{sec:building}}
Let us define the states on which we will compute expectation values. We define the conformal variables $z=e^{-ix+\tau}$ and $\bar{z}=e^{ix+\tau}$
with $\tau=it$. A state $|\phi\rangle$ for a field $\phi(x,t)$ is defined by
\begin{equation}
|\phi\rangle=\phi(z=0)\rvac\,.
\end{equation}
For example one has
\begin{equation}
|\eta^1\rangle=\eta^1_0\rvac\,.
\end{equation}
The derivatives with respect to $z$ such as $|\partial_z\eta^1\rangle$ deserve some comments. Because of the definition of the modes $\psi$, we still have at all time
\begin{equation}
\eta^1(x,t)=\sum_k\eta_k^1(t)e^{ikx}=\eta^1_0+i\sum_{k\neq 0}\frac{\psi^1_k(t)-\bar{\psi}^1_{-k}(t)}{k\sqrt{2\kappa}}e^{ikx}\,.
\end{equation}
Without interaction we simply have $\psi^1_k(\tau)=\psi^1_k(0)e^{-k\tau}$. But with the interaction the time evolution of the $\psi$'s is not trivial anymore, and involves a term of order $g$ with a product of three $\psi$'s. The expression of $\eta$ in terms of $z,\bar{z}$, and even more for $\partial_z\eta$, is not valid anymore. In particular one cannot build the states as usual:
\begin{equation}
|\partial_z\eta^1 \rangle\neq  -i\psi^1_{-1}\rvac\,,
\end{equation}
contrary to the case $g=0$. The left-hand side now involves a term of order $g$. Precisely, using:
\begin{equation}
\partial_z=-\frac{1}{2i}e^{ix-\tau}\partial_x+\frac{1}{2}e^{ix-\tau}\partial_\tau\,,
\end{equation}
we get
\begin{equation}
\partial_z\eta^1(x,\tau)=\frac{-i}{\sqrt{2\kappa}}\sum_{k}\left(\psi^1_k(\tau)+ g\kappa \sqrt{\frac{\kappa}{2}}\frac{dV}{d\pi^1_k}(\tau)\right)e^{ix(k+1)}e^{-\tau}\,.
\end{equation}
Taking $z=0$ selects $k=-1$ because of the factor $e^{ix(k+1)}$, hence
\begin{equation}
| \partial_z\eta^1\rangle=\frac{-i}{\sqrt{2\kappa}}\psi^1_{-1}\rvac-i\frac{\kappa g}{2}\frac{dV}{d\pi^1_{-1}}\rvac\,,
\end{equation}
that is
\begin{equation}
| \partial_z\eta^1\rangle=-i(1+\frac{g}{2}\eta^2_0\eta^1_0)\psi^1_{-1}\rvac\,.
\end{equation}
Remark moreover that one has indeed $F_1(1+\frac{g}{2}\eta^2_0\eta^1_0)\psi^1_{-1}\rvac=0$ and $J_+(1+\frac{g}{2}\eta^2_0\eta^1_0)\psi^1_{-1}\rvac=0$ and this state is indeed highest-weight.

However, if we look at terms involving a $\eta$ times a derivative, since we have
\begin{equation}
\eta^1_k\frac{dV}{d\pi^1_{-k}}=0\,,
\end{equation}
we do have
\begin{equation}
|\eta^1 \partial_z\eta^1\rangle=-i\eta^1_0\psi^1_{-1}\rvac\,,
\end{equation}
without any corrections in $g$.\\

For terms such as $|\eta^1\partial_z\eta^1...\partial_z^m\eta^1\rangle$ the corresponding calculations happen to be not as easily tractable, but the state $\eta^1_0\psi^1_{-1}...\psi^1_{-m}$ is indeed annihilated by $F_1$ (and $J_+$). Similarly, $|\partial_z\eta^1...\partial_z^m\eta^1\rangle$ is not as easily computed, but $(1+m\frac{g}{2}\eta^2_0\eta^1_0)\psi^1_{-1}...\psi^1_{-m}$ is annihilated by $F_1$ and $J_+$.

\subsubsection{Regularization}
Any complex values for the $\xi$'s give a Hamiltonian that commutes with the charges (that also depend on $\xi$), and thus that has the $osp(1|2)$ symmetry and the classical non-linear sigma model as classical limit. However, since some of them appear in expectation values, one has to fix their value with an exterior argument. This is the only additional information that we need to quantize the model. Here we choose to fix the zero mode of the Hamiltonian. It reads
\begin{equation}
H_0=-\kappa^{-1}(1-g\eta^1_0\eta^2_0)\pi^1_0\pi^2_0-\kappa^{-1}ig(\eta^2_0\pi^2_0+\eta^1_0\pi^1_0)\xi_0\,.
\end{equation}
This zero mode should give the Laplace-Casimir operator of the algebra, see \cite{berezintolstoy}, thus be proportional to $(1-g\eta^1_0\eta^2_0)\pi^1_0\pi^2_0-ig(\eta^2_0\pi^2_0+\eta^1_0\pi^1_0)$. Hence we impose the value
\begin{equation}
\expvalue{0}{\xi_0}=-1\,.
\end{equation}

As for the value of $\expvalue{0}{\xi_{-1}}$, it does not enter the Hamiltonian nor the construction of the states, and thus has no influence on any expectation values.\\

Equivalently we could impose that the fermionic charges (and the bosonic ones) are not modified by the normal order. This constrains $\xi_0$ to take the value $-1$ and $\xi_{-1}$ the value $0$.

\subsubsection{Correction to the energy levels \label{sec:osp12corr}}
To evaluate the corrections at order $g$ to the energy levels, one has to compute the matrix elements $\langle \phi_1|H|\phi_2\rangle$ where $\phi_1$ and $\phi_2$ are eigenstates of the unperturbed $H$ with the same energy. In the following we will compute the action of $H$ on some state $|\phi\rangle$. Consider one term in the Hamiltonian \eqref{eq:hami} and denote $n$ and $\bar{n}$ the sum of the indices of the $\psi$'s and $\bar{\psi}'s$ that compose this term. In general we have $n-\bar{n}=0$, but not necessarily $n=\bar{n}=0$ separately. Thus a state $|\phi_1\rangle$ with a certain value of $n-\bar{n}$ is mapped by $H$ onto another state with the same value of $n-\bar{n}$. Since $|\phi_1\rangle$ and $|\phi_2\rangle$ have the same energy, they must have the same value of $n+\bar{n}$, hence the same value of $n$ and $\bar{n}$ separately to have a non-zero matrix elements by $H$. This way one can consider only the 'conservative' part of $H$, i.e. its terms with $n=\bar{n}=0$. It corresponds to summing the indices of the $\psi$'s to zero, and to summing the indices of the $\bar{\psi}$'s to zero separately as well. Note that in the term $\eta^2_k\pi^2_k+\eta^1_k\pi^1_k$, the only 'conservative' term is the zero mode $\eta^2_0\pi^2_0+\eta^1_0\pi^1_0$.
\begin{itemize}
\item $|0\rangle$. We have
\begin{equation}
H|0\rangle=0\,,
\end{equation}
that will be the reference state (in all finite sizes $L$) for our computations.

\item $|\eta^1\rangle\propto \eta^1_0\rvac$. We have
\begin{equation}
\normord{V}\eta^1_0\rvac=0\,,
\end{equation}
thus
\begin{equation}
\begin{aligned}
H\eta^1_0\rvac=-\kappa^{-1}g\eta^1_0\rvac\,,
\end{aligned}
\end{equation}
and a correction $-\kappa^{-1}g=-1/3\log L$. Note that it is below the ground state of the zero magnetization sector.

The other states in the multiplet are $(1-g\eta^2_0\eta^1_0)\rvac$ and $\eta^2_0\rvac$, i.e. the states obtained from $\eta^1_0\rvac$ by applying the lowering operators $J_-$ and $F_-$. They indeed have the same correction:
\begin{equation}
\begin{aligned}
H(1-g\eta^2_0\eta^1_0)\rvac&=0-gH\eta^2_0\eta^1_0\rvac\\
&=-\kappa^{-1}g\rvac+O(g^2)\\
&=-\kappa^{-1}g(1-g\eta^2_0\eta^1_0)\rvac+O(g^2)\,,
\end{aligned}
\end{equation}
and
\begin{equation}
H\eta^2_0\rvac=-\kappa^{-1}g\eta^2_0\rvac\,,
\end{equation}
hence
\begin{equation}
\frac{L^2 \Delta e_L}{2\pi v_F}=-\kappa^{-1}g\,.
\end{equation}


\item  $|\eta^1\partial\eta^1....\partial^m\bar{\partial}\eta^1....\bar{\partial}^{\bar{m}}\eta^1\rangle\propto \eta^1_0\psi^1_{-1}...\psi^1_{-m}\bar{\psi}^1_{-1}...\bar{\psi}^1_{-\bar{m}}\rvac$. We have
\begin{equation}
\begin{aligned}
&\normord{V} \eta^1_0\psi^1_{-1}...\psi^1_{-m}\bar{\psi}^1_{-1}...\bar{\psi}^1_{-\bar{m}}\rvac\\
&=\frac{1}{2\kappa^2}\left( \sum_{n=1}^{\bar{m}}\normord{\frac{-\bar{\psi}^2_n}{-n}(-i\eta^1_0)\psi^2_0\bar{\psi}^1_{-n}} +  \sum_{n=1}^{\bar{m}}\sum_{k=1}^m\normord{\frac{-\bar{\psi}^2_n}{-n}\frac{\psi^1_{-k}}{-k}\psi^2_k\bar{\psi}^1_{-n}}\right)\eta^1_0\psi^1_{-1}...\psi^1_{-m}\bar{\psi}^1_{-1}...\bar{\psi}^1_{-\bar{m}}\rvac\\
&+\frac{1}{2\kappa^2}\left(\sum_{n=1}^m \normord{\frac{\psi^2_n}{n}(-i\eta^1_0)\bar{\psi}^2_0\psi^1_{-n}}+\sum_{n=1}^m\sum_{k=1}^{\bar{m}} \normord{\frac{\psi^2_n}{n}\frac{-\bar{\psi}^1_{-k}}{k}\bar{\psi}^2_k\psi^1_{-n}}\right)\eta^1_0\psi^1_{-1}...\psi^1_{-m}\bar{\psi}^1_{-1}...\bar{\psi}^1_{-\bar{m}}\rvac\\
&=\frac{1}{2\kappa^2}\left(-\sum_{n=1}^{\bar{m}}\frac{-n}{n}-\sum_{n=1}^{\bar{m}}\sum_{k=1}^m\frac{-n}{n}\frac{k}{k}-\sum_{n=1}^{m}\frac{-n}{n}-\sum_{n=1}^m\sum_{k=1}^{\bar{m}}\frac{-n}{n}\frac{k}{k}\right)\eta^1_0\psi^1_{-1}...\psi^1_{-m}\bar{\psi}^1_{-1}...\bar{\psi}^1_{-\bar{m}}\rvac\\
&=\frac{1}{2\kappa^2}(\bar{m}(m+1)+m(\bar{m}+1))\eta^1_0\psi^1_{-1}...\psi^1_{-m}\bar{\psi}^1_{-1}...\bar{\psi}^1_{-m}\rvac\,.
\end{aligned}
\end{equation}
Indeed in $\normord{V}$, every $\psi^2_k$ with $k>m$ or $\bar{\psi}^2_k$ with $k>\bar{m}$ anticommutes with all the fields in the state and annihilates $\rvac$, hence the restriction over the summations. 

Thus
\begin{equation}
\label{correction1}
H\eta^1_0\psi^1_{-1}...\psi^1_{-m}\bar{\psi}^1_{-1}...\bar{\psi}^1_{-\bar{m}}\rvac=(m(m+1)-(2m\bar{m}+m+\bar{m}+2)g)\eta^1_0\psi^1_{-1}...\psi^1_{-m}\bar{\psi}^1_{-1}...\bar{\psi}^1_{-\bar{m}}\rvac\,,
\end{equation}
hence
\begin{equation}
\frac{L^2 \Delta e_L}{2\pi v_F}=\tfrac{1}{2}m(m+1)+\tfrac{1}{2}\bar{m}(\bar{m}+1)-(m\bar{m}+\tfrac{m+\bar{m}}{2}+1)\kappa^{-1}g\,.
\end{equation}
Note that the correction is not the sum of left and right contributions, but involves also a cross-term $m\bar{m}$. 

For symmetric states $m=\bar{m}$ we find
\begin{equation}
\label{eq:osp12gapsym}
\boxed{\frac{L^2 \Delta e_L}{2\pi v_F}=m(m+1)-(m^2+m+1)\kappa^{-1}g}\,,
\end{equation}
while when  $\bar{m}=0$,  the logarithmic correction becomes linear
\begin{equation}
\label{eq:correction2m}
\frac{L^2 \Delta e_L}{2\pi v_F}=\tfrac{1}{2}m(m+1)-(\tfrac{m}{2}+1)\kappa^{-1}g\,.
\end{equation}
In general, the correction is  thus not simply the sum of left and right contributions.

\item $|\partial\eta^1...\partial^m\eta^1\bar{\partial}\eta^1...\bar{\partial}^m\eta^1\rangle \propto (1+mg\eta^2_0\eta^1_0)\psi^1_{-1}...\psi^1_{-m}\bar{\psi}^1_{-1}...\bar{\psi}^1_{-m}$.  

Here the computation for the part involving the $\psi$'s is similar to the previous case. As for the $\eta^2_0\eta^1_0$ part, only the unperturbed Hamiltonian acts on it. Combining the two parts, one finds
\begin{equation}
\label{eq:osp12zero}
\frac{L^2 \Delta e_L}{2\pi v_F}=m(m+1)-m(m-1)\kappa^{-1}g\,.
\end{equation}

Note that the other states in the multiplet for $m=1$ are $(\eta^1_0\psi^2_{-1}+\psi^1_{-1}\eta^2_0)\rvac$ and $(1+\frac{g}{2}\eta^2_0\eta^1_0)\psi^2_{-1}\rvac$ and give the same correction as expected, hence the importance of the factor $mg\eta^2_0\eta^1_0$ that comes from the discussion in section \ref{sec:building}.


\end{itemize}

\subsection{The spectrum from the spin chain}

In the remainder of this section we provide two alternative means of deriving \eqref{eq:osp12gapsym}, or at least special cases thereof.
The first of these relies on the Bethe-ansatz diagonalization of the spin chain Hamiltonian, and the other on the computation of three-point functions---either
directly, or using a trick reminiscent of Wick's theorem. While both of these methods are of independent relevance, the reader interested mainly in results
for other models may chose to skip directly to section~\ref{sec:osp22}.

\subsubsection{Bethe equations}
We now move on to the corresponding $osp(1|2)$ spin chain.  Like the sigma model which involves as a basic degree of freedom a field in the vector representation of the algebra, the spin chain involves a tensor product of fundamental  representations  of $osp(1|2)$. In contrast with the other superalgebras we will encounter in this paper, $osp(1|2)$ has a simple representation theory. Its action on the spin chain  is fully reducible, and the Hilbert space decomposes onto a direct sum of ``spin $j$'' representations, with dimension $4j+1$. Here $j$ is the eigenvalue of the $J_z$ generator on the highest weight state, and $j=1/2$ corresponds to the fundamental. 

The spectrum of the Hamiltonian is described by one family of roots $\lambda_i$ satisfying the Bethe equations \cite{galleasmartins,martinsosp12}
\begin{equation}
\left( \frac{\lambda_i+i/2}{\lambda_i-i/2}\right)^L=\prod_{j\neq i}\frac{\lambda_i-\lambda_j+i}{\lambda_i-\lambda_j-i}\cdot \frac{\lambda_i-\lambda_j-i/2}{\lambda_i-\lambda_j+i/2}\,.
\end{equation}
An eigenvalue of the Hamiltonian for one set of solutions $\lambda_1,...,\lambda_M$ to these equations is then
\begin{equation}
e_L=-\frac{1}{L}\sum_{i=1}^M \frac{1}{\lambda_i^2+1/4}\,.
\end{equation}
The spin $j$ (ie, the eigenvalue of $J_z$) corresponding to a solution with $M$ roots is linked to $M$ through
\begin{equation}
\label{eq:relcharge}
M=L-2j\,.
\end{equation}
Moreover the $osp(1|2)$ Bethe states are highest-weight states.

This kind of relation, that links the number of Bethe roots to the value of the charges of a state, can be simply deduced from a direct diagonalisation of the Hamiltonian in small sizes. But in some cases it can be obtained analytically from commutation relation between the total charges and the monodromy matrix components.


\subsubsection{Bethe root structure}
The first task when studying a spin chain with the Bethe ansatz is to find the structure of the roots that correspond to the energies (at least the low-lying ones). There is no generic way of determining this structure from the Bethe equations, implying that a numerical study is an inevitable step. The two options are either to use the  McCoy method \cite{mccoy,mccoy2}, or to proceed by a trial-and-error approach.

We observe that on the lattice in size $L$, the field $\eta^1\partial\eta^1...\partial^m\eta^1\bar{\partial}\eta^1...\bar{\partial}^{\bar{m}}\eta^1$ is obtained with $L-1-m-\bar{m}$ real Bethe roots, with $m$ positive vacancies and $\bar{m}$ negative vacancies. The field $\partial\eta^1...\partial^m\eta^1\bar{\partial}\eta^1...\bar{\partial}^{m}\eta^1$ is obtained with $L-2m$ Bethe roots, among which $L-2m-2$ are real and symmetrically distributed, and $2$ form an exact $2$-string at $\pm i/2$, i.e. a pair of complex conjugate Bethe roots whose values are exactly $\pm i/2$. The field $\partial\eta^1...\partial^m\eta^1\bar{\partial}\eta^1...\bar{\partial}^{\bar{m}}\eta^1$ when $m\neq \bar{m}$ is obtained with $L-m-\bar{m}$ Bethe roots, among which $L-m-\bar{m}-2$ are real with $m$ positive vacancies and $\bar{m}$ negative vacancies, and $2$ form an approximate $2$-string at $\pm i/2$ with large real part, on the side where there are the most vacancies. See Figure \ref{fig:rootosp12} for a plot of some root structures.

 \begin{figure}[H]
 \begin{center}
\includegraphics[scale=0.5]{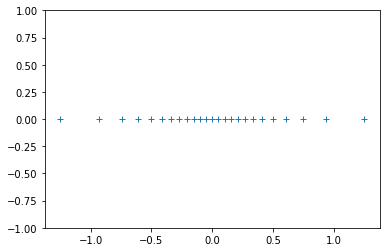} 
\includegraphics[scale=0.5]{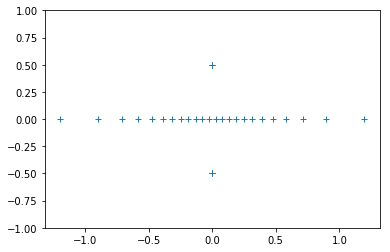}
\end{center}
\caption{Bethe roots in the complex plane for the lowest state of magnetization $1$ (left) and $0$ (right), for $L=26$.}
\label{fig:rootosp12}
\end{figure}

The results for the gaps $\Delta e_L$ of the ground state of the sectors of magnetization $j$ reads then
\begin{equation}
\label{osp12leadingcorr}
\frac{L^2 \Delta e_L}{2\pi v_F} =
\begin{cases}
 j^2-\frac{1}{4}-\left(j^2+\frac{3}{4}\right)\kappa^{-1}g\,, & \quad \text{if $j$ is half-integer} \,, \\
 j(j+1)-j(j-1)\kappa^{-1}g\,, & \quad \text{if $j$ is integer} \,,
\end{cases}
\end{equation}
in agreement  with (\ref{eq:osp12gapsym}) after identifying $j\equiv m+{1\over 2}$ for half-integer spins, and with \eqref{eq:osp12zero} where $j\equiv m$ for integer spins, while the value of the Casimir on these states is
\begin{equation}
\label{eq:casimirosp12}
{\cal C}=j(2j+1)\,.
\end{equation}
See below for a discussion of the symmetries at finite and vanishing $g$. 

Note that all these corrections previously derived from field theory can be computed analytically within the Bethe ansatz; see eq.~\eqref{eq:pert} in Appendix \ref{sec:appc}.

\subsubsection{Numerical results\label{sec:numericsosp12}}
We present here the numerical verification of the logarithmic corrections, carried out with the Bethe ansatz. Because of the logarithms, large sizes are needed to get a good precision. The general idea is to use a Newton method to solve the Bethe equations, using the solution at size $L$ to build an initial guess at size $L+2$ close enough to the solution to make the Newton method converge. Then a fit as a quotient of two polynomials in $(\log L)^{-1}$ is performed. Precisely, we used the function
\begin{equation}
f_n(L)=\frac{a_0+a_1 (\log L)^{-1}+...+a_{n-1}(\log L)^{-n+1}}{1+b_1(\log L)^{-1}+...+b_n(\log L)^{-n}}\,,
\end{equation}
and fitted the parameters $a_0,...,a_{n-1},b_1,...,b_n$ for a value of $n$ depending on the state.

From the Bethe ansatz one computes $Z^{m,\bar{m}}_L=(\frac{L^2}{2\pi v_F}(e_L-e_L^0)-(h+\bar{h}))\log L$, where $e_L^0$ is the energy of the ground state and $e_L$ the energy of the state of study (here, the one with $m$ positive vacancies and $\bar{m}$ negative vacancies; a $1/2$ vacancy on both sides counts for an odd number of total vacancies), and looks for its limit value, see Figure \ref{fig:osp12}. 

 \begin{figure}
 \begin{center}
\includegraphics[scale=0.5]{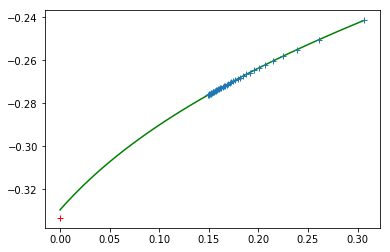} 
\includegraphics[scale=0.5]{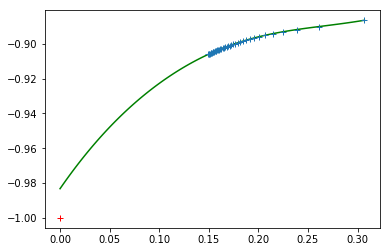}
\includegraphics[scale=0.5]{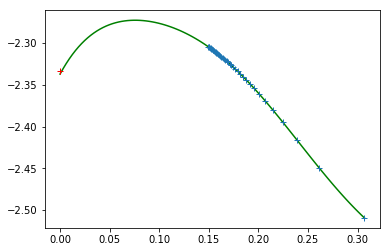} 
\includegraphics[scale=0.5]{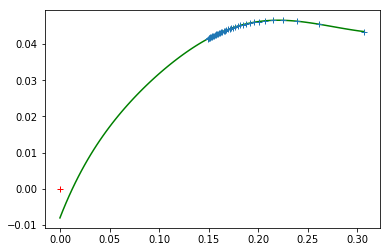} 
\includegraphics[scale=0.5]{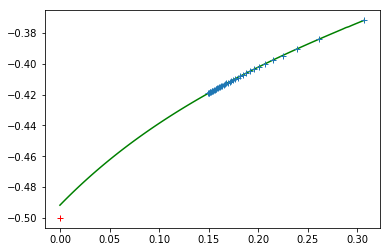} 
\includegraphics[scale=0.5]{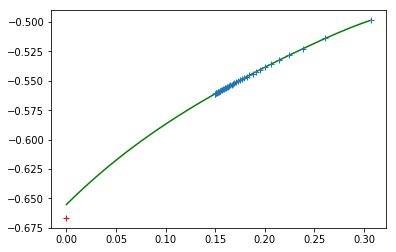} 
\end{center}
\caption{In reading direction: plots of $Z^{1/2,1/2}_L$, $Z^{1/2+1,1/2+1}_L$, $Z^{1/2+2,1/2+2}_L$, $Z^{1,1}_L$, $Z^{1/2+1,1/2}_L$, $Z^{1/2+2,1/2}_L$ as a function of $1/\log L$, together with their extrapolated curves $f_6$, $f_7$, $f_6$, $f_5$, $f_8$, $f_8$. The theoretical results are, respectively, $-1/3$, $-1$, $-7/3$, $0$, $-1/2$, $-2/3$.}
\label{fig:osp12}
\end{figure}

\subsection{Relation with $3$-point functions \label{sec:3ptf}}
We discuss in this section how the logarithmic corrections can be related to the $3$-point functions in the plane. Our calculation parallels  the work by Cardy for quasi-primary fields  \cite{cardy1986}, but applies here to the logarithmic case.

\subsubsection{Quasi-primary fields}
It is known that for quasi-primary fields the logarithmic corrections to the energy levels are linked to the structure constant between the fields of the level and the marginal operator that perturbs the Hamiltonian \cite{cardy1986}. We briefly remind here the reader of this relation. Assume that a Hamiltonian $H_0$ is perturbed by a potential $g V$: 
\begin{equation}
H=H_0-\kappa\frac{g}{2\pi} \int_0^{2\pi} \mathcal{V}(x,t)dx\,.
\end{equation}
Define $\delta \hat{e}$ by
\begin{equation}
\kappa^{-1}g\delta \hat{e}=\frac{L^2(e_L^{(g)}-e_L^{(0)})}{2\pi v_F}\,,
\end{equation}
where $e_L^{(g)}$ is the energy level of a given state with the perturbation $g$, and $e_L^{(0)}$ the energy level of the same state without the perturbation.
The corrections to the energy $\delta \hat{e}$  of a state $|\phi\rangle$ is
\begin{equation}
\delta \hat{e}=-\kappa^2\frac{1}{2\pi}\int_0^{2\pi} \langle\phi | \mathcal{V}(x,t) |\phi\rangle dx\,.
\end{equation}
And we have for a field $\phi$ of conformal weights $h=\bar{h}$
\begin{equation}
\langle\phi | \mathcal{V}(x,t) |\phi\rangle= \lim_{z,w\to 0} \lvac z^{-2h}\bar{z}^{-2h}\phi(1/z) \mathcal{V}(x,t) \phi(w)\rvac\,.
\end{equation}
If $\phi$ and $\mathcal{V}$ are quasi-primary then their three-point function is constrained to be exactly
\begin{equation}
\lvac \phi(z) \mathcal{V}(y) \phi(w)\rvac=\frac{C}{|z-y|^4 |y-w|^4 |z-w|^{4h-4}}\,,
\end{equation}
where $C$ is the structure constant. This gives
\begin{equation}
\delta \hat{e}=-\kappa^2 C\,.
\end{equation}

\subsubsection{Fields with logarithms}
The previous argument does not apply to the logarithmic cases discussed above. Indeed, $\eta^1$ and $\eta^2$ are not quasi-primary: their correlation involves $\log$ that is not a scale covariant function. It thus needs a more detailed study. The purpose of the subsequent sections is to study the link between the correction to the energy levels for the states $|\eta^1\partial\eta^1...\partial^m\eta^1\rangle$ and the three-point function between them and the perturbative potential.\\

Denote $\phi^1_{m}(z)$ the field $\eta^1(z)\partial\eta^1(z)...\partial^m\eta^1(z)$ and $\phi^2_{m}(z)$ the field $\partial^m\eta^2(z)...\partial\eta^2(z)\eta^2(z)$. They have scaling dimensions $m(m+1)/2$. The corresponding state $|\phi^1_m\rangle$ is given by the constant coefficient in $\phi^1_m(z)$, and is $\propto \eta^1_0\psi^1_{-1}...\psi^1_{-m}$. The state $\langle \phi^2_m|$ is defined as the 
conjugate of $|\phi^1_m\rangle$, thus $\propto \psi^2_{m}...\psi^2_{1}\pi^1_0$. It is given by the coefficient $\log |z|^2 z^{-m(m+1)}$ in $\phi^2_m(z)$, denoted $\text{coeff }_{\log |z|^2 z^{-m(m+1)}} (\lvac \phi^2_m(z))$ (one could give an integral formula for this, but it is unnecessary). Note that in absence of $\log$, this matches the usual definitions. One can express these as
\begin{equation}
|\phi^1_m\rangle=\phi^1_m(0,0)\rvac,\quad \quad \langle \phi^2_m|= \text{coeff }_{\log |z|^2 z^{-m(m+1)}}( \lvac \phi^2_m(z))\,.
\end{equation}
The Hamiltonian is perturbed by $-\kappa g/(2\pi)\int \mathcal{V}(x,t=0)dx$. Since the perturbation $\kappa^{-1}g\delta \hat{e}$ to the energy level of state $\phi_m$ is given by

\begin{equation}
\delta \hat{e}=-\frac{\kappa^2}{\langle \phi^2_m|\phi^1_m\rangle} \langle \phi^2_m|\frac{1}{2\pi} \int_0^{2\pi} \mathcal{V}(x,t=0)dx |\phi^1_m\rangle\,,
\end{equation}
one sees that it can be expressed in terms of the 3-point function $G_{\phi^1_m}(z_2,z,z_1)$ where
\begin{equation}
G_{X}(z_2,z,z_1)=\lvac X^\dagger(z_2) \mathcal{V}(z)X(z_1)\rvac\,,
\label{3ptfctGX}
\end{equation}
for a field $X(z)$. Precisely:
\begin{equation}
\label{delta}
\delta \hat{e}=-\frac{\kappa^2}{\langle \phi^2_m|\phi^1_m\rangle} \text{coeff }_{\log|z_2|^2 z_2^{-m(m+1)}\times z^0}(G(z_2,z,0))\,.
\end{equation}

\subsubsection{An explicit calculation of the 3-point function $\langle \eta^2\mathcal{V}\eta^1\rangle$}

The fields are assumed to be radially ordered $|z_1|<|z|<|z_2|$.

Let us first compute explicitly the 3-point function $G_{\eta}(z_2,z,z_1)$:
\begin{equation}
\label{correlator}
G_{\eta}(z_2,z,z_1)=\lvac \eta^2(z_2) \mathcal{V}(z) \eta^1(z_1)\rvac\,,
\end{equation}
at points $z=e^{-ix+\tau}$. The potential $\mathcal{V}(x,t)$ is
\begin{equation}
\mathcal{V}(x,t)=\normord{\mathcal{V}(x,t)}+i\kappa^{-2}\xi_0\sum_{p,k}(\eta^2_k(t)\pi^2_{k-p}(t)+\eta^1_k(t)\pi^1_{k-p}(t))e^{ipx}\,.
\end{equation}
Within the correlator \eqref{correlator} the part with the normally ordered potential $\normord{\mathcal{V}(x,t)}$ is zero since there are four fields to contract with only two fields at our disposal. Thus we have
\begin{equation}
G_{\eta}(z_2,z,z_1)=-\lvac \eta^2(z_2)\left( i\kappa^{-2}\xi_0\sum_{p,k}(\eta^2_k(t)\pi^2_{k-p}(t)+\eta^1_k(t)\pi^1_{k-p}(t))e^{ipx} \right) \eta^1(z_1)\rvac\,.
\end{equation}
We can now use \eqref{derivatives} to express each of the $\eta^{1,2}_k(t)$ and $\pi^{1,2}_k(t)$ in terms of $z=e^{-ix+\tau}=e^{-ix+it}$. We find
\begin{equation}
\begin{aligned}
G_{\eta}(z_2,z,z_1)=&-\kappa^{-2}i\lvac \left(\eta^2_0 -i\frac{\pi^1_0}{2\kappa} \log |z_2|^2 +i\sum_{n\neq 0}\frac{\psi^2_nz_2^{-n}-\bar{\psi}^2_{-n}\bar{z}_2^{n}}{n\sqrt{2\kappa}} \right)\\
&\sum_{k,p}\frac{i}{2k}(\psi^1_k \psi^2_{-k+p}z^{-p}+\psi^1_k\bar{\psi}^2_{k-p}z^{-k}\bar{z}^{-k+p}-\bar{\psi}^1_{-k}\psi^2_{-k+p}\bar{z}^kz^{k-p}-\bar{\psi}^1_{-k}\bar{\psi}^2_{k-p}\bar{z}^{p}\\
&-\psi^2_k\psi^1_{-k+p}z^{-p}-\psi^2_k\bar{\psi}^1_{k-p}z^{-k}\bar{z}^{-k+p}+\bar{\psi}^2_{-k}\psi^1_{-k+p}z^{k-p}\bar{z}^{k}+\bar{\psi}^2_{-k}\bar{\psi}^1_{k-p}\bar{z}^{p})\\
&\left(\eta^1_0 +i\frac{\pi^2_0}{2\kappa} \log |z_1|^2 +i\sum_{m\neq 0}\frac{\psi^1_mz_1^{-m}-\bar{\psi}^1_{-m}\bar{z_1}^{m}}{m\sqrt{2\kappa}} \right)\rvac\,,
\end{aligned}
\end{equation}
where we use the shortcuts $i\psi^{1,2}_0/(0\sqrt{2\kappa})=\eta^{1,2}_0\pm i (2\kappa)^{-1}\pi^{2,1}_0\log |z|^2$ and $i\bar{\psi}^{1,2}_0/0=0$ to simplify the notations. One has
\begin{equation}
\begin{aligned}
&G_{\eta}(z_2,z,z_1)=\\
&-\tfrac{\kappa^{-3}}{2}\log |z_2|^2\lvac \pi^1_0 \eta^1_0 \psi^2_0\eta^1_0\rvac-\tfrac{\kappa^{-3}}{4}i\log |z_2|^2 \sum_{p>0}\lvac \pi^1_0\eta^1_0\psi^2_p\tfrac{\psi^1_{-p}}{-p}\rvac  z^{-p}z_1^p-\tfrac{\kappa^{-3}}{4}i\log |z_2|^2\sum_{p<0}\lvac \pi^1_0\eta^1_0\bar{\psi}^2_{-p}\tfrac{\bar{\psi}^1_p}{p}\rvac \bar{z}^p \bar{z_1}^{-p}\\
&-\tfrac{\kappa^{-3}}{4}\sum_{n\neq 0,m}\lvac \tfrac{\psi^2_n}{n}\tfrac{\psi^1_{-n}}{-n}\psi^2_{-m}\tfrac{\psi^1_m}{m}\rvac z_2^{-n}z^{n+m}z_1^{-m}-\tfrac{\kappa^{-3}}{4}\sum_{n\neq 0,m}\lvac \tfrac{\psi^2_n}{n}\psi^1_{-n}\frac{\psi^2_{-m}}{-m}\tfrac{\psi^1_m}{m}\rvac z_2^{-n}z^{n+m}z_1^{-m}\\
&-\tfrac{\kappa^{-3}}{4}\sum_{n\neq 0,m}\lvac \tfrac{\psi^2_n}{n}\tfrac{\psi^1_{-n}}{-n}\bar{\psi}^2_{m}\tfrac{\bar{\psi}^1_{-m}}{-m}\rvac z_2^{-n}z^{n}\bar{z}^{-m}\bar{z}_1^{m}+\tfrac{\kappa^{-3}}{4}\sum_{n\neq 0,m}\lvac \tfrac{\psi^2_n}{n}\psi^1_{-n}\tfrac{\bar{\psi}^2_{m}}{-m}\tfrac{\bar{\psi}^1_{-m}}{-m}\rvac z_2^{-n}z^{n}\bar{z}^{-m}\bar{z}_1^{m}\\
&+\tfrac{\kappa^{-3}}{4}\sum_{n\neq 0,m}\lvac \tfrac{\bar{\psi}^2_{-n}}{-n}\tfrac{\bar{\psi}^1_{n}}{-n}\psi^2_{-m}\tfrac{\psi^1_{m}}{m}\rvac \bar{z_2}^{n}z^{m}\bar{z}^{-n}z_1^{-m}-\tfrac{\kappa^{-3}}{4}\sum_{n\neq 0,m}\lvac \tfrac{\bar{\psi}^2_{-n}}{-n}\bar{\psi}^1_{n}\tfrac{\psi^2_{-m}}{-m}\tfrac{\psi^1_{m}}{m}\rvac \bar{z_2}^{n}z^{m}\bar{z}^{-n}z_1^{-m}\\
&+\tfrac{\kappa^{-3}}{4}\sum_{n\neq 0,m}\lvac \tfrac{\bar{\psi}^2_{-n}}{-n}\tfrac{\bar{\psi}^1_{n}}{-n}\bar{\psi}^2_{m}\tfrac{\bar{\psi}^1_{-m}}{-m}\rvac \bar{z_2}^{n}\bar{z}^{-n-m}\bar{z}_1^{m}+\tfrac{\kappa^{-3}}{4}\sum_{n\neq 0,m}\lvac \tfrac{\bar{\psi}^2_{-n}}{-n}\bar{\psi}^1_{n}\tfrac{\bar{\psi}^2_{m}}{-m}\tfrac{\bar{\psi}^1_{-m}}{-m}\rvac \bar{z_2}^{n}\bar{z}^{-n-m}\bar{z}_1^{m}\,.
\end{aligned}
\end{equation}
Evaluating each scalar product gives
\begin{equation}
\begin{aligned}
&G_{\eta}(z_2,z,z_1)=\\
&\tfrac{\kappa^{-3}}{2}\log |z_2|^2+\tfrac{\kappa^{-3}}{4}\log |z_2|^2 \sum_{p>0}z^{-p}z_1^p+\tfrac{\kappa^{-3}}{4}\log |z_2|^2 \sum_{p<0}\bar{z}^p\bar{z_1}^{-p}\\
&-\tfrac{\kappa^{-3}}{4}\log |z|^2 \sum_{p>0}z_2^{-p}z^p-\tfrac{\kappa^{-3}}{4}\log |z|^2 \sum_{p<0}\bar{z_2}^p\bar{z}^{-p}\\
&-\tfrac{\kappa^{-3}}{4}\sum_{n>0,m<0}\tfrac{z_2^{-n}z^{n+m}z_1^{-m}}{n}+\tfrac{z_2^{-n}z^{n+m}z_1^{-m}}{m}-\tfrac{\kappa^{-3}}{4}\sum_{n>0,m>0}\tfrac{z_2^{-n}z^n\bar{z}^{-m}\bar{z_1}^m}{n}-\tfrac{z_2^{-n}z^n\bar{z}^{-m}\bar{z_1}^m}{m}\\
&+\tfrac{\kappa^{-3}}{4}\sum_{n<0,m<0}\tfrac{\bar{z}_2^nz^{m}\bar{z}^{-n}z_1^{-m}}{n}-\tfrac{\bar{z}_2^nz^m\bar{z}^{-n}z_1^{-m}}{m}+\tfrac{\kappa^{-3}}{4}\sum_{n<0,m>0}\tfrac{\bar{z}_2^n\bar{z}^{-n-m}\bar{z}_1^m}{n}+\tfrac{\bar{z}_2^n\bar{z}^{-n-m}\bar{z}_1^m}{m}\\
&-\tfrac{\kappa^{-3}}{2}\sum_{n>0}\tfrac{z_2^{-n}z^n}{n}+\tfrac{\kappa^{-3}}{2}\sum_{n<0}\tfrac{\bar{z}_2^n\bar{z}^{-n}}{n}\,,
\end{aligned}
\end{equation}
or in a simpler form
\begin{equation}
G_{\eta}(z_2,z,z_1)=\frac{\kappa^{-3}}{4}\left(\frac{z}{z-z_2}+\frac{\bar{z}}{\bar{z}-\bar{z}_2} \right)\log |z-z_1|^2+\frac{\kappa^{-3}}{4}\left(\frac{z}{z-z_1}+\frac{\bar{z}}{\bar{z}-\bar{z}_1} \right)\log |z-z_2|^2\,.
\end{equation}
Formula \eqref{delta} gives here, with $\langle \phi^1_0|\phi^2_0\rangle=(2\kappa)^{-1}$
\begin{equation}
\delta \hat{e}=-\kappa^2 \times 2\kappa \times \tfrac{\kappa^{-3}}{4} \times 2=-1\,.
\end{equation}
which is indeed the correction computed in \eqref{eq:correction2m}.

\subsubsection{2-point functions}
If we were to compute the 3-point function $\lvac \partial\eta^2\eta^2\mathcal{V}\eta^1\partial\eta^1\rvac$ with the same method as in the previous example, one would have to take into account the term $\normord{\mathcal{V}}$ and the computations would become quite cumbersome. Actually, such a computation can always be recast into a product of 2-point functions, like a Wick's theorem. Indeed, since the anticommutator of two modes is a complex number, to evaluate the 3-point function \eqref{3ptfctGX} one has to contract every mode of the middle field $\mathcal{V}$ with modes of the right and left fields, and then contract the remaining modes between them. The 2-point functions that appear in the result involve the following fields and their derivatives:
\begin{equation}
\begin{aligned}
&\eta^1(z)=\eta^1_0 +i\frac{\pi^2_0}{2\kappa} \log |z|^2 +i\sum_{n\neq 0}\frac{\psi^1_nz^{-n}+\bar{\psi}^1_{n}\bar{z}^{-n}}{n\sqrt{2\kappa}}\,,\qquad \eta^2(z)=\eta^2_0 -i\frac{\pi^1_0}{2\kappa} \log |z|^2 +i\sum_{n\neq 0}\frac{\psi^2_nz^{-n}+\bar{\psi}^2_{n}\bar{z}^{-n}}{n\sqrt{2\kappa}}\\
&\partial_x\eta^1(z)=-\frac{1}{\sqrt{2\kappa}}\sum_{n\neq 0}\psi^1_nz^{-n}-\bar{\psi}^1_{n}\bar{z}^{-n}\,,\qquad \partial_x\eta^2(z)=-\frac{1}{\sqrt{2\kappa}}\sum_{n\neq 0}\psi^2_nz^{-n}-\bar{\psi}^2_{n}\bar{z}^{-n}\\
&\pi^1(z)=\pi^1_0 +\sqrt{\frac{\kappa}{2}}\sum_{n\neq 0}\psi^2_nz^{-n}+\bar{\psi}^2_n\bar{z}^{-n}\,,\qquad \pi^2(z)=\pi^2_0 -\sqrt{\frac{\kappa}{2}}\sum_{n\neq 0}\psi^1_nz^{-n}+\bar{\psi}^1_n\bar{z}^{-n}\,.
\end{aligned}
\end{equation}
The 2-point function between these fields are known or computed without problems \cite{kausch}. For example
\begin{equation}
\begin{aligned}
2\kappa\lvac \eta^2(z)\eta^1(w)\rvac&=-i\log |z|^2 \lvac \pi^1_0\eta^1_0\rvac -\sum_{n>0,m<0}\frac{z^{-n}w^{-m}}{nm}\lvac \psi^2_n\psi^1_m\rvac -\sum_{n>0,m<0}\frac{\bar{z}^{-n}\bar{w}^{-m}}{nm}\lvac \bar{\psi}^2_n\bar{\psi}^1_m\rvac\\
&=\log |z|^2-\sum_{n>0}\frac{z^{-n}w^n}{n}-\sum_{n>0}\frac{\bar{z}^{-n}\bar{w}^n}{n}\\
&=\log |z|^2+\log(1-w/z)+\log (1-\bar{w}/\bar{z})
\end{aligned}
\end{equation}
hence
\begin{equation}
\lvac \eta^2(z)\eta^1(w)\rvac=\frac{1}{2\kappa}\log|z-w|^2\,.
\end{equation}

Similarly
\begin{equation}
\begin{aligned}
\lvac \eta^2(z)\partial_x\eta^1(w)\rvac&=-\lvac \eta^1(z) \partial_x\eta^2(w)\rvac=\frac{1}{2\kappa}\left(\frac{iw}{z-w}+\frac{-i\bar{w}}{\bar{z}-\bar{w}}\right)\\
\lvac \eta^1(z)\pi^1(w)\rvac&=\lvac \eta^2(z)\pi^2(w)\rvac=\frac{i}{2}\left(\frac{w}{z-w}+\frac{\bar{w}}{\bar{z}-\bar{w}} \right)\,.
\end{aligned}
\end{equation}
For instance, these formulas enable us to reexpress the previous 3-point function as
\begin{equation}
G_\eta(z_2,z,z_1)=-\kappa^{-2}i(-\langle \eta^2(z_2) \pi^2(z)\rangle\langle \eta^2(z)\eta^1(z_1)\rangle+\langle \eta^2(z_2)\eta^1(z)\rangle\langle\pi^1(z)\eta^1(z_1)\rangle)\,,
\end{equation}
where we use the simplified notation $\langle X\rangle$ for $\lvac X\rvac$.

Because of the fields $\eta$ that involve $\log$ there is no scale invariance and the 2-point function of the fields $\phi^1_m$ is not as simply constrained as usual. In particular there are sub-leading corrections to the dominant terms. In the following we will denote by $\sim$ an equality up to sub-leading terms. The computation of the dominant behaviour of the 2-point functions of the fields $\phi^1_m$ is classical. We have
\begin{equation}
\begin{aligned}
\lvac \phi^2_m(z_2)\phi^1_m(z_1)\rvac&=\lvac \partial^m\eta^2(z_2)...\partial\eta^2(z_2)\eta^2(z_2)\eta^1(z_1)\partial\eta^1(z_1)...\partial^m\eta^1(z_1)\rvac\\
& \sim \frac{1}{2\kappa}\log|z_2-z_1|^2 \lvac \partial^m\eta^2(z_2)...\partial\eta^2(z_2)\partial\eta^1(z_1)...\partial^m\eta^1(z_1)\rvac\\
&\sim\frac{1}{2\kappa} \log|z_2-z_1|^2 \sum_{\sigma\in \mathfrak{S}_m} (-1)^\sigma \prod_{k=1}^m \lvac \partial^k\eta^2(z_2)\partial^{\sigma(k)}\eta^1(z_1)\rvac\\
&\sim \frac{1}{2\kappa}\log|z_2-z_1|^2 \sum_{\sigma\in \mathfrak{S}_m} (-1)^\sigma \prod_{k=1}^m \frac{(-1)^{k-1}(k+\sigma(k)-1)!}{(z_2-z_1)^{k+\sigma(k)}2\kappa}\\
&\sim\frac{ \log|z_2-z_1|^2}{(z_2-z_1)^{m(m+1)}(2\kappa)^{m+1}} \det ( (-1)^{k-1}(k+p-1)!)_{k,p}\,,
\end{aligned}
\end{equation}
where in the second line the dominant term is given by contracting $\eta^1$ with $\eta^2$ (otherwise the power-law is the same but without $\log$, thus sub-dominant). This gives the norm 
\begin{equation}
\langle \phi^2_m|\phi^1_m\rangle =(2\kappa)^{-m-1} \det \left\{ (-1)^{k-1}(k+p-1)! \right\}_{k,p=1}^m =(2\kappa)^{-m-1}(-1)^{\lfloor m/2 \rfloor}m! \prod_{k=1}^{m-1}(k!)^2 \,.
\end{equation}

\subsubsection{Dominant behaviour of the 3-point functions $\langle \phi^2_m\mathcal{V}\phi^1_m\rangle$}
Using the 2-point functions one can compute all the $\langle \phi^2_m\mathcal{V}\phi^1_m\rangle$. For example one has
\begin{equation}
\label{corr1}
\begin{aligned}
&\langle \partial\eta^2(z_2)\eta^2(z_2)\mathcal{V}(z)\eta^1(z_1)\partial\eta^1(z_1)\rangle=\\
&\tfrac{\kappa^{-4}}{8}\log|z_2-z|^2 \frac{|z|^2}{(z_2-z)^2|z_1-z|^2}+\tfrac{\kappa^{-4}}{8}\log|z_2-z_1|^2\frac{z}{(z_2-z)^2(z_1-z)} \\
&-\tfrac{\kappa^{-4}}{8}\log|z_2-z|^2\left(\frac{z}{(z_1-z)(z_2-z_1)^2}+\frac{\bar{z}}{(\bar{z}_1-\bar{z})(z_2-z_1)^2}+\frac{z}{(z_2-z_1)(z_1-z)^2} \right)\\
&+\tfrac{\kappa^{-4}}{4}\frac{|z|^2}{(z_2-z)^2|z_1-z|^2}+\tfrac{\kappa^{-4}}{8}\frac{\bar{z}}{(\bar{z}_2-\bar{z})(z_1-z)(z_2-z_1)}+\tfrac{\kappa^{-4}}{8}\frac{z}{(z_2-z)(z_1-z)(z_2-z_1)}\\
&+(z_1\to z_2,\quad z_2\to z_1)\,.
\end{aligned}
\end{equation}
The dominant behaviour is given by
\begin{equation}
G_{\eta\partial\eta}(z_2,z,z_1)\sim \tfrac{\kappa^{-4}}{8}\log|z_2-z|^2\frac{|z|^2}{(z_2-z)^2|z_1-z|^2}+\tfrac{\kappa^{-4}}{8}\log|z_1-z|^2\frac{|z|^2}{(z_1-z)^2|z_2-z|^2}\,.
\end{equation}
Formula \eqref{delta} gives for the full correlation function
\begin{equation}
\delta \hat{e}=-\kappa^2\times (2\kappa)^2\times \frac{\kappa^{-4}}{8}\times 3=-3/2\,.
\end{equation}
However formula \eqref{delta} does not capture only the dominant term in \eqref{corr1}, but also the sub-dominant term $-\log|z_2-z|^2\left(\frac{z}{(z_1-z)(z_2-z_1)^2}+\frac{\bar{z}}{(\bar{z}_1-\bar{z})(z_2-z_1)^2}\right)$. The dominant term comes from the normally ordered part of the potential $\normord{\mathcal{V}}$ whereas the second term comes from the regularized part $\xi_0(\eta^2\pi^2+\eta^1\pi^1)$. Both contribute to the displacement of energies, but only the first one is visible at leading order in the 3-point function. Note that this sub-dominant term is  not even the next-to-leading order term.\\

Let us evaluate the dominant term in the 3-point function $G_{\phi_m}(z_2,z,z_1)$. Let us first remark that in $\mathcal{V}$ the regularized term will always contribute one power less than the normally ordered term, so that the dominant term is given by $\normord{\mathcal{V}}$. We have thus
\begin{equation}
G_{\phi_m}(z_2,z,z_1)\sim \lvac \partial^m\eta^2(z_2)...\eta^2(z_2) \normord{(\eta^2\eta^1\partial_x\eta^2\partial_x\eta^1+\kappa^{-2}\eta^2\eta^1\pi^1\pi^2)(z)}\eta^1(z_1)...\partial^m\eta^1(z_1)\rvac\,.
\end{equation}
A priori, the dominant term will be given by contracting the four $\eta$ together, yielding a $\log |z_1-z|^2 \log |z_2-z|^2$. However we have the relations, using the abbreviated notations $\partial \eta$ for $\partial_z\eta$:
\begin{equation}
\label{relations}
\begin{aligned}
\lvac \pi^1(z)\partial^k\eta^1(w)\rvac&=-\kappa\lvac \partial_x\eta^2(z)\partial^k\eta^1(w)\rvac\\
\lvac \pi^2(z)\partial^k\eta^2(w)\rvac&=\kappa\lvac \partial_x\eta^1(z)\partial^k\eta^2(w)\rvac\,,
\end{aligned}
\end{equation}
valid for all $k\geq 1$. Thus
\begin{equation}
\lvac \partial^m\eta^2(z_2)...\partial\eta^2(z_2) \normord{(\partial_x\eta^2\partial_x\eta^1+\kappa^{-2}\pi^1\pi^2)(z)}\partial\eta^1(z_1)...\partial^m\eta^1(z_1)\rvac=0\,,
\end{equation}
and the $\log^2$ term vanishes. Similarly, if one contracts only $\eta^2(z_2)$ with $\eta^1(z)$, then one has to contract $\partial_x\eta^2(z)$ with $\eta^1(z_1)$ and $\pi^1(z)$ with $\eta^1(z_1)$ since the relations \eqref{relations} are verified for all $k$ but $k=0$ (otherwise the terms in $\partial_x\eta^2$ and $\pi^1$ will cancel out). Then:
\begin{equation}
\begin{aligned}
G_{\phi_m}(z_2,z,z_1)\sim & (2\kappa)^{-3}\log|z_2-z|^2\frac{4i\bar{z}}{\bar{z}-\bar{z}_1}\lvac \partial^m\eta^2(z_2)...\partial\eta^2(z_2)\normord{\eta^2\pi^2(z)}\partial\eta^1(z_1)...\partial^m\eta^1(z_1)\rvac\\
&+(2\kappa)^{-3}\log|z_1-z|^2\frac{4i\bar{z}}{\bar{z}-\bar{z}_2}\lvac \partial^m\eta^2(z_2)...\partial\eta^2(z_2)\normord{\eta^1\pi^1(z)}\partial\eta^1(z_1)...\partial^m\eta^1(z_1)\rvac\,.
\end{aligned}
\end{equation}
Contracting the remaining fields in the normal order, one gets
\begin{equation}
\begin{aligned}
&\lvac \partial^m\eta^2(z_2)...\partial\eta^2(z_2)\normord{\eta^2\pi^2(z)}\partial\eta^1(z_1)...\partial^m\eta^1(z_1)\rvac\\
&=\sum_{k,p=1}^m(-1)^{k+p+1}\lvac \partial^k\eta^2(z_2)\pi^2(z)\rvac \lvac \eta^2(z)\partial^p\eta^1(z_1)\rvac  \lvac \prod_{a=1,\neq k}^m \partial^a\eta^2(z_2)\prod_{b=1,\neq q}^m \partial^b\eta^1(z_1)\rvac\\
&=-\sum_{k,p=1}^m \frac{iz}{2}\frac{k!(p-1)! (2\kappa)^{-m}}{(z_2-z)^{k+1}(z_1-z)^p(z_2-z_1)^{m(m+1)-k-p}} \det ((-1)^{a-1} (a+b-1)!)_{a\neq k, b\neq q}\,.
\end{aligned}
\end{equation}
Denote $H_m$ the $m\times m$ matrix whose $(a,b)$ coefficient is $(-1)^{a-1}(a+b-1)!$. Using the relation between the adjugate matrix $\text{adj} H_m$ and its inverse, $\text{adj} H_m=(H_m^{-1})^t\det H_m$, we have
\begin{equation}
\label{corr3}
\begin{aligned}
G_{\phi_m}(z_2,z,z_1)&\sim (2\kappa)^{-m-3} \frac{ |z|^2 \log|z_2-z|^2}{|z-z_1|^2} 2\det H_m\sum_{k,p=1}^m \frac{(-1)^{k+p}k! (p-1)!(H_m^{-1})_{p,k}}{(z_2-z)^{k+1}(z_1-z)^{p-1}(z_2-z_1)^{m(m+1)-k-p}}\\
&+(z_1\to z_2,\quad z_2\to z_1)\,.
\end{aligned}
\end{equation}
This is the full dominant terms in the 3-point function. As already said, formula \eqref{delta} also counts a sub-dominant term in the 3-point function that is obtained by taking the regularized part of the potential. This term is
\begin{equation}
\begin{aligned}
&\lvac \partial^m\eta^2(z_1)...\eta^2(z_2)\normord{\eta^1\pi^1}(z)\eta^1(z_1)...\partial^m\eta^1(z_1)\rvac\\
&\sim \log |z_2-z|^2  \lvac \partial^m\eta^2(z_1)...\eta^2(z_2)\pi^1(z)\eta^1(z_1)...\partial^m\eta^1(z_1)\rvac\,.
\end{aligned}
\end{equation}
If one contracts $\pi^1(z)$ with a $\partial^k\eta^1(z_1)$, the resulting power of $z_2$ will be $-m(m+1)+k$ when the power of $z$ is zero and $z_1=0$, and will contribute to formula \eqref{delta} only if $k=0$. Thus the only term that counts is $(2\kappa)^{-m-3}\log|z-z_2|^2 \langle \pi^1(z)\eta^1(z_1) \rangle \det H_m (z_2-z_1)^{-m(m+1)}$. It contributes to $-1$ to $\delta \hat{e}$.

To apply now formula \eqref{delta} to the 3-point function \eqref{corr3}, let us expand
\begin{equation}
\begin{aligned}
&\frac{(-1)^{k+p}k! (p-1)!(H_m^{-1})_{p,k}}{(z_2-z)^{k+1}(0-z)^{p-1}(z_2-0)^{m(m+1)-k-p}}\\
&=(-1)^{k-1}k!(p-1)! (H_m^{-1})_{p,k} z^{1-p}z_2^{-m(m+1)+p-1}\sum_{q\geq 0}\frac{(k+q)!}{k!q!}(z/z_2)^q\,.
\end{aligned}
\end{equation}
Only $q=p-1$ contributes to $\delta \hat{e}$. This gives
\begin{equation}
\begin{aligned}
\delta \hat{e}&=-1-\frac{1}{2}\sum_{k,p=1}^m(-1)^{k-1}(k+p-1)!(H_m^{-1})_{p,k}\\
&=-1-\frac{1}{2}\sum_{k,p=1}^m(H_m)_{k,p}(H_m^{-1})_{p,k}\\
&=-1-\frac{1}{2}\tr H_m H_m^{-1}\,,
\end{aligned}
\end{equation}
hence
\begin{equation}
\delta \hat{e}=-1-\frac{m}{2}\,,
\end{equation}
recovering the previously derived correction \eqref{eq:correction2m}.

\subsubsection{Conclusion}
We conclude that in case of non-quasi-primary fields the relation between the logarithmic corrections to the energy levels and the three-point function is more involved than in \cite{cardy1986}. In particular the three-point function exhibits $m^2$ many equally dominant terms (i.e. with the same total divergence power, see \eqref{corr3}), that all contribute to the scaled gap. It moreover involves sub-dominant terms that contribute to the energy (although in a way independent from the magnetization).

\subsection{Symmetries}

The action of the $OSp(1|2)$ symmetry in the supersphere sigma model was already discussed in \cite{birgit2}. Interestingly, while the symmetry is spontaneously broken right at the conformal fixed point (here the sympletic fermion theory), the fact that this symmetry is present for all finite values of the system size (and thus finite values of the coupling constant $g$) gives rise to an enhancement of degeneracies at the fixed point. A very simple example of this is the ground state with $h=\bar{h}=0$. This state is degenerate four times at the fixed point, as the result of a degeneracy between the ground state (corresponding to an $OSp(1|2)$ singlet) and the order parameter multiplet (the $OSp(1|2)$ vector). Accordingly, the leading term in (\ref{osp12leadingcorr}) vanishes when $j={1\over 2}$. The correction term does break the degeneracy, in agreement with the fact that the only remaining symmetry at finite coupling is $OSp(1|2)$.

Similarly, we get eight fields with $(h,\bar{h})=(1,0)$: five come from the $OSp(1|2)$ currents (in the five dimensional adjoint), and three come from derivatives of the order parameter fields with vanishing conformal weight.

\section{$OSp(2|2)$}
\label{sec:osp22}

\subsection{The spectrum from field theory}
\subsubsection{The action}
For the $osp(2|2)$ case the constraint \eqref{supersphere} can be satisfied by setting
\begin{equation}
\phi_1=(1-\eta^2\eta^1)\cos \phi,\quad \phi_2=(1-\eta^2\eta^1)\sin \phi\,,
\end{equation}
so that the action reads
\begin{equation}
S=\frac{\kappa}{2\pi g}\int dxdt \left( \frac{1}{2}\partial_\mu\phi\partial_\mu\phi+\partial_\mu\eta^2\partial_\mu\eta^1-\eta^1\eta^2\partial_\mu\eta^1\partial_\mu\eta^2-\partial_\mu\phi\partial_\mu\phi\eta^2\eta^1\right)\,.
\end{equation}
Rescaling all the fields $\phi\to \sqrt{g}\phi$, $\eta^{1,2}\to \sqrt{g}\eta^{1,2}$ yields
\begin{equation}
\label{osp22action}
S=\frac{\kappa}{2\pi}\int dxdt \left( \frac{1}{2}\partial_\mu\phi\partial_\mu\phi+\partial_\mu\eta^2\partial_\mu\eta^1-g\eta^1\eta^2\partial_\mu\eta^1\partial_\mu\eta^2-g\partial_\mu\phi\partial_\mu\phi\eta^2\eta^1\right)\,,
\end{equation}
with here
\begin{equation}
g=\frac{\kappa}{2\log L}\,. \label{gkappa2logL}
\end{equation}

Note that the boson $\phi$ with original radius $2\pi$ becomes a boson with radius $2\pi/\sqrt{g}$.

\subsubsection{The normally ordered Hamiltonian}
To find the Hamiltonian, we write the action as
\begin{equation}
S=\frac{\kappa}{2\pi}\int dxdt\left( -\partial_x\eta^2\partial_x\eta^1+\dot{\eta}^2\dot{\eta}^1-\frac{1}{2}(\partial_x\phi)^2+\frac{1}{2}\dot{\phi}^2+g\mathcal{V}(\eta^1,\eta^2,\phi)\right)\,.
\end{equation}
With the modes
\begin{equation}
\phi(x,t)=\sum_k\phi_k(t)e^{ikx}\,,
\end{equation}
it reads
\begin{equation}
S=\kappa\int dt\left(-k^2\eta^2_k\eta^1_{-k}+\dot{\eta}^2_k\dot{\eta}^1_{-k}-\frac{1}{2}k^2\phi_k\phi_{-k}+\frac{1}{2}\dot{\phi}_k\dot{\phi}_{-k}+gV\right)\,.
\end{equation}
The conjugate momentum to $\phi_k$ is
\begin{equation}
\pi^\phi_k=\kappa\dot{\phi}_{-k}+\kappa g\frac{dV}{d\dot{\phi}_k}\,.
\end{equation}
The quantization procedure imposes at equal times
\begin{equation}
[\phi_k,\pi^\phi_p]=i\delta_{k,p}\,.
\end{equation}

The Hamiltonian is then
\begin{equation}
H=\kappa k^2\eta^2_k\eta^1_{-k}+\kappa^{-1}\pi^2_k\pi^1_{-k}+\frac{\kappa}{2}k^2\phi_k\phi_{-k}+\frac{\kappa^{-1}}{2}\pi^\phi_k\pi^\phi_{-k}-\kappa g  V\,.
\end{equation}

The $osp(2|2)$ charges are
\begin{equation}
\begin{aligned}
J_z&=\frac{\eta^1_k\pi^1_k-\eta^2_k\pi^2_k}{2i}\,,\qquad J_+=-i\eta^1_k\pi^2_k\,,\qquad J_-=-i\eta^2_k\pi^1_k\\
F_1&=\cos(\sqrt{g}\phi)_k (1-g\eta^2_l\eta^1_m)\pi^2_{-n}-\sqrt{g}\eta^1_k\sin(\sqrt{g}\phi)_l\pi^\phi_{-m}\\
F_2&=\sin(\sqrt{g}\phi)_k (1-g\eta^2_l\eta^1_m)\pi^2_{-n}+\sqrt{g}\eta^1_k\cos(\sqrt{g}\phi)_l\pi^\phi_{-m}\\
G_1&=\cos(\sqrt{g}\phi)_k (1-g\eta^2_l\eta^1_m)\pi^1_{-n}+\sqrt{g}\eta^2_k\sin(\sqrt{g}\phi)_l\pi^\phi_{-m}\\
G_2&=\sin(\sqrt{g}\phi)_k (1-g\eta^2_l\eta^1_m)\pi^1_{-n}-\sqrt{g}\eta^2_k\cos(\sqrt{g}\phi)_l\pi^\phi_{-m}\\
Q&=\frac{1}{2}\pi^\phi_0\,.
\end{aligned}
\end{equation}
The temporal derivatives can be computed as in the $osp(1|2)$ case. For example at order $g$ one has
\begin{equation}
\partial_t F_1=\partial_t F_1^{osp(1|2)}+g\left(-\frac{\kappa}{2}\phi_k\phi_l\eta^1_m m^2-\kappa^{-1}\eta^1_k\pi^\phi_{-l}\pi^\phi_{-m}+\kappa\eta^1_k\phi_l\phi_m m^2+\kappa \frac{dV_\phi}{d\eta^2_0}\right)\,,
\end{equation}
if one decomposes the potential $V$ as $V=V^{osp(1|2)}+V_\phi$.

We now impose the following perturbation
\begin{equation}
\mathcal{V}(\eta^1,\eta^2,\phi)=(\partial_x\phi)^2\eta^2\eta^1-(\partial_t\phi)^2\eta^2\eta^1+\eta^2\eta^1\partial_x\eta^2\partial_x\eta^1-\eta^2\eta^1\partial_t\eta^2\partial_t\eta^1\,,
\end{equation}
that gives
\begin{equation}
V=-kl\phi_k\phi_l\eta^2_m\eta^1_n-\kappa^{-2}\pi^\phi_{-k}\pi^\phi_{-l}\eta^2_m\eta^1_n-mn\eta^2_k\eta^1_l\eta^2_m\eta^1_n+\kappa^{-2}\eta^2_k\eta^1_l\pi^1_{-m}\pi^2_{-n}\,.
\end{equation}
This ensures the conservation of the $osp(2|2)$ charges.

Define now the modes
\begin{equation}
a_k=\frac{-ik\phi_k\kappa^{1/2}+\pi^\phi_{-k}\kappa^{-1/2}}{\sqrt{2}},\quad \bar{a}_k=\frac{-ik\phi_{-k}\kappa^{1/2}+\pi^\phi_{k}\kappa^{-1/2}}{\sqrt{2}}\,,
\end{equation}
that satisfy
\begin{equation}
[a_k,a_{-k}]=k,\quad [\bar{a}_k,\bar{a}_{-k}]=k\,.
\end{equation}
The former modes read
\begin{equation}
\phi_k=\frac{i}{k\sqrt{2\kappa}}(a_k-\bar{a}_{-k}),\quad \pi^\phi_k=\sqrt{\kappa/2}(a_k+\bar{a}_{-k})\,,
\end{equation}
and the potential can be rewritten
\begin{equation}
\label{potentialosp22}
V=\frac{1}{2\kappa^2}\left(-i\sqrt{2\kappa}\eta^2_0+\frac{\psi^2_k-\bar{\psi}^2_{-k}}{k}\right)\left(-i\sqrt{2\kappa}\eta^1_0+\frac{\psi^1_l-\bar{\psi}^1_{-l}}{l}\right)(\psi^2_m\bar{\psi}^1_{-n}+\bar{\psi}^2_{-m}\psi^1_n+a_m\bar{a}_{-n}+\bar{a}_{-m}a_n)\,.
\end{equation}
The bosonic part is already normally ordered, and the fermionic part is the same as in the $osp(1|2)$ case. Hence
\begin{equation}
V=\normord{V}+i\kappa^{-2}(\eta^2_k\pi^2_k+\eta^1_k\pi^1_k)\xi_0\,.
\end{equation}
The total Hamiltonian reads then
\begin{equation}
\label{Hosp22}
\begin{aligned}
H&=\frac{1}{2}\sum_{k<0} (a_k a_{-k}+ \bar{a}_k\bar{a}_{-k})+a_0^2+\sum_{k<0}(\psi^2_k\psi^1_{-k}-\psi^1_{-k}\psi^2_k+\bar{\psi}^2_k\bar{\psi}^1_{-k}-\bar{\psi}^1_{-k}\bar{\psi}^2_k)+2\psi^2_0\psi^1_0-\frac{c}{12}-\kappa g\normord{V}\\
&-ig\kappa^{-1}(\eta^2_k\pi^2_k+\eta^1_k\pi^1_k)\xi_0\,.
\end{aligned}
\end{equation}

\subsubsection{Construction of the states}
Once again derivatives $\partial_z\eta^1$ involve terms $(1+g/2\eta^2_0\eta^1_0)$, that vanish when a $\eta^1$ is already present in a state. The highest-weight state of $J_z$-charge $m+1/2$ and $Q_z$-charge $n$ is
\begin{equation}
|\cos(2n\sqrt{g}\phi)\eta^1\partial\eta^1...\partial^m\eta^1\bar{\partial}\eta^1...\bar{\partial}^m\eta^1\rangle=\cos(2n\sqrt{g}\phi_0)\eta^1_0\psi^1_{-1}...\psi^1_{-m}\bar{\psi}^1_{-1}...\bar{\psi}^1_{-m}\rvac\,.
\end{equation}

\subsubsection{Regularization}
As in the $osp(1|2)$ case,  we need a regularization, i.e., fixing the value of $\xi_0$. In this case we did not write the $osp(2|2)$ charges in terms of their normally ordered version. They also would depend on the $\xi$'s, and any value for the $\xi$'s would give a Hamiltonian with the $osp(2|2)$ symmetry and with the classical non-linear sigma model as classical limit. In the $osp(1|2)$ case, the regularization that we chose corresponded to imposing that the charges are not modified by the normal order. Here we can impose a similar constraint by constraining $\eta^1_0\rvac$, $\eta^2_0\rvac$, $\cos(\sqrt{g}\phi_0)(1-g\eta^2_0\eta^1_0)\rvac$ and $\sin(\sqrt{g}\phi_0)(1-g\eta^2_0\eta^1_0)\rvac$ to belong to the same representation, and thus to have the same energy at order $g$. We have
\begin{equation}
H\eta^1_0\rvac=g\kappa^{-1}\xi_0 \eta^1_0\rvac,\quad H\eta^2_0\rvac=g\kappa^{-1}\xi_0 \eta^2_0\rvac\,,
\end{equation}
and
\begin{equation}
\begin{aligned}
H\cos(\sqrt{g}\phi_0)(1-g\eta^2_0\eta^1_0)\rvac&=\frac{1}{2\kappa}(g-2g)\cos(\sqrt{g}\phi_0)(1-g\eta^2_0\eta^1_0)\rvac\\
&=-\frac{\kappa^{-1}g}{2}\cos(\sqrt{g}\phi_0)(1-g\eta^2_0\eta^1_0)\rvac\,.
\end{aligned}
\end{equation}
Thus we impose
\begin{equation}
\lvac \xi_0\rvac=-\frac{1}{2}\,.
\end{equation}

\subsubsection{Corrections to the energy levels \label{sec:osp22corr}}
As soon as the states involve bosons only through $\cos(2n\sqrt{g}\phi_0)$ the bosonic part of the potential \eqref{potentialosp22} does not play any role at order $g$, and the fermionic part is the $osp(1|2)$ case. Only the unperturbed part of the Hamiltonian plays a role for the bosons at order $g$. Thus the bosonic and the fermionic part are actually decoupled at order $g$ and the calculations for the fermionic part are identical to the $osp(1|2)$ case. With
\begin{equation}
a_0^2\cos(2n\sqrt{g}\phi_0)\rvac=2\kappa^{-1}n^2 g\cos(2n\sqrt{g}\phi_0)\rvac\,,
\end{equation}
we get
\begin{equation}
\label{eq:osp22}
\begin{aligned}
&H|\cos(2n\sqrt{g}\phi)\eta^1\partial\eta^1...\partial^m\eta^1\bar{\partial}\eta^1...\bar{\partial}^m\eta^1\rangle\\
&=(m(m+1)-(m(m+1)+\tfrac{1}{2}-2n^2)\kappa^{-1}g)|\cos(2n\sqrt{g}\phi)\eta^1\partial\eta^1...\partial^m\eta^1\bar{\partial}\eta^1...\bar{\partial}^m\eta^1\rangle\,,
\end{aligned}
\end{equation}
hence
\begin{equation}
\label{eq:osp22gapnm}
\boxed{
\frac{L^2 \Delta e_L}{2\pi v_F}= m(m+1)-(m(m+1)+\tfrac{1}{2}-2n^2)\kappa^{-1}g\,.
}
\end{equation}
Similarly for non-symmetric states $|\cos(2n\sqrt{g}\phi)\eta^1\partial\eta^1...\partial^m\eta^1\bar{\partial}\eta^1...\bar{\partial}^{\bar{m}}\eta^1\rangle$ one has
\begin{equation}
\label{eq:osp22nonsym}
\frac{L^2 \Delta e_L}{2\pi v_F}=\tfrac{m(m+1)}{2}+\tfrac{\bar{m}(\bar{m}+1)}{2}-(m\bar{m}+\tfrac{m+\bar{m}}{2}+\tfrac{1}{2}-2n^2)\kappa^{-1}g\,.
\end{equation}

\subsubsection{Density of critical exponents}

The previous formula gives an infinite number of fields with the same conformal weight $h=m(m+1)$ when $L\to\infty$, thanks to the bosonic degree of freedom $n$. In finite size, the degenerescence is lifted with a $2\kappa^{-1}gn^2=n^2/\log L$ and one actually sees a continuum of conformal weights starting from $m(m+1)$. The question is to find the number of states that are there between $h$ and $h+dh$ in finite size $L$. Denote $C^m_L(h)$ the number of such states for magnetization $m$. This number is for large $L$
\begin{equation}
C^m_L(h)=\# \left\{n\geq 0, h\leq \frac{n^2}{\log L}\leq h+dh \right\}\,,
\end{equation}
if we assume that the higher-order correction terms to the conformal dimensions behave as $n^{k} (\log L)^{-p}$ with $k<2p$ (this is true in the XXX case for example, see \cite{lukyanov}). $\# S$ denotes the number of elements in the set $S$. At first order in $dh/h$ these $n$ must satisfy
\begin{equation}
\sqrt{h\log L}\leq n\leq \sqrt{h\log L}(1+\tfrac{dh}{2h})\,.
\end{equation}
Hence:
\begin{equation}
C^m_L(h)=\frac{1}{2}\sqrt{\frac{\log L}{h}}dh+O(dh^2)\,.
\end{equation}
Introduce now the variable $s$ by $h=s^2$. This gives the density of states $\rho^m(s)$ for the variable $s$ for magnetization $m$
\begin{equation}
\rho^m(s)=1\,,
\end{equation}
in the sense that there are $\rho^m(s)\sqrt{\log L}ds$ states with magnetization $m$ whose $s$ is between $s$ and $s+ds$ in size $L$. This is the dominant behaviour of the density as $L\to\infty$. The  corrections in finite size may contain a more complicated behaviour such as the black hole CFT \cite{ikhlef2}.


\subsection{The spectrum from the spin chain}
\subsubsection{Bethe equations}

In the  $osp(2|2)$ case, typical irreducible representations are characterized by a pair of numbers $q,j$ (and denoted $[q,j]$ in what follows) which are the eigenvalues of the generators $Q_z$ and $J_z$ on the highest-weight state. Here $q$ can be any complex number, while $j=0,1/2,\ldots$. These representations have dimension $8j$ \cite{lu}, and Casimir 
\begin{equation}
\label{eq:casimirosp22}
{\cal C}=2j^2-2q^2\,.
\end{equation}

Note that in contrast with $osp(1|2)$, the tensor products of the $[0,1/2]$ representations at each site of the chain involve not completely reducible representations. The simplest example of this is the tensor product of $[0,1/2]$ with itself which is a direct sum of the eight-dimensional adjoint $[0,1]$ and of an indecomposable mixing the atypical representations $[\pm 1/2,1/2]$ (both of dimension 3) and two copies of the identity $[0,0]$. For example, the ground state in even sizes is $8$ times degenerated, has a Bethe state with charges $Q_z=1/2,J_z=1/2$, and decomposes into $[1/2,1/2]$, $[-1/2,1/2]$ and two copies of the identity $[0,0]$. The Bethe state with charges $Q_z=3/2,J_z=1/2$ belongs to a $8\times \tfrac{1}{2}=4$-dimensional irreducible representation $[3/2,1/2]=(3/2,1/2)\oplus (1,0)\oplus (2,0)$. Another Bethe state with charges $Q_z=-3/2,J_z=1/2$ belongs to a similar $4$-dimensional irreducible representation, making this energy level $8$ times degenerated.

The $osp(2|2)$ spin chain is described by two families of roots $\lambda_i$ and $\mu_i$ satisfying the Bethe equations \cite{galleasmartins,galleasmartins2}
\begin{equation}
\begin{aligned}
\left( \frac{\lambda_i+i/2}{\lambda_i-i/2}\right)^L&=\prod_{j}\frac{\lambda_i-\mu_j+i}{\lambda_i-\mu_j-i}\\
\left( \frac{\mu_i+i/2}{\mu_i-i/2}\right)^L&=\prod_{j}\frac{\mu_i-\lambda_j+i}{\mu_i-\lambda_j-i}\,.
\end{aligned}
\end{equation}
An eigenvalue of the Hamiltonian for one set of solutions $\lambda_1,...,\lambda_{M_1},\mu_1,...,\mu_{M_2}$ to these equations is then
\begin{equation}
e_L=-\frac{1}{L}\sum_{i=1}^{M_1} \frac{1}{\lambda_i^2+1/4}-\frac{1}{L}\sum_{i=1}^{M_2} \frac{1}{\mu_i^2+1/4}\,.
\end{equation}
The spins $j$ and $q$ (ie, the eigenvalue of $J_z$ and $Q_z$ respectively) corresponding to a solution with $M_{1,2}$ roots $\lambda_i,\mu_i$ are given by
\begin{equation}
M_1=L/2-(j+q)\,,\quad \quad M_2=L/2-(j-q)\,.
\end{equation}

Note that if another grading is chosen, i.e., another choice for the fermionic sign in \eqref{eq:grading}, the Bethe equations would be different. As far as the eigenvalues are concerned, the two gradings are equivalent, see appendix \ref{sec:grading}.

In the $osp(2|2)$ we observe that the Bethe states have same charges $J_z$ and $Q_z$ as the highest-weight state of the multiplet they belong to.

\subsubsection{Root structure}

On the lattice in even size $L$, the field $\cos(2n\sqrt{g}\phi_0)\eta^1\partial\eta^1...\partial^m\eta^1\bar{\partial}\eta^1...\partial^m\eta^1$ is obtained with $L/2-(n+m+1/2)$ roots $\lambda$ and $L/2-(m-n+1/2)$ roots $\mu$. They are real and symmetrically distributed. See Figure \ref{fig:rootosp22} for a plot of some root structures.

 \begin{figure}
 \begin{center}
\includegraphics[scale=0.5]{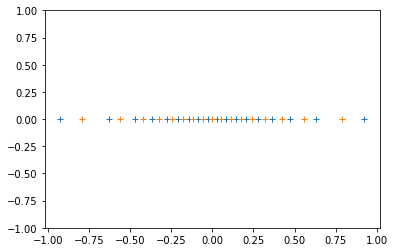} 
\end{center}
\caption{Bethe roots in the complex plane for the ground state (first family $\lambda_i$ in blue, second family $\mu_i$ in orange), for $L=36$.}
\label{fig:rootosp22}
\end{figure}

The gap computed previously for the ground state of the sector $Q_z=q$, $J_z=j$ reads, when $2j$ and $q$ are integers with same parity:
\begin{equation}
\frac{L^2\Delta e_L}{2\pi v_F}=j^2-\frac{1}{4}-\left(j^2-2q^2+\frac{1}{4} \right)\kappa^{-1}g \label{osp22gaps}\,.
\end{equation}

Like for $osp(1|2)$, we see that the vector representation is degenerate with the ground state in the limit $g\to 0$ since $j^2-{1\over 4}=0$ for $j=1/2$. We also see that the order $g$ corrections vanishes when  $j=1/2,q=\pm 1/2$: this is compatible with the fact that the corresponding representations are ``mixed'' with the identity in a bigger $osp(2|2)$ indecomposable representation. 

\subsubsection{Numerical results}
We give numerical evidence in Figure \ref{fig:osp22} for the formula given in \eqref{eq:osp22}, with $Z_L^{q,j}$ denoting the measured $Z^{q,j}_L=(\frac{L^2}{2\pi v_F}(e_L-e_L^0)-(h+\bar{h}))\log L$ in finite size $L$ for the state $[q,j]$. Here $e_L^0$ denotes the state $[-1/2,1/2]$.

 \begin{figure}
 \begin{center}
\includegraphics[scale=0.5]{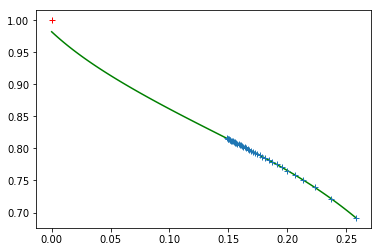} 
\includegraphics[scale=0.5]{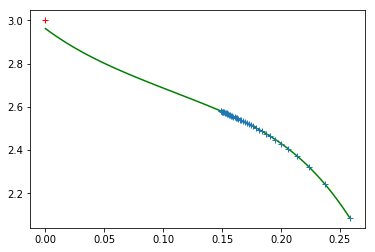} 
\includegraphics[scale=0.5]{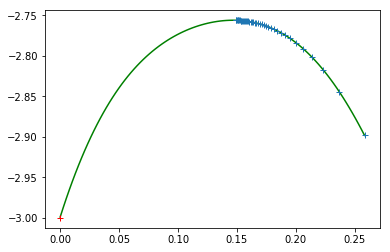} 
\includegraphics[scale=0.5]{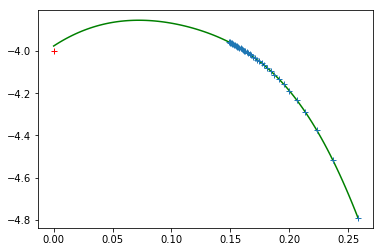} 
\includegraphics[scale=0.5]{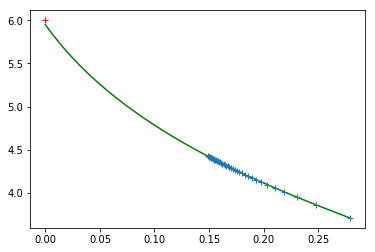} 
\includegraphics[scale=0.5]{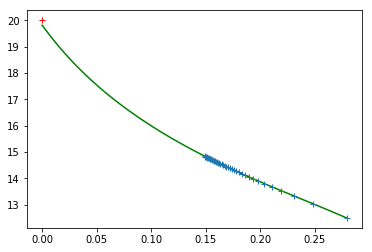} 
\end{center}
\caption{In reading direction: plots of $Z^{-3/2,3/2}_L$, $Z^{-5/2,5/2}_L$, $Z^{-1/2,5/2}_L$, $Z^{-3/2,7/2}_L$, $Z^{-5/2,1/2}_L$,$Z^{-9/2,1/2}_L$, as a function of $1/\log L$, together with their extrapolated curve with functions $f_{12}$, $f_{12}$, $f_{10}$, $f_{10}$, $f_{12}$, $f_{10}$. The theoretical results are, respectively, $1$, $3$, $-3$, $-4$, $6$, $20$.}
\label{fig:osp22}
\end{figure}

\section{$OSp(3|2)$}

\subsection{The spectrum from the spin chain}

\subsubsection{The Bethe equations \label{sec:osp32multiplet}}

Irreducible  representations of $osp(3|2)$ are characterized by a pair of numbers $q,j$ corresponding to the spin of the underlying $o(3)$ and $sp(2)$ bosonic sub-algebras. Here $q=0,1,\ldots $ is integer, and $j=0,1/2,\ldots$ is half-integer. The five-dimensional fundamental representation is $[0,1/2]$ and the twelve-dimensional adjoint representation is $[0,1]$. The (quadratic) Casimir eigenvalues are 
\begin{equation}
\label{eq:casimirosp32}
{\cal C}=j (2j-1)-\frac{1}{2}q(q+1)\,.
\end{equation}

The $osp(3|2)$ spin chain is described by two families of roots $\nu_i$ and $\mu_i$ satisfying the Bethe equations \cite{galleasmartins,frahmmartins}
\begin{equation}
\begin{aligned}
\left(\frac{\nu_i+i/2}{\nu_i-i/2} \right)^L&=\prod_{j=1}^{M_2}\frac{\nu_i-\mu_j+i/2}{\nu_i-\mu_j-i/2}\\
1&=\prod_{j=1}^{M_1'}\frac{\mu_i-\nu_j+i/2}{\mu_i-\nu_j-i/2}\prod_{j\neq i}\frac{\mu_i-\mu_j-i/2}{\mu_i-\mu_j+i/2}\,.
\end{aligned}
\end{equation}
As already explained, the Bethe equations depend on the choice of the grading, see appendix \ref{sec:grading}. It turns out that they are more convenient in another grading. We write $\lambda_i$ and $\mu_i$ the roots of the Bethe equations in this second grading. They read
\begin{equation}
\label{osp32bis}
\begin{aligned}
\left(\frac{\lambda_i+i/2}{\lambda_i-i/2} \right)^L&=\prod_{j=1}^{M_2}\frac{\lambda_i-\mu_j+i/2}{\lambda_i-\mu_j-i/2}\\
1&=\prod_{j=1}^{M_1}\frac{\mu_i-\lambda_j+i/2}{\mu_i-\lambda_j-i/2}\prod_{j\neq i}\frac{\mu_i-\mu_j+i/2}{\mu_i-\mu_j-i/2}\cdot \frac{\mu_i-\mu_j-i}{\mu_i-\mu_j+i}\,.
\end{aligned}
\end{equation}
For each solution $(\nu_i,\mu_i)$ of the equations in the first grading there exist a solution $(\lambda_i,\mu_i)$ of the equations in the second grading, and vice versa. The $\mu_i$ stay the same (as anticipated by the notation), and the roots $\lambda_i$ and $\nu_i$ are related by the fact that together, they form all the roots of the following polynomial
\begin{equation}
P(X)=(X+i/2)^L\prod_{j=1}^{M_2}(X-\mu_j-i/2)-(X-i/2)^L\prod_{j=1}^{M_2}(X-\mu_j+i/2)\,.
\end{equation}
In the second grading an eigenvalue of the Hamiltonian is given by
\begin{equation}
e_L=\frac{1}{L}\sum_{i=1}^{M_1} \frac{1}{\lambda_i^2+1/4}\,.
\end{equation}
The spins $j$ and $q$ (ie, the eigenvalue of $J_z$ and $Q_z$ respectively) for a solution with $M_1$ roots $\lambda_i$ and $M_2$ roots $\mu_i$ are related to $M_1$ and $M_2$ through
\begin{equation}
\label{eq:chargesosp32}
M_1=L-q\,,\quad \quad M_2=L-2j-q\,.
\end{equation}
However an important remark has to be made. It is known that for the XXX spin chain, the Bethe vectors are highest-weight vectors with respect to $J_z$, meaning that they are annihilated by the total $J_+$. It turns out that it is not the case for $osp(3|2)$: some Bethe vectors are indeed annihilated by the raising operators of the $sp(2)$ and $o(3)$ subalgebras, but not by the raising operators of the full $osp(3|2)$ algebra. To see this, one can go to the $q$-deformed version where most of the $osp(3|2)$ degeneracies are lifted. In this case there are states with similar root structure as in the $q=1$ undeformed case, but with additional roots with imaginary part $i\pi/2$. When $q\to 1$, the energy of theses states converge to the same multiplet with the same energy in finite size, since the extra roots at $i\pi/2$ have no effect in this limit. For example there is one state at $q\neq 1$ that has one extra root $\lambda_1=i\pi/2$ compared to the $q=1$ case, that falls into the multiplet when $q\to 1$. In its multiplet at $q\neq 1$ there is the state that becomes annihilated by all the raising operators of $osp(3|2)$ when $q\to 1$, which is the highest-weight state. The important point is that the charges of this state with an extra root has a $Q_z$ decreased by $1$ and a $J_z$ increased by $1/2$ compared to the state that can be built with the Bethe ansatz at $q=1$. Therefore the charge of the multiplet of the Bethe vector has actually a $Q_z$ decreased by $1$ and a $J_z$ increased by $1/2$. This is important for the bosonic part of logarithmic corrections to match the value of the Casimir, but it will also be important in section \ref{sec:eigenvalues}.

This created some confusion in \cite{frahmmartins}. To make contact with their\footnote{The subscripts are the author's initials \cite{frahmmartins,vanderjeugt}.} notations $(p_{\rm FM},q_{\rm FM})$ for labelling the Bethe states (but not the multiplets), we have $p_{\rm FM}=2j$ and $q_{\rm FM}=q/2$. As for the notations $(p_{\rm VdJ},q_{\rm VdJ})$ in \cite{vanderjeugt}, we have $p_{\rm VdJ}=q$ and $q_{\rm VdJ}=j$.

For example the first excited state belongs to a $12$-dimensional multiplet and the Bethe state has charges $(Q_z,J_z)=(2,0)$. In \cite{frahmmartins} it was interpreted as the irreducible representation $q_{\rm VdJ}=1, p_{\rm VdJ}=0$ of dimension $12$ in  \cite{vanderjeugt}, whereas it is actually the irreducible representation $q_{\rm VdJ}=1/2, p_{\rm VdJ}=1$ of dimension $12$ as well. It is a reducible representation for $o(3)\otimes sp(2)$ that reads $(1,1/2)\oplus (0,0) \oplus (2,0)$ in terms of $Q_z,J_z$. Only the state with $(Q_z,J_z)=(1,1/2)$ is annihilated by the raising operators of the whole $osp(3|2)$ algebra.

\subsubsection{Root structure}
A particular feature of this model is that the energy of the ground state $e_L^0$ is independent of $L$. In terms of Bethe roots it is given by $L$ coinciding roots $\mu$ at zero, see \cite{frahmmartins}. Note that in the $O(1)$ model a similar phenomenon inspired the Razumov-Stroganov conjecture concerning the entries of the eigenvector associated to this particular eigenvalue \cite{razumovstroganov,knutsonpzj}.

The first state whose Bethe state has charges $j$ integer and $q$ even is given on the lattice in the second grading by $L/2-(q/2+j)$ strings composed of $2$ roots $\lambda_i$ whose imaginary part is approximately $\pm 3/4$ and $2$ roots $\mu_i$ whose imaginary part is approximately  $\pm 1/4$, plus $2j$ real roots $\lambda_i$ taking a large value, lying outside the strings. This is illustrated in Figure~\ref{fig:rootosp32}.

 \begin{figure}[H]
 \begin{center}
\includegraphics[scale=0.5]{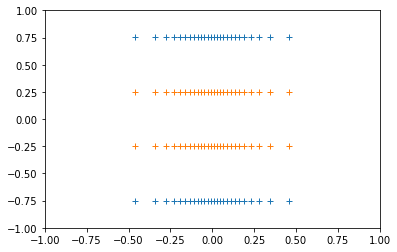} 
\includegraphics[scale=0.5]{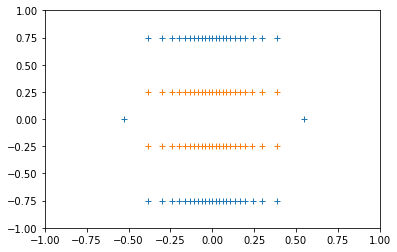}
\end{center}
\caption{Bethe roots in the complex plane for the states $j=0,q=2$ (left) and $j=1,q=2$ (right), for $L=54$.}
\label{fig:rootosp32}
\end{figure}

The presence of strings is a complication, both numerically and analytically. The 
typical deviation of their imaginary part from $\pm i/4$ or $\pm 3i/4$ is observed to behave as $\log L/L$ with the size of the system.

\subsubsection{Numerical results}
We observe numerically the following behaviour at large $L$, in terms of the charges $j$ and $q$ in \eqref{eq:chargesosp32} \textit{of the Bethe states}
\begin{equation}
\label{eq:osp32logcorr}
\frac{L^2\Delta e_L}{2\pi v_F}= j(j+1)-\left(j(j+1)-\frac{1}{2}q(q-1) \right)\kappa^{-1}g\,.
\end{equation}
In terms of the charges $j$ and $q$ \textit{of the multiplet} it belongs to, it reads
\begin{equation}
\label{eq:osp32logcorr2}
\boxed{
\frac{L^2\Delta e_L}{2\pi v_F}= j^2-\frac{1}{4}-\left(j^2-\frac{1}{4}-\frac{1}{2}q(q+1) \right)\kappa^{-1}g\,.}
\end{equation}
The bosonic part corresponds to the Casimir \eqref{eq:casimirosp32}, but not the fermionic part.

Here we see the importance of considering the charges of the multiplet and not those of the Bethe state. To our knowledge it has not been noticed before for this model. It is observed only for $osp(3|2)$ and not $osp(1|2)$ nor $osp(2|2)$, and in some other spin chains with Lie superalgebra symmetry their highest-weight property has been proven \cite{essler}, suggesting that it is peculiarity of this model rather than a common feature. However one can ask if this also happens in the higher-rank superalgebras studied numerically in \cite{FM}. From our experience it seems that studying the $q$-deformed version of the spin chain helps understanding these aspects: most of the $osp(3|2)$ degeneracies are then lifted and more Bethe states can be built that fall into a same multiplet as $q\to 1$; but with different charges, and in particular with higher $J_z$ charge than that of the only Bethe state that can be built at $q=1$.

In Figure \ref{fig:osp32num} are shown the numerical verifications of formula \eqref{eq:osp32logcorr}, where $Z_L^{j,q}$ denotes the measured $Z^{j,q}_L=(\frac{L^2}{2\pi v_F}(e_L-e_L^0)-(h+\bar{h}))\log L$ in finite size $L$ for the state whose Bethe state has $J_z=j$, $Q_z=q$, where $e_L^0$ is the reference state $j,q=0,0$.

\begin{figure}
 \begin{center}
\includegraphics[scale=0.5]{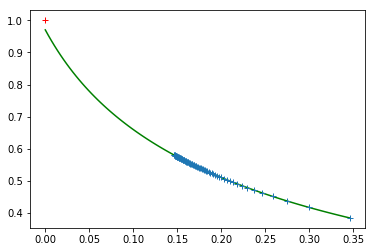} 
\includegraphics[scale=0.5]{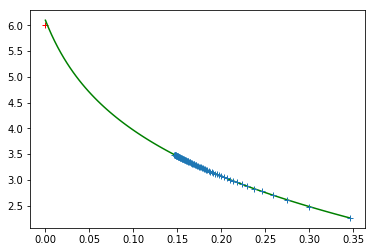} 
\includegraphics[scale=0.5]{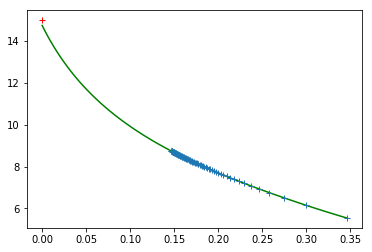} 
\includegraphics[scale=0.5]{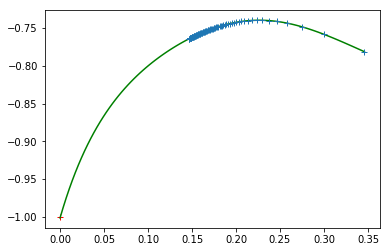} 
\includegraphics[scale=0.5]{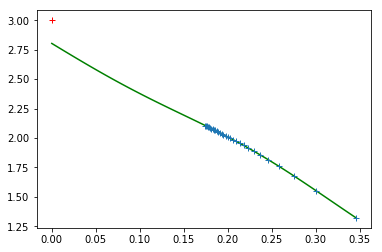}
\includegraphics[scale=0.5]{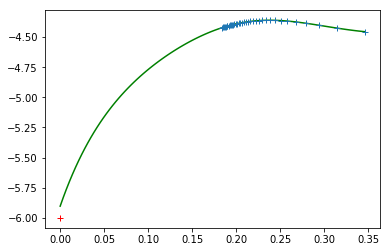}  
\end{center}
\caption{In reading direction: plots of $Z^{0,2}_L$, $Z^{0,4}_L$, $Z^{0,6}_L$, $Z^{1,2}_L$, $Z^{1,4}_L-Z^{0,2}_L$, $Z^{2,2}_L-Z^{0,2}_L$, as a function of $1/\log L$, together with their extrapolated curve with functions $f_{8}$, $f_{8}$, $f_{8}$, $f_{5}$, $f_5$, $f_8$. The theoretical results are, respectively, $1$, $6$, $15$, $-1$, $3$, $-6$.}
\label{fig:osp32num}
\end{figure}

\section{$OSp(4|2)$}
\subsection{Motivations}
The $OSp(4|2)$ sigma model is special from the field theoretic point of view. In this case indeed, it is known that the beta function vanishes to all orders, and it is expected that the sigma model is exactly conformal, with a line of fixed points that is closely related with the Kosterlitz-Thouless phase in the underlying $O(2)$ XY-model. It is also expected that the sigma model is dual in a certain sense to a Gross-Neveu model---just like the $O(2)$ free compact boson model is dual to the massless Thirring model. A lot of progress on this special case has been obtained in the string theory literature \cite{candu1,cagnazzo,mitev}. On the side of lattice regularizations, the dense loop soup with varying intersection weight $w$ has been studied with algebraic and direct diagonalization techniques \cite{candu2,candu3}. To this day however, no integrable version of this model is known to exist for arbitrary $w$.  This is related with singular properties of the $osp(r|2s)$ $R$ matrix for $r-2s=2$, as we now discuss.

\subsection{The $R$-matrix}
The $osp(4|2)$ spin chain constructed from section \ref{sec:chain} has an ill-defined Hamiltonian that can be somehow regularized \cite{FM} by replacing $\lambda$ by $\lambda(2-r+2s)/2$ and then setting $r-2s\to 2$. One gets the $R$-matrix
\begin{equation}
R_{ab}(\lambda)=P_{ab}+\frac{\lambda}{1-\lambda}E_{ab}\,,
\end{equation}
where $E_{ab}$ is a generator of the Temperley-Lieb algebra 
with parameter $N=2$, well-known from the $su(2)$ case, but represented here as a $36\times 36$ matrix. This $R$-matrix has indeed the $osp(4|2)$ symmetry and the theory obtained is relativistic. However, the aforementioned limit makes no sense in terms of the Bethe equations: these do not depend smoothly on the variables $r,s$ that are anyway discrete.

To get around this problem, we can have a look at the spin chains with $sl(4|2)^{(2)}$ symmetry, which contains the $osp(4|2)$ symmetry. It is well known that the spin chains with $sl(n|m)$ symmetry are non-relativistic \cite{saleur1999}, and the same is in fact true for the spin chains with 
 $sl(r|2)^{(2)}$ symmetry for $r=1,2,3,4$ (this was previously noticed for $sl(2|2)^{(2)}$ \cite{martinsramos}). However, we can now consider the $q$-deformed version of $sl(4|2)^{(2)}$, denoted $sl(4|2)^{(2)}_q$ \cite{galleasmartins}. This model contains as a sub-spectrum the levels of the $su(2)_{q'}$ (with $q'=-q^2$)  spin chain,  among which is the ground state for the antiferromagnetic regime. While the limit $q\to 1$ is non-relativistic (just like the ferromagnetic $su(2)$ spin $1/2$ chain),  the limit $q\to i$, after a rescaling $\lambda\to \lambda (q-i)$ happens to be relativistic, with a spectrum that contains the  levels of the antiferromagnetic $su(2)$ spin $1/2$ chain. It turns out that the transfer matrix obtained this way has exactly the same eigenvalues with the same degeneracies as the $osp(4|2)$ rescaled $R$-matrix discussed above, and the limit $q\to i$ makes perfect sense in terms of the Bethe ansatz equations. We shall thus study this model in the following.

\subsection{A brief description of $sl(4|2)^{(2)}_q$}
Let us first discuss briefly the Bethe equations of $sl(4|2)^{(2)}_q$. These are, in the second grading \cite{galleasmartins} with $q=e^{i\gamma}$
\begin{equation}
\left\{
\begin{aligned}
\left(\dfrac{\sinh(\lambda_i+i\gamma/2)}{\sinh(\lambda_i-i\gamma/2)} \right)^L&=\prod_{j=1}^{M_2}\dfrac{\sinh(\lambda_i-\mu_j+i\gamma/2)}{\sinh(\lambda_i-\mu_j-i\gamma/2)}\\
1&=\prod_{j=1}^{M_1}\dfrac{\sinh(\mu_i-\lambda_j+i\gamma/2)}{\sinh(\mu_i-\lambda_j-i\gamma/2)}\prod_{j\neq i}^{M_2}\dfrac{\sinh(\mu_i-\mu_j-i\gamma)}{\sinh(\mu_i-\mu_j+i\gamma)}\prod_{j=1}^{M_3}\dfrac{\sinh(2\mu_i-2\nu_j+i\gamma)}{\sinh(2\mu_i-2\nu_j-i\gamma)}\\
1&=\prod_{j=1}^{M_2}\dfrac{\sinh(2\nu_i-2\mu_j-i\gamma)}{\sinh(2\nu_i-2\mu_j+i\gamma)}\prod_{j\neq i}^{M_3}\dfrac{\sinh(2\nu_i-2\nu_j+2i\gamma)}{\sinh(2\nu_i-2\nu_j-2i\gamma)}\,.
\end{aligned}
\right.
\end{equation}
The ground state and first excitations are essentially given by configurations with three degrees of freedom $(n,m,p)$. These numbers correspond to   $L/2-m$ strings composed of $2$ roots $\lambda$ approximately at $\pm i\pi/4$ and $2$ roots $\mu$ approximately at $\pm i(\pi/4-\gamma/2)$, $n-1$ large real roots $\lambda$, and  an antistring at $i\pi/2$ composed of $p$ roots $\lambda$ and $p-1$ roots $\mu$.

\subsection{The $osp(4|2)$ equations \label{sec:osp42deg}}
To get the limit $q\to i$, let us set $\gamma=\pi/2-\epsilon$, and denote $\epsilon\xi_i$ the (real) center of the $L/2-m$ strings. Note that the antistrings disappear in this limit, so that all the states $(n,m,p)$ are the same for different $p$'s. The product of the first Bethe equations for $\lambda_i$ and 
its conjugate $\lambda_i^*$ gives 
\begin{equation}
\begin{aligned}
&\left( \frac{\sinh(\epsilon\xi_i +i (\pi/2-\epsilon/2))}{\sinh(\epsilon\xi_i+i\epsilon/2)} \frac{\sinh(\epsilon\xi_i +i (-\epsilon/2))}{\sinh(\epsilon\xi_i+i(-\pi/2+\epsilon/2))} \right)^L\\
&=\prod_{j} \frac{\sinh(\epsilon\xi_i-\epsilon\xi_j+i\pi/2)}{\sinh(\epsilon\xi_i-\epsilon\xi_j+i\epsilon)} \frac{\sinh(\epsilon\xi_i-\epsilon\xi_j+i(\pi/2-\epsilon))}{\sinh(\epsilon\xi_i-\epsilon\xi_j+i(-\pi/2+\epsilon))}  \frac{\sinh(\epsilon\xi_i-\epsilon\xi_j+-i\epsilon)}{\sinh(\epsilon\xi_i-\epsilon\xi_j-i\pi/2)} \,,
\end{aligned}
\end{equation}
which is, in the limit $\epsilon\to 0$, exactly the square of the Bethe equations for $su(2)$. Thus we have
\begin{equation}
\label{BExxx}
\left( \frac{\xi_i+i/2}{\xi_i-i/2}\right)^L=\pm\prod_{j\neq i}\frac{\xi_i-\xi_j+i}{\xi_i-\xi_j-i}\,,
\end{equation}
where we note the $\pm 1$. These equations are always true for the strings, whether there are real roots $\lambda$ or not. In case of isolated $\lambda_i$ roots, that we denote $\theta_i$ in the limit $\epsilon\to 0$, they have to satisfy the equation
\begin{equation}
\left( \frac{\sinh(\theta_i+i\pi/4)}{\sinh(\theta_i-i\pi/4)}\right)^L=1\,,
\end{equation}
hence
\begin{equation}
\Theta_i\equiv \frac{\sinh(\theta_i+i\pi/4)}{\sinh(\theta_i-i\pi/4)}=e^{2ik\pi/L}\,,
\end{equation}
for an integer $k$. Moreover, if there are strings, the second Bethe equation give another constraint
\begin{equation}
\prod_j \Theta_j^2=1\,.
\end{equation}
The eigenvalue of the transfer matrix becomes for $L/2-m=M$ strings
\begin{equation}
\Lambda(\lambda)=(-1)^L(-1)^M\left((\lambda+i)^L\prod_j \frac{\lambda-\xi_j-i/2}{\lambda-\xi_j+i/2}\prod_j \Theta_j +\lambda^L \prod \frac{\lambda-\xi_j+3i/2}{\lambda-\xi_j+i/2}\prod_j \Theta_j^{-1} \right)\,.
\end{equation}
We see that the eigenvalues obtained are either $su(2)$ eigenvalues, or $(-1)$ times $su(2)$ eigenvalues, except for the pseudo-vacuum which can modified (almost multiplied, up to an exponentially small term in $L$) by any root of unity. See appendix \ref{sec:osp42} for the numerical verification of these observations by direct diagonalization of the transfer matrix in small sizes.

In conclusion, it seems unfortunately that the only integrable $OSp(4|2)$ model accessible to us corresponds to a very special point on the sigma model critical line, where the underlying $O(2)$ theory is in fact at the $SU(2)$ invariant point. This corresponds to the special value $w=0$ in the dense loop soup, where loops are in fact not allowed to cross. The exponents are exactly the same as those of the level-one $SU(2)$ WZW theory, only degeneracies are different. This could of course have been expected a priori, since the $R$ matrix has the same abstract  Temperley-Lieb form as the $R$ matrix for the 6-vertex model at $\Delta=-1$. The only difference is that we have here a $36\times 36$ representation of the generator $E_i$, as opposed to a $4\times 4$ one for the 6-vertex model.

\section{Physical properties of dense loop soups with crossings}
This section is devoted to the application of the previous energy calculations to  ``watermelon'' exponents in loop soups with crossings.

We begin with some results on the $osp(r|2s)$ spin chains by showing that they describe intersecting loop soups with loop weight $N=r-2s$ \cite{martinsnienhuisrietman}. We then show for the first time that the spectrum of the $osp(r|2s)$ model is exactly included---in finite-size---in the spectrum of all the $osp(r+p|2s+p)$ models, as observed but not understood nor proved in \cite{FM}. We finally establish a correspondance between sectors of fixed charges and specific properties of loop configurations. This enables us to compute some watermelon 2-point functions that exhibits logarithmic behavior, which is the main new result of this section.

\subsection{A model for loops with crossings}


Let us first explain why the $osp(r|2s)$ vertex model can be reformulated as  a model for intersecting loops with weight $N=r-2s$. The mere observation that the $R$-matrix \eqref{Rmatrix} is built from elements of the Brauer algebra \cite{martinsnienhuisrietman,nienhuisrietman} is a bit unsatisfactory, since it does not explain how to treat the boundary conditions and the special weight that comes with them, and also because in this context the role of the graded tensor product is not clear.

In this section we prove the equivalence of the $osp(r|2s)$ model with a model of intersecting loops, starting directly from the expression of the transfer matrix \eqref{tm}, where from \eqref{Rmatrix}
\begin{equation}
\label{eq:r}
R(\lambda)^{lj}_{ik}=\lambda \delta_{ij}\delta_{kl}+(-1)^{p_i p_j}\delta_{il}\delta_{jk}+\frac{2\lambda}{2-r+2s-2\lambda}(-1)^{i>r+s}(-1)^{j\leq s} \delta_{ik'}\delta_{jl'}\,.
\end{equation}
We recall the definition of the conjugate index, $i'=D+1-i$, for any $i=1,\ldots,D$ with $D=r+2s$. We will work at constant spectral parameter $\lambda$ and will omit the dependence on $\lambda$ in order to simplify the notations. Define the partition function $Z$ of this model on an $L\times M$ lattice as
\begin{equation}
\label{eq:partitionfunction}
Z=\tr (t(\lambda)^M)\,,
\end{equation}
where $t(\lambda)$ is the transfer matrix given by \eqref{eq:T}--\eqref{tm}.
 It is thus
\begin{equation}
\begin{aligned}
Z=\sum_{c_\cdot^\cdot, \alpha_\cdot^\cdot =1}^D \prod_{m=1}^M R^{\alpha_L^{m+1} c_L^m}_{c_1^m \alpha_L^{m}}R^{\alpha_{L-1}^{m+1} c_{L-1}^m}_{c_L^m \alpha_{L-1}^{m}}...R^{\alpha_1^{m+1} c_1^m}_{c_2^m \alpha^{m}_1}
 (-1)^{p_{c_1^m}}(-1)^{\sum_{j=2}^L(p_{\alpha_j^m}+p_{\alpha^{m+1}_j})\sum_{i=1}^{j-1}p_{\alpha^m_i}}\,,
\end{aligned}
\end{equation}
with the identification $M+1\equiv 1$ due to periodic boundary conditions. As usual, each $R^{lj}_{ik}$ can be represented as the intersection of four lines, with the upper line carrying the index $\alpha_i^{m+1}$, the bottom line $\alpha_i^{m}$, the right line $c_i^m$ and the left line $c_{i+1}^m$:
\[\begin{tikzpicture}
 \draw[thick] (0,0)--(1.5,0);
 \draw[thick] (0.75,-0.75)--(0.75,0.75);
 \draw (0,0) node[left] {$c_{i+1}^m$};
 \draw (1.5,0) node[right] {$c_{i}^m$};
 \draw (0.75,-0.75) node[below] {$\alpha_i^m$};
 \draw (0.75,0.75) node[above] {$\alpha_i^{m+1}$};
\end{tikzpicture}
\]

 The $R$-matrix \eqref{Rmatrix} is a sum of three terms that impose 
 $(\alpha_i^m,c_{i+1}^m)=(c_i^m,\alpha^{m+1}_i)$, or $(\alpha_i^m,c_i^m)=((c_{i+1}^m)',(\alpha_i^{m+1})')$, or $(\alpha_i^m,c_i^m)=(\alpha_i^{m+1},c_{i+1}^m)$. 
 These terms are generators in the Brauer algebra \cite{martinsnienhuisrietman,nienhuisrietman} and can be represented diagrammatically as in Figure~\ref{fig:loopsoup}. The first two diagrams consist of a pair of ``corners'' that we shall henceforth refer to as North-East (NE), South-West (SW),
 South-East (SE) and North-West (NW), as indicated in the figure. The third diagram realizes loop crossings and corresponds algebraically to the (graded) permutation operator.
  \begin{figure}[H]
 \begin{center}

\begin{tikzpicture}
 \draw[thick] (0,0)--(0.5,0) arc(270:360:0.25) -- (0.75,0.75);
 \draw[thick] (0.75,-0.75)--(0.75,-0.25) arc(180:90:0.25) -- (1.5,0);
 \draw (0.75,-0.75) node[below] {$I$};
 \draw (0.35,0.35) node {\scriptsize SE};
 \draw (1.15,-0.35) node {\scriptsize NW};
\begin{scope}[xshift=3cm]
 \draw[thick] (0,0)--(0.5,0) arc(90:0:0.25) -- (0.75,-0.75);
 \draw[thick] (0.75,0.75)--(0.75,0.25) arc(180:270:0.25) -- (1.5,0);
 \draw (0.75,-0.75) node[below] {$E$};
 \draw (0.35,-0.35) node {\scriptsize NE};
 \draw (1.15,0.35) node {\scriptsize SW};
\end{scope}
\begin{scope}[xshift=6cm]
 \draw[thick] (0,0)--(1.5,0);
 \draw[thick] (0.75,-0.75)--(0.75,0.75);
 \draw (0.75,-0.75) node[below] {$X$};
\end{scope}
\end{tikzpicture}

\end{center}
\caption{Graphical representation of the three possible ways to match the four indices, corresponding to the three generators of \eqref{Rmatrix} (where the first term is $X$, the second one $I$ and the last one $E$). The first two diagrams define four types of corners denoted NE, SW, SE and NW. }
\label{fig:loopsoupne}
\end{figure}
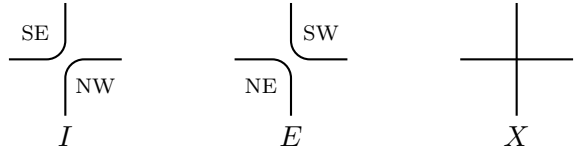
The graphical representation of Figure~\ref{fig:loopsoupne} naturally induces a representation of the partition function as a sum over dense intersecting loops, each vertical (resp. horizontal) edge carrying an index $\alpha$ (resp. $c$). In this representation $L$ is the horizontal length, and $M$ the vertical height of the $L\times M$ lattice.

\medskip

There are three issues to be resolved in order to define a proper model of intersecting loops:
\begin{enumerate}
 \item There are all the fermionic signs that seem to weigh each configuration with an arbitrary sign.
 \item The loops are not all equivalent since they carry an index $i=1,\ldots,D$ that must eventually be summed over.
  Moreover this index changes to its conjugate value along a loop at the SE
  and NW corners. 
 \item The weight of each fixed-index loop is one, but the proper weight after summation over the index will have to be worked out by
 taking carefully into account the boundary conditions.
\end{enumerate}
All these issues are of course related, and analysing them properly will lead to the resolution of the problem. 
The crux resides in a proper understanding of the fermionic signs.
This relies on the following lemma,
the proof of which is relegated to Appendix~\ref{app:lemma}.

\begin{lemma}
\label{lemma}
If the index $a$ of a loop in a configuration is changed so that $p_a$ changes, 
all other things being equal, then the weight of the configuration is multiplied by $(-1)^{b_v+1}$ where $b_v$ is the number of times the loop crosses the top and bottom boundaries (i.e., those corresponding to the direction of length $M$).
\end{lemma}

\begin{proof}
See Appendix~\ref{app:lemma}.
\end{proof}

One sees that the number of times a boundary is crossed in the vertical or horizontal directions plays a different role. We will say that a loop is \textit{non-contractible} in the vertical direction (or simply non-contractible loop) if it crosses the whole lattice in the vertical direction, i.e. if it is possible to follow the loop from top to bottom without crossing the vertical boundary conditions. We say that a loop is contractible if it is not non-contractible. We will denote by \textit{even/odd non-contractible} a non-contractible loop that crosses the top and bottom boundaries an even/odd number of times (without saying anything about the right and left boundaries). In the subsequent subsections we will also represent by a diagram like \tiktwo{0}{0.5} or \tiktwo{0.5}{0} the sum of all configurations of loops which possess a certain number of non-contractible loops linked from top to bottom as indicated by the diagram. We refer the reader to Figure \ref{fig:exampleloops} for illustrative examples on a $2\times 2$ lattice with periodic boundary conditions.

From this lemma comes the theorem:

\begin{theorem}[Intersecting loop soup model]
$Z=\tr (t^M)$  is the partition function for a model of intersecting loops with loop weight $r-2s$ for contractible loops and $r\pm 2s$ for odd/even non-contractible loops.
\end{theorem}

\begin{proof}
If a configuration contains only loops with bosonic indices, all the $(-1)^{p_a}$ are $+1$ and the weight for a loop (with an index) is $+1$. But if a loop $l$ is contractible, because of lemma \ref{lemma} its weight is $1$ if it is bosonic and $-1$ if it is fermionic, thus after summation over the indices the weight is $r-2s$. If the loop is non-contractible the 
fermionic weight is $\pm 1$ if it is odd/even, thus a weight $r\pm 2s$ after summation.
\end{proof}

We note that if the transfer matrix were defined as the trace (and not the supertrace) of the monodromy matrix, i.e. if in \eqref{eq:T} there were no $(-1)^{p_i}$, then one would have a weight $r+2s$ for the odd non-contractible loops in the horizontal direction as well. On the contrary, to give the same weight $r-2s$ to all the loops (contractible or not), one would need to modify the trace in \eqref{eq:partitionfunction} and define $Z=\tr (K t^M)$ with $K$ a matrix that will assign the desired weights according to the sector. 

Note finally that this discussion is reminiscent of  the problem of dimer covering on the torus \cite{kasteleyn}, and also of variants of Kirchhoff's theorem for modified Laplacians, see \cite{saleurpolymer}. 

\begin{figure}[H]
\begin{center}
 \includegraphics[scale=0.3]{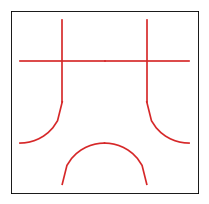}\hspace{1cm} \includegraphics[scale=0.3]{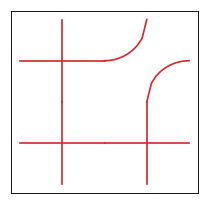} \hspace{1cm}\includegraphics[scale=0.3]{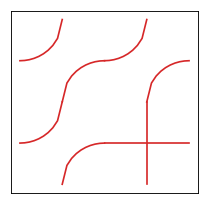}
\end{center}
\caption[blabla]{First figure: two contractible loops. Indeed, you cannot follow any of the two loops from top to bottom without crossing the vertical periodic boundary. Second figure: two odd non-contractible loops, and one contractible loop. Indeed, the horizontal line at the bottom is a contractible loop, and the two other loops are non-contractible and cross the vertical periodic boundary exactly once. Third figure: one even non-contractible loop. Indeed, there is only one loop that crosses the vertical periodic boundary twice. 

If we now denote with a diagram like $\tiktwo{0}{0.5}$ the connections between the four beginnings of strands $\bullet$ at the top and the bottom of the lattice, then the strands in the first figure are connected like $\tikc$, in the second figure like $\tiktwo{0}{0.5}$, and in the last figure like $\tiktwo{0.5}{0}$. 
}
\label{fig:exampleloops}
\end{figure}

\subsection{Inclusion of $osp$ spectra}
To prove the inclusion of the spectra for the $osp$ chains, another lemma is needed:
\begin{lemma}
 \label{lemma2}
If $A$ and $M$ are square matrices of size $n$ and $n+m$ such that
\begin{equation}
\label{eq:condition}
\forall k\in \mathbb{N}\,,\quad  \forall i,j \in \{1,...,n\}\,,\quad (M^k)_{i,j}=(A^k)_{i,j}\,,
\end{equation}
then the spectrum of $A$ is included in the spectrum of $M$ (with the degeneracies).
\end{lemma}

\begin{proof}
Writing $M$ in block form
\begin{equation}
M=\left(\begin{matrix}
A&B\\
C&D
\end{matrix}
 \right)\,,
\end{equation}
the condition \eqref{eq:condition} implies
\begin{equation}
\label{eq:condi2}
\forall k\in \mathbb{N},\quad BD^kC=0\,.
\end{equation}
Let now $\lambda\in\mathbb{C}$ not in the spectrum of $D$. Since $D-\lambda I_m$ is invertible one can use Schur's complement to write
\begin{equation}
\det (M-\lambda I_{n+m})=\det(A-\lambda I_n - B(D-\lambda I_m)^{-1}C) \det(D-\lambda I_m)\,.
\end{equation}
From Cayley-Hamilton theorem, $(D-\lambda I_m)^{-1}$ is a polynomial in $(D-\lambda I_m)$, thus in $D$, so that with \eqref{eq:condi2} we have $B(D-\lambda I_m)^{-1}C=0$. It follows that
\begin{equation}
\det(M-\lambda I_{n+m})=\det (A-\lambda I_n) \det(D-\lambda I_m)\,.
\end{equation}
Since the function $\lambda\to \det(M-\lambda I_{n+m})$ is continuous in $\lambda$ and since the spectrum of $D$ is finite, the previous equation is true for all $\lambda\in\mathbb{C}$. Thus whenever $\lambda$ is an eigenvalue of $A$, $\det(M-\lambda I_{n+m})=0$ and it is also an eigenvalue of $M$. Moreover since $\det(D-\lambda I_m)$ is a polynomial in $\lambda$ there cannot be poles and the eigenvalues of $A$ in the spectrum of $M$ have at least the degeneracies they have in the spectrum of $A$. (However in general the eigenvectors of $M$ corresponding to the eigenvalues of $A$ cannot be expressed simply: in particular they may have non-zero $i$-th components for $i>n$.)
\end{proof}

One can now prove the theorem:

\begin{theorem}[Inclusion of spectra]
The spectrum of the $osp(r|2s)$ spin chain is included in finite size in the spectrum of the $osp(r+p|2s+p)$ spin chain for all even $p>0$.
\end{theorem}

\begin{proof}
Denote $t$ the transfer matrix of the $osp(r|2s)$ spin chain, and $T$ the one of the $osp(r+p|2s+p)$ spin chain. Let $J$ be a subset of $2p$ indices among which $p$ are bosonic and $p$ are fermionic, and $I=\{1,...,r+2s+2p\} \setminus J$. The indices of $t$ are identified with $I$.

$(T^M)^{\alpha_1^{M+1}...\alpha_L^{M+1}}_{\alpha_1^1...\alpha_L^1}$ is the partition function of the model on a $L\times M$ lattice with fixed boundary conditions at the top and bottom boundaries. Inside the configuration, every loop whose index is in $J$ has to be contractible, since at the up and down boundaries the indices must be in $I$. As $J$ contains as many bosonic as fermionic indices, lemma \ref{lemma} implies that these configurations add up to zero. Therefore all the loops can be considered having their indices in $I$, which is exactly $(t^M)^{\alpha_1^{M+1}...\alpha_L^{M+1}}_{\alpha_1^1...\alpha_L^1}$. Then lemma \ref{lemma2} applies and proves the theorem.
\end{proof}

Note that taking the supertrace of the monodromy matrix is crucial to have this property. Otherwise the non-contractible loops in the horizontal direction would not cancel out. Notice also that integrability does not play any role here, so it is true for arbitrary weights in the $R$-matrix. 
Recall that such an inclusion is observed for $gl(r|s)$ models as well \cite{canduglnm}.

\subsection{Charges and loop configurations \label{sec:chargesloops}}

%

While most of our discussion about critical exponents has been based on studies of the integrable Hamiltonian, it is usually the case that the same universal properties would be obtained by focussing instead on the transfer matrix. Indeed, taking the Hamiltonian limit amounts to taking the continuum limit in the (imaginary) time direction, something that is not supposed to modify the continuum description of the lattice model. The transfer matrix language is on the other hand more natural to describe loops, especially when the spectral parameter $\lambda=1$, corresponding to an {\sl isotropic} loop soup on the square lattice. We have checked that the log of the largest eigenvalues of the transfer matrix have the same behaviour as those of the Hamiltonian \eqref{scalegaps}, simply with the Fermi velocity $v_F$ replaced by a sound velocity $\sin (\lambda v_F)$, that is $1$ at the isotropic point.

This means that the finite-size corrections to the first excited states of the Hamiltonian, among which are the lowest eigenvalues in a sector imposing specific values of charges, correspond in the transfer matrix point of view to finite-size corrections to the largest eigenvalues of the transfer matrix $t$ in a sub vector space with specific values of charges. We recall that the partition function in \eqref{eq:partitionfunction} is given by the trace of the $M$-th power of $t=t(\lambda)$ the transfer matrix, that is dominated when $M\to\infty$ by the largest eigenvalue of $t$. Similarly, the trace of the $M$-th power of $t$ over a sub vector space where the charges take specific values, is dominated by the largest eigenvalue of $t$ in this sub vector space.

The question is now to understand the kind of constraint that is imposed on the intersecting loop soup when this trace over a sub vector space where the charges take specific values is performed.

Let us treat the case of $osp(2|2)$, the simplest example with a ``fermionic charge'' $J_z$ and a ``bosonic charge'' $Q_z$. In the grading given by \eqref{eq:grading}, they are represented by

\begin{equation}
2J_z=\left(\begin{matrix} 1 &0&0&0\\ 0&0&0&0\\0&0&0&0\\0&0&0&-1 \end{matrix} \right)\,,\quad 2Q_z=\left(\begin{matrix} 0&0&0&0\\ 0&1&0&0\\0&0&-1&0\\0&0&0&0 \end{matrix} \right)\,.
\end{equation}

When one traces over the vector space where $Q_z$ is equal to $q$, one considers only the configurations where at the bottom (and at the top) of the lattice, there are $2q$ more strands with index $2$ than strands with index $3$. As already said, along a loop an index $i$ is replaced by its conjugate $i'$ (i.e., $1 \leftrightarrow 4$, $2 \leftrightarrow 3$ ) every time a NW or SE corner is encountered. It implies that a strand carrying a $2$ at the bottom cannot directly (without crossing the vertical boundary) join another strand carrying a $2$ at the bottom. Then the $2q$ extra $2$'s at the bottom and at the top of the lattice have to be connected between themselves by going through the whole lattice in the vertical direction. Since the loops with bosonic index (whether contractible or not) are always given the same weight equal to $1$, it comes, after summing over the indices, that the boundary condition imposes to have (at least) $2q$ loops that propagate through the lattice in the vertical direction that are given weight $1$. Note that an extra strand with index $2$ at the bottom can be connected to any extra strand with index $2$ at the top, with the same weight $1$. For $2q=2$ these configurations are $\tiktwo{0}{0.5}+\tiktwo{0.5}{0}$.

If one traces over the vector space where $J_z$ is equal to $j$, the same reasoning shows that the configurations are constrained to have $2j$ more strands with index $1$ than index $4$ at the bottom, and that those at the top and the bottom of the lattice have to be connected between themselves. However since the index is fermionic, the weight of such a loop is $1$ (resp. $-1$) if it crosses the top and bottom boundary an odd (resp. even) number of times, from lemma \ref{lemma}. The total weight given to these loops is then exactly the signature of the permutation that maps the bottom $2j$ extra $1$'s to the top $2j$ extra $1$'s they are connected to. For $2j=2$ these configurations are $\tiktwo{0}{0.5}-\tiktwo{0.5}{0}$.

If one traces over the vector space where both $J_z$ and $Q_z$ are fixed as $j$ and $q$ respectively, then the configurations are constrained to possess $2j$ strands with index $1$ and $2q$ strands with index $2$ to propagate through the lattice from bottom to top. The total weight given to these loops is then the signature $\epsilon(\sigma_j)$ of the permutation  $\sigma_j$ that maps the bottom $2j$ extra $1$'s to the top $2j$ extra $1$'s they are connected to, without considering the $2q$ bosonic strands. Let us now denote $P_{L,M}(\sigma)$  the sum of all the configurations on a $L\times M$ lattice with $2j+2q$ non-contractible strands, where the $2j$ 'fermionic' strands are at the left of the lattice at row $1$, and where the $2j+2q$ strands are permuted by $\sigma$ at row $M$. Since a fermionic strand has to be connected to another fermionic strand through the periodic vertical boundary, this permutation has to be decomposable into $\sigma=\sigma_j\sigma_q$ where $\sigma_j$ acts trivially on the bosonic strands and $\sigma_q$ trivially on the fermionic strands. Then, denoting by $\tr_{j,q}$ the trace over the sector $J_z=j,Q_z=q$, one can express the trace of the $M$-th power of the transfer matrix $t$ as
\begin{equation}
\label{eq:youngsym}
\tr_{j,q} (t^M)=\sum_{\sigma_j,\sigma_q}\epsilon(\sigma_j)Z_{L,M}(\sigma_j\sigma_q)\,,
\end{equation}
where the sum runs over the permutations $\sigma_j$ and $\sigma_q$ of $2j+2q$ elements, that leave invariant the last $2q$ elements (respectively the first $2j$ elements). We recall that $\epsilon(\sigma_j)$ is the signature of the permutation $\sigma_j$. Here are some examples:
\begin{equation}
\begin{aligned}
\tr_{1,0} (t^M)&=Z_{LM}((1\,2))-Z_{LM}((2\,1))\\
\tr_{0,1} (t^M)&=Z_{LM}((1\,2))+Z_{LM}((2\,1))\\
\tr_{1,1} (t^M)&=Z_{LM}((1\,2\,3\,4))-Z_{LM}((2\,1\,3\,4))+Z_{LM}((1\,2\,4\,3))-Z_{LM}((2\,1\,4\,3))
\end{aligned}
\end{equation}
where we write $(i_1\,...\,i_n)$ the permutation that maps $1$ onto $i_1$, etc, $n$ onto $i_n$. The configuration of strands connections to which these three traces correspond are respectively $\tiktwo{0}{0.5}-\tiktwo{0.5}{0}$, $\tiktwo{0}{0.5}+\tiktwo{0.5}{0}$ and $\tikfour{0}{0.5}{1}{1.5}-\tikfour{0.5}{0}{1}{1.5}+\tikfour{0}{0.5}{1.5}{1}-\tikfour{0.5}{0}{1.5}{1}$.

Equation \eqref{eq:youngsym} is reminiscent of the Young symmetrizer for the Young tableau
\begin{equation}
\begin{ytableau}
   \none &  1& \none[\dots] &2q \\
     1\\
     \none[\dots]\\
     2j\\
  \end{ytableau}
  \label{supertableau}
\end{equation}
which here takes the form of a ``supertableau'' applying independently a symmetrizer to the $2q$ bosonic strands and an antisymmetrizer to the $2j$ fermionic strands. Compared to the usual Young supertableux, e.g. in \cite{gourdin}, there is an empty box at the top left merely because we shifted the first row, to make explicit the fact that each box must be counted either in the column or in the row. 

\ytableausetup{smalltableaux}
\ytableausetup{aligntableaux=center}

An unpleasant aspect of formula \eqref{eq:youngsym} is that it depends on the position of the ``fermionic'' strands, whereas we would like to have a geometrical meaning for strands without specifying their ``bosonic'' or ``fermionic'' nature. This important issue will be addressed in the next subsection.


These considerations can be generalized without difficulties to $osp(r|2s)$ for arbitrary $r$ and $s$. There is only an additional important remark to make on the case $r$ odd. Indeed in this case there is an index $i$ which is not associated to any charge, for example for $osp(1|2)$ the index $2$ does not affect the charge $J_z$. Then the number of extra strands associated to the charges can be odd (whereas in the case $r$ even it is necessarily even): in this case for even $L$ there will 
be one extra strand with this index $i$ that acts as a bosonic strand with a weight $1$. For odd $L$ the same observation holds for an even number of strands associated to the charges.

\subsection{Transfer matrix eigenvalues and loop configurations \label{sec:eigenvalues}}

On an $L\times M$ lattice with $M\to\infty$ the trace of the $M$-th power of the transfer matrix in size $L$ over the vector space with given charges behaves as $\propto \lambda_1^M$ where $\lambda_1$ is the maximal eigenvalue of the transfer matrix in this sector. Denoting by $\lambda_0$ the maximal eigenvalue of the transfer matrix, the quantity $(\log \lambda_1-\log \lambda_0)^{-1}$ gives the correlation length on an infinite cylinder of circumference $L$ for the property of the configurations induced by the sector of $\lambda_1$. 

In the limit $M\to\infty$ some remarks have to be made on \eqref{eq:youngsym}. On an infinite cylinder there is no periodic boundary conditions to impose that a bosonic (resp.\ fermionic) strand falls back on a bosonic (resp.\ fermionic) strand, i.e. that bosonic and fermionic strands are permuted among themselves after a certain number of applications of the transfer matrix. In \eqref{eq:youngsym} for $M$ large, imposing the decomposition $\sigma=\sigma_j\sigma_q$ instead of taking a generic permutation $\sigma$ only changes a multiplicative factor that is independent of $M$, and thus does not affect the free energy that is in both cases $-\log \lambda_1$. Any permutation $\sigma$ should be possible in \eqref{eq:youngsym} in this $M \to \infty$ limit, not only those that can be decomposed into $\sigma_j\sigma_q$. Thus on an infinite cylinder we have
\begin{equation}
\label{eq:youngnew}
e^{-MF_M(j,q)}=\sum_{\sigma\in\mathfrak{S}_{2j+2q}}\epsilon_{2j}(\sigma)\tilde{Z}_{L,M}(\sigma)\,,
\end{equation}
with $F_M(j,q)\to -\log\lambda_1$ when $M\to\infty$, and where $\tilde{Z}_{L,M}(\sigma)$ is the sum of all the configurations where $2j+2q$ strands (with the $2j$ fermionic ones at the left at row $1$) are permuted by $\sigma$ after $M$ rows, on a $L\times M$ lattice \textit{without} periodic boundary condition in the $M$ direction. $\epsilon_{2j}(\sigma)$ is the 'partial signature' of the first $2j$ elements of $\sigma$, i.e., 
attributes a factor $-1$ to each $(i_1,i_2)$ with $i_1 < i_2 \leq 2j$ such that $\sigma(i_1)>\sigma(i_2)$. The sum now runs over all the permutations $\sigma$ of $2j+2q$ elements.

For example, for $j=1$, $q=1/2$ there are three strands, with $2$ 'fermionic' strands at the left. It gives the configurations $\tikthree{0}{0.5}{1}- \tikthree{0.5}{0}{1}+\tikthree{0}{1}{0.5}+\tikthree{0.5}{1}{0}-\tikthree{1}{0}{0.5}-\tikthree{1}{0.5}{0}$. 

The advantage of \eqref{eq:youngnew} 
is that although the summands still depend on the initial position of the fermionic strands, it can be easily transformed into a version that does not distinguish between bosonic and fermionic strands, by summing over 
all ${2j+2q \choose 2j}$ ways of attributing $2j$ fermionic and $2q$ bosonic labels to the $2j+2q$ strands. The fermionic signs are then attributed as before. Thus one gets
\begin{equation}
\label{eq:youngnew1}
e^{-M\tilde{F}_M(j,q)}=\sum_{\sigma\in\mathfrak{S}_{2j+2q}}\left({2j+2q\choose 2j}-2\iota_{2j}(\sigma) \right)\tilde{Z}_{L,M}(\sigma)\,,
\end{equation}
with $\tilde{F}_M(j,q)\to -\log\lambda_1$ when $M\to\infty$, and where $\iota_{2j}(\sigma)$ is the number of subsets of $2j$ strands among the $2j+2q$ strands permuted by $\sigma$, that intersect between themselves an odd number of times. A formal mathematical definition of $\iota_{2j}(\sigma)$ is
\begin{equation}
\label{eq:iota}
\iota_{2j}(\sigma)=\#\left\{I\subset \{1,...,2j+2q\}\,\, \text{such that }\,\, \# I=2j\,,\text{and } \#\{(i,j) \in I^2\,\,,\,\, i<j\,, \sigma(i)>\sigma(j) \} \text{ is odd} \right\}
\end{equation}
Note that with this definition \eqref{eq:youngnew1} no longer refers to bosonic/fermionic strands, but simply to a total number of $2j+2q$ unspecified strands. Clearly, $\iota_0(\sigma)=\iota_1(\sigma)=0$. Moreover, $\iota_2(\sigma)$ is exactly the total number of intersections between the strands (in a graphical representation where two strands intersect $0$ or $1$ time and do not wind around the horizontal periodic boundary). And $\iota_3(\sigma)$ is the total number of intersections between the strands, where each intersection between strands $i<j$ is weighted by $2j+2q-|i-j|-|\sigma(i)-\sigma(j)|$.

For instance, one has the following correspondances on the \textit{infinite} cylinder between the Young tableaux and the loop configurations:
\begin{equation}
\begin{aligned}
\ytableaushort{123} \longrightarrow \tikthree{0}{0.5}{1}+ \tikthree{0.5}{0}{1}+\tikthree{0}{1}{0.5}+\tikthree{0.5}{1}{0}+\tikthree{1}{0}{0.5}+\tikthree{1}{0.5}{0}\qquad\qquad (j=0\,,q=3/2)\\
\ytableaushort{1,2,3} \longrightarrow \tikthree{0}{0.5}{1}- \tikthree{0.5}{0}{1}-\tikthree{0}{1}{0.5}+\tikthree{0.5}{1}{0}+\tikthree{1}{0}{0.5}-\tikthree{1}{0.5}{0}\qquad\qquad (j=3/2\,,q=0)\\
\ytableaushort{\none 1,1,2} \longrightarrow 3\tikthree{0}{0.5}{1}+\tikthree{0.5}{0}{1}+\tikthree{0}{1}{0.5}-\tikthree{0.5}{1}{0}-\tikthree{1}{0}{0.5}-3\tikthree{1}{0.5}{0}\qquad\qquad (j=1\,,q=1/2)\\
\end{aligned}
\end{equation}
These combinations now refers to three generic strands, without the need to specify which ones are fermionic, contrarily to \eqref{eq:youngnew} where the different strands are either fermionic or bosonic.

Note that due to the horizontal periodic boundary, there can be a permutation between two strands without crossings. The fermionic signs also count these situations.\\

When the circumference of the cylinder itself becomes large, a conformal transformation onto the plane gives access to the critical exponents of the corresponding watermelon exponents on the plane. To use the Bethe ansatz to compute these corrections in large sizes $L$, one needs to find which eigenvalue is maximal for each sector. A possible source of difficulty is that the sector associated to other conserved quantities in which this $\lambda_1$ lies may change with the size of the system. For example for $osp(1|2)$ the state with integer magnetization $j>0$ with minimal energy in the thermodynamic limit does not have symmetric Bethe roots, but in small sizes nothing prevents the state with equal magnetization but with symmetric Bethe roots from having lower energy. The determination of the finite-size corrections to these states close to the thermodynamic limit nevertheless permits to determine which one is the lowest.

In the following, we explicitly check the correspondance between the constraint encoded by a tableau \begin{ytableau}
   \none &  1& \none[\scriptstyle \dots] &2q \\
     \scriptstyle 1\\
     \none[\scriptstyle \dots]\\
     \scriptstyle 2j\\
  \end{ytableau} and the eigenvalue of the transfer matrix, with a numerical code for loops with crossings. 
  That is, we start from an eigenvalue of the $osp$ transfer matrix, find its Bethe roots, deduce the corresponding charges from them, and then compare to a numerical transfer matrix that implements \eqref{eq:youngnew} on an intersecting loop soup with fermionic/bosonic strands, and \eqref{eq:youngnew1} on a generic intersecting loop soup, and verify that it gives exactly the same eigenvalue. 

\subsubsection{$osp(2|2)$}
In Table \ref{tab:osp22}
we give the explicit correspondance between some eigenvalues of the transfer matrix of the $osp(2|2)$ model in size $L=6$, and the Bethe roots and the kind of constraints that it imposes on the loop configurations.
\ytableausetup{smalltableaux}
\ytableausetup{aligntableaux=center}
\begin{table}[H]
\begin{center}
\begin{tabular}{|c|c|c|c|}
\hline
Eigenvalue & Bethe roots $\{\lambda\},\{ \mu \}$ & $J_z,Q_z$ & Constraint on the loops\\
\hline
\hline
167.295 & $\{ 0.02897,-0.02897,0\},\{0.0180,-0.0180\}$ &$1/2,1/2$ & no constraint, or \ytableaushort{\none1,1 }\\
\hline
95.732 & $\{0.5992,-0.5992,0.1471,-0.1471\},\{0\}$ & $1/2,3/2$ & \ytableaushort{\none123,1 } \\
\hline
63.761 &$\{0.3281,-0.0286 \},\{0.3281,-0.0286 \}$ & $1,0$&  \ytableaushort{1,2}\\
\hline
44.140 &$\{-0.5774,-0.1355,0.1584\},\{ -0.1182\}$ & $1,1$ & \ytableaushort{\none12,1,2 }\\
\hline
29.63 &$\{ -0.147,0.147\},\{0 \}$ &$3/2,1/2$ &\ytableaushort{\none1,1,2,3 }\\
\hline
22.750 &$\{0.8660,-0.8660,0.2887,-0.2887,0 \},\{\}$ & $1/2,5/2$ & \ytableaushort{\none12345,1}\\
\hline
3.482 &$\{-0.1340 \},\{ -0.1340\}$ & $2,0$ & \ytableaushort{1,2,3,4}\\
\hline
\end{tabular}
\end{center}
\caption{Correspondance between some transfer matrix eigenvalues, Bethe roots, charges, and loop configurations for the $osp(2|2)$ case in size $L=6$ at the isotropic integrable point.}
\label{tab:osp22}
\end{table}






Note that the Bethe roots for integral charges are associated to non-symmetric states, given in \eqref{eq:osp22nonsym} with $\bar{m}=m-1$.
Note also that the equivalence between the absence of constraints and  $\ytableaushort{\none1,2}$ on the infinite cylinder is very specific to $osp(2|2)$ in even size where the weight of a loop is zero, since in both cases it amounts to forbidding the contraction between the two strands.

\subsubsection{$osp(3|2)$}
The same work can be done for $osp(3|2)$ where the weight for a (contractible) loop is $1$. Here we notice the importance of the remark of section \ref{sec:osp32multiplet}, that gives the true charges of the multiplet a Bethe vector belongs to, and that has direct consequences on the configurations of loops associated to it. The $J_z$ and $Q_z$ indicated in Table \ref{tab:osp32} are those of the highest-weight of the multiplet the Bethe vector belongs to (note that with the conventions of $osp(3|2)$, $Q_z=q$ imposes $q$ bosonic strands and not $2q$).
\begin{table}[H]
\begin{center}
\begin{tabular}{|c|c|c|c|}
\hline
Eigenvalue & Bethe roots $\{\lambda\},\{ \mu \}$ & $J_z,Q_z$ & Constraint on the loops\\
\hline
\hline
656.84 & degenerate roots  &$0,0$ & no constraint\\
\hline
584.97 & $\begin{matrix}\{-0.14\pm 0.75 i ,-0.14\pm 0.75 i ,\pm 0.75 i \},\\\{-0.14\pm 0.25 i ,-0.14\pm 0.25 i ,\pm 0.25 i  \}\end{matrix}$ & $1/2,1$ & \ytableaushort{\none1,1 }\\
\hline
 323.40 &  $\begin{matrix}\{ 0.0645 \pm  0.7501 i ,-0.0645 \pm  0.7501 i\},\\\{0.0645 \pm 0.249 i,-0.0645 \pm 0.249 i \}\end{matrix}$ & $1/2,3$ &\ytableaushort{\none123,1 }\\
\hline
175.96 & $\begin{matrix}\{-0.555,-0.057\pm 0.749 i, 0.057 \pm 0.749i ,0.555 \},\\ \{ -0.057 \pm 0.249 i,0.057 \pm 0.249 i\}\end{matrix}$ &$3/2,1$ & \ytableaushort{\none1,1,2,3 }\\
\hline
 100.40 &  $\begin{matrix}\{\pm 0.7500i \},\\\{ \pm 0.2499i\}\end{matrix}$ &$1/2,5$ & \ytableaushort{\none12345,1 }\\
 \hline
 67.27 & $\begin{matrix}\{-0.883,\pm 0.7500 i,0.883 \},\\\{ \pm 0.2499 i\}\end{matrix}$ &$3/2,3$ &\ytableaushort{\none123,1,2,3 }\\
 \hline
 19.39 & $\begin{matrix}\{ -0.883,-0.308,\pm 0.7500 i,0.308,0.883\},\\\{\pm 0.2499 i \}\end{matrix}$ & $5/2,1$&\ytableaushort{\none1,1,2,3,4,5 }\\
\hline
\end{tabular}
\end{center}
\caption{Correspondance between some transfer matrix eigenvalues, Bethe roots, charges, and loop configurations for the $osp(3|2)$ case in size $L=8$ at the isotropic integrable point.}
\label{tab:osp32}
\end{table}



\subsubsection{$osp(1|2)$}
The same exercice for $osp(1|2)$ is a bit formal since it gives a weight $-1$ to each contractible loop, that cannot be interpreted as a probability. However the correspondance between eigenvalues of the transfer matrix and specific configurations of loops still holds, see Table \ref{tab:osp12}.

\begin{table}[H]
\begin{center}
\begin{tabular}{|c|c|c|c|}
\hline
Eigenvalue & Bethe roots $\{\lambda\}$ & $J_z$ & Constraint on the loops\\
\hline
\hline
254.23 &$\{-0.558,-0.227,0,0.227,0.558 \}$ & $1/2$ & 
  \begin{ytableau}
   \none & *(lightgray) 1 \\
   1
  \end{ytableau}\\
\hline
225.95 & $\{-0.449,-0.126,0.5 i, -0.5 i, 0.126, 0.449 \}$ &$0$ & no constraint \\
\hline
95.23 &  $\{-0.616,-0.262,-0.028,0.200 \}$& $1$ &\ytableaushort{1,2}\\
\hline
44.54 &$\{-0.233,0,0.233 \}$ & $3/2$ &
  \begin{ytableau}
   \none & *(lightgray) 1 \\
   1\\
   2\\
   3
  \end{ytableau}\\
\hline
7.59 & $\{-0.265,-0.018 \}$ & $2$ &\ytableaushort{1,2,3,4}\\
\hline
\end{tabular}
\end{center}
\caption{Correspondance between some transfer matrix eigenvalues, Bethe roots, charges, and loop configurations for the $osp(1|2)$ case in size $L=6$ at the isotropic integrable point.}
\label{tab:osp12}
\end{table}
Tableaux with odd number of boxes would appear for odd sizes only. The fact that there is one ``bosonic'' strand for half-integer spin (denoted by a grey box) is explained in the last paragraph of section \ref{sec:chargesloops}.

Imposing a constraint can increase the ``partition function'' only because some Boltzmann weights are negative in the $osp(1|2)$ case. Remark also that since there is only one fermionic charge in $osp(1|2)$, one cannot get configurations like $\ytableaushort{1234 }$  and we are almost restricted to purely determinant-like combinations of probabilities. This is reminiscent of the correlation functions for the spanning trees and forests model that also exhibits an $osp(1|2)$ symmetry \cite{ivashkevich,caracciolojacobsensaleur}. However these combinations should appear in the $osp(3|4)$ model.

\subsection{Watermelon $2$-point functions for loops with crossings \label{sec:watermelon}}

We here collect our results for the logarithmic scaling of two-point functions in models of intersecting loops.

In the geometry of the plane in the scaling limit, for a permutation $\sigma\in\mathfrak{S}_{n+m}$ of $n+m$ elements, we denote by $P^{n+m}_\sigma(x)$ the probability that $n+m$ strands emanate from some small neighborhood and come close together again in another small neighborhood, separated from the first one by a distance $x$, with their ordering having been permuted by $\sigma$ in-between the two neighborhoods, see Figure \ref{fig:watermelon}.

\begin{figure}[H]
\begin{center}
\begin{tikzpicture}
\draw [line width=0.5mm, black] plot [smooth, tension=-1] (0,0) to[out=45,in=180-20]  (1.5,1) to[out=-20,in=180+40]  (4,-1) to[out=40,in=180+80] (6,0);
\draw [line width=0.5mm, black] plot [smooth, tension=-1] (0,0) to[out=20,in=180+20]  (2,0.5) to[out=20,in=180-70] (6,0);
\draw [line width=0.5mm, black] plot [smooth, tension=-1] (0,0) to[out=-10,in=180+20]  (2,-0.5) to[out=20,in=180-20]  (4,0.25) to[out=-20,in=180+30] (6,0);
\draw [line width=0.5mm, black] plot [smooth, tension=-1] (0,0) to[out=-50,in=180-0]  (2,-1) to[out=-0,in=180-20]  (5.25,0.5) to[out=-20,in=180-30] (6,0);
 \node [black] at (0,0){\Huge $\bullet$};
 \node [black] at (6,0){\Huge $\bullet$};
 \node [white] at (0,0){\huge $\bullet$};
 \node [white] at (6,0){\huge $\bullet$};
\end{tikzpicture}
\caption{Example of a $4$-legs watermelon $2$ point function.}
\label{fig:watermelon}
\end{center}
\end{figure}
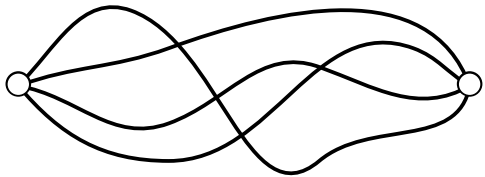

We recall that the fully-packed $O(n)$ loop model on the square lattice \cite{nienhuis82,nienhuis84,nienhuis87,saleurduplantier2}, i.e. the model consisting in filling a square lattice with the two tiles $P$ and $E$ in \eqref{fig:loopsoupne} and with a weight $N\equiv -2\cos\gamma$ for a loop is critical. We remind the reader that to make a connection between the loop configurations on the cylinder and $P^{n+m}_\sigma(x)$, we use conformal invariance to map the plane onto the cylinder, sending points $0$ and $x$ to $\mp \infty$. Then the configurations of loops on the plane that come from a neighbourhood of $0$ and meet again in a neighbourhood of $x$ exactly correspond on the cylinder to strands propagating all along the cylinder without forming loops, see Figure \ref{fig:looplane}. For example, this permits to show that the two-point function $P^2_\sigma(x)$ with $\sigma$ the identity decays as
\begin{equation}
P^2_\sigma(x)\sim \frac{1}{x^{(1-\gamma/\pi)-\tfrac{(\gamma/\pi)^2}{1-\gamma/\pi}}}
\end{equation}
 \begin{figure}[H]
 \begin{center}
\includegraphics[scale=0.6]{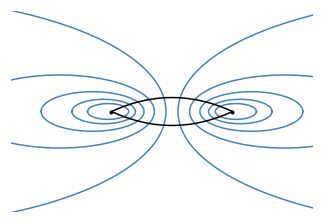} 
\end{center}
\caption{Two strands joining a neighbourhood of $0$ and a neighbourhood of $x$ (in black), together with contour lines (in blue) of a conformal transformation that maps the plane onto the cylinder, with points $0$ and $x$ mapped to $\mp\infty$. Each blue line corresponds to a line encircling the cylinder at a constant height after the mapping.}
\label{fig:looplane}
\end{figure}

In this section we use all our previous results to determine the generalization of these power-law decays to the intersecting loop models. The same reasoning can be used to translate an intersecting loops configurations on the plane such as Figure \ref{fig:watermelon} to configurations of loops on the cylinder, where a certain number of strands propagate without forming loops, and undergoing a given permutation.

We saw that on the cylinder, the largest eigenvalue of the transfer matrix in a given sector is related to loop configurations through \eqref{eq:youngnew1}. Following this relation, we define the weight
\begin{equation}
\label{eq:www}
W^{n,m}_{\sigma}=1-\tfrac{2\iota_m(\sigma)}{{n+m\choose m}}\,,
\end{equation}
which is merely the same weight as in \eqref{eq:youngnew1} after a normalization. We recall that $\iota_m(\sigma)$ is the number of subsets of $m$ strands among the $n+m$ strands that intersect between themselves an odd number of times (in a graphical representation where there is no winding around the two end points), defined in \eqref{eq:iota}. We gave an example with the permutation $\sigma=(4\,1\,3\,2)$ in Figure \ref{fig:watermelon}. We have in this figure $\iota_0=\iota_1=0$, $\iota_2=4$, $\iota_3=2$, $\iota_4=0$, so that $W^{4,0}_\sigma=1,W^{3,1}_\sigma=1,W^{2,2}_\sigma=-1/3,W^{1,3}_\sigma=0,W^{0,4}_\sigma=1$. Note that we always have $W^{n,0}_\sigma=W^{n-1,1}_\sigma=1$.

\subsubsection{Even number of legs}


Equation \eqref{eq:youngnew1} then directly translates into relations between the $P^{n+m}_\sigma(x)$. The cases of loop weights $N=1,0,-1$ are related to the $osp(r|2)$ models with $r=3,2,1$ respectively. The results can be read off from eqs.~\eqref{eq:osp32logcorr2}, \eqref{osp22gaps} and \eqref{eq:osp12gapsym}, 
by taking into account \eqref{scalegaps}--\eqref{afflecklogeq} that link the two terms in the expression to respectively the power law and logarithmic exponents. 
In the following, $\sim$ gives the asymptotic behaviour in $x$ up to a constant multiplicative factor.

\subsubsection*{Loop weight $1$~:}
\begin{equation}
\boxed{\sum_{\sigma \in \mathfrak{S}_{n+m}} W^{n,m}_{\sigma} P^{n+m}_\sigma(x) \sim  x^{-\tfrac{m^2-1}{2}}(\log x)^{\tfrac{m^2-1}{2}-n(n+1)}\,, \quad \text{for $n$ and $m$ odd}\,.}
\end{equation}
The case $m=1$ where $W^{n,1}_{\sigma}=1$ is consistent with Monte Carlo simulations in \cite{nahum}.

\subsubsection*{Loop weight $0$~:}

\begin{equation}
\label{eq:watermelon0}
\boxed{
 \sum_{\sigma \in \mathfrak{S}_{n+m}} W^{n,m}_{\sigma} P^{n+m}_\sigma(x) \sim
 \begin{cases}
 x^{-\tfrac{m^2}{2}}(\log x)^{\tfrac{m^2}{4}-\tfrac{n^2}{2}}\,, & \text{for $m\geq 2$ even and $n$ even} \,, \\
 x^{-\tfrac{m^2-1}{2}}(\log x)^{\tfrac{m^2}{4}-\tfrac{n^2}{2}+\tfrac{1}{4}}\,, & \text{for $n$ and $m$ odd}\,. \\
 \end{cases}}
\end{equation}

\subsubsection*{Loop weight $-1$~:}
One has access with $osp(1|2)$ to much less information. Keeping in mind that $P$ is not a probability in this case but only a ratio of two partition functions, one can still write
\begin{equation}
\boxed{
\begin{aligned}
\sum_{\sigma \in \mathfrak{S}_{m}} W^{0,m}_{\sigma} P^{m}_\sigma(x) &\sim x^{-\tfrac{m(m+2)}{2}}(\log x)^{\tfrac{m(m-2)}{6}}\,, \quad \text{for $m$ even} \,,\\
 \sum_{\sigma \in \mathfrak{S}_{1+m}}W^{1,m}_{\sigma} P^{1+m}_\sigma(x) & \sim x^{-\tfrac{m^2-1}{2}}(\log x)^{\tfrac{m^2+3}{6}}\,, \qquad \quad \text{for $m$ odd}\,.
\end{aligned}}
\end{equation}\\

\subsubsection{One leg \label{sec:oneleg}}
The information on the watermelon exponents for an odd number of legs is contained in the spin chains of odd size $L$. In particular the one-leg case corresponds to the order parameter. For a loop weight $0$, it corresponds to the case $m=0,n=0$ in \eqref{eq:osp22gapnm}, and from \eqref{gkappa2logL} this gives a gap
\begin{equation}
\frac{L^2 \Delta e_L}{2\pi v_F}=-\frac{1}{4\log L}\,,
\end{equation}
corresponding to the following behaviour of the order parameter 
\begin{equation}
\langle \phi(x)\phi(0)\rangle \sim (\log x)^{\tfrac{1}{2}}\,.
\end{equation}
For a loop weight $1$, the Bethe roots associated to the ground state of the $osp(3|2)$ model in odd size $L$ are composed of $(L-1)/2$ strings, that happen to give exactly the same energy $e_L$ in all odd sizes, as in the even size case. Thus the energy gap is exactly $0$ and one gets
\begin{equation}
\langle \phi(x)\phi(0)\rangle \sim 1\,.
\end{equation}
For a loop weight $-1$, the ground state in odd size corresponds to $m=0$ in \eqref{eq:osp12gapsym}, that gives a gap
\begin{equation}
\frac{L^2 \Delta e_L}{2\pi v_F}=-\frac{1}{3\log L}\,,
\end{equation}
hence
\begin{equation}
\langle \phi(x)\phi(0)\rangle \sim (\log x)^{\tfrac{2}{3}}\,.
\end{equation}
For these three cases we observe the behavior
\begin{equation}
\langle \phi(x)\phi(0)\rangle \sim (\log x)^{\tfrac{N-1}{N-2}}\,,
\end{equation}
(recall $N\equiv r-2s$) which will be discussed further in the conclusion.

\subsection{Away from integrability}

The integrable spin chains previously studied correspond to a crossing weight $w$ equal to $(2-N)/4$. For these values---and by using the Bethe-ansatz---
 we showed  that the leading logarithmic corrections are indeed  described by the supersphere sigma model. If there is universality,  this correspondance has no particular reason to hold only at the integrable point: the supersphere sigma model should be relevant to describe the long distance physics of the dense loop soups  for all finite crossing weights $w>0$.  Of course, away from the integrable point, this might be much more difficult to check, since then only direct numerical simulations are available. In Figure \ref{fig:logcor} we show as an example the measured logarithmic corrections corresponding to the 4-leg watermelon 2-point function, for different values of crossing weight $w$, for vanishing loop weight $N=0$.  The leading correction studied in this paper corresponds to the purple line. While Bethe-ansatz results show it does indeed give the correct results in the $L\to\infty$ limit, it is clear that for the sizes studied using direct transfer matrix diagonalization, next order corrections play an important role. It seems however that these corrections can be captured quite easily. The full solution to the RG equations for the sigma model coupling constant is 
 \begin{equation}
 {1\over g}={1\over g_0}+{2-N\over\kappa}\log(L/L_0)\,.
 \end{equation}
 Setting $L_0=1$ (i.e. measuring lengths in units of the lattice spacing) gives
 \begin{equation}
 g=\frac{\kappa}{2-N}\frac{1}{\log L+\tfrac{\kappa}{(2-N)g_0}}\label{RGsol}\,.
 \end{equation}
 Meanwhile, we find that for $N=0$ the numerical results can be collapsed approximately on (\ref{RGsol}) with ${\kappa\over 2g_0}\approx \pi w$. In other words, we have, to a very good approximation,   $g_0\approx {\kappa\over 2\pi w}$. Hence, we see that $w$ plays the role of the inverse bare coupling constant
 
 This observation suggests that at large $w$ the corrections to the gap can be obtained with the same formulas we have derived earlier in this paper, but by using, instead of the running coupling constant $g\propto 1/\log L$, the bare constant $g_0\propto 1/w$. Conversely, we can also imagine solving the problem at large $w$ by elementary means, hence ``re-deriving'' the formulas for the corrections.


 \begin{figure}[H]
 \begin{center}
\includegraphics[scale=0.35]{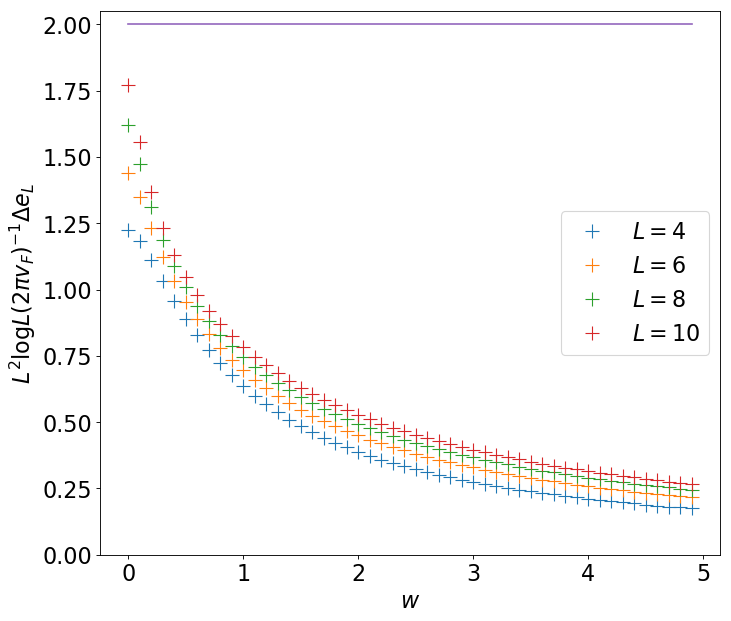} 
\includegraphics[scale=0.35]{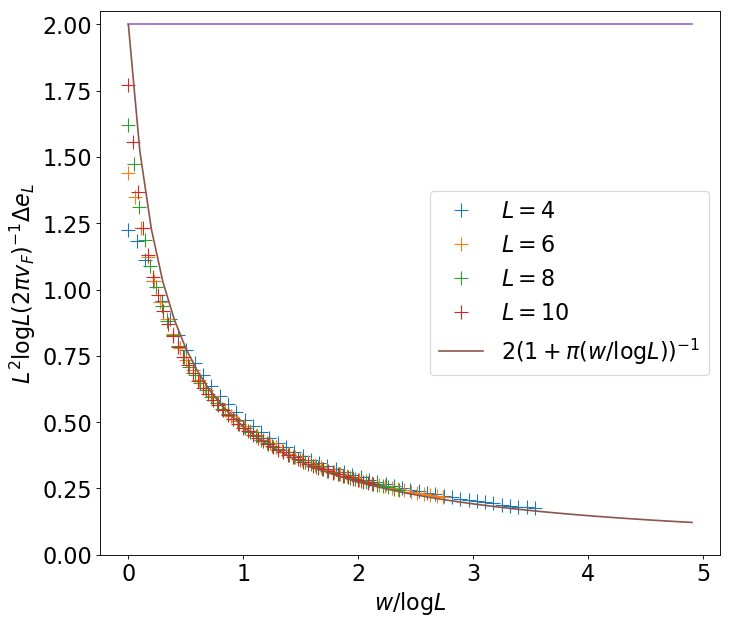} 
\end{center}
\caption{Left: measure of $\tfrac{L^2\log L}{2\pi v_F}\Delta e_L$ for the $4$-legs watermelon 2-point function with loop weight $0$, as a function of $w$ for different sizes. The limit value predicted by the supersphere sigma model is indicated in purple. The integrable point is $w=0.5$. Right: the same data as a function of $w/\log L$.}
\label{fig:logcor}
\end{figure}

%

To this end, we fix the weight for a loop to $N\geq 0$ and denote $T$ the transfer matrix on a cylinder of size $L\times M$. This transfer matrix acts on $L$ strands that are either connected to another strand, or are free, see Figure \ref{fig:extr}. The states are described by a vector space which can be decomposed into a direct sum $\oplus_k E_k$, where $E_k$ is the vector space generated by the states with $k$ free strands among $L$ strands, thus of dimension ${L\choose k}(L-k)!!$.

An important observation is that the transfer matrix, after building a row with $L$ tiles chosen from the three possible tiles in Figure \ref{fig:loopsoupne}, cannot create new free strands, i.e. $T E_k\subset \oplus_{k'\leq k}E_{k'}$. Hence, the transfer matrix $T$ is block-triangular and to find its eigenvalues one can work in a specific sector with a fixed number of free strands $k$. We will denote by $T_k$ the restriction of the transfer matrix to $E_k$ the sector with $k$ free strands.
\begin{figure}[H]
\begin{equation}
 \tikTra\qquad\qquad\qquad \tikTrb
\end{equation}
 \caption{Examples of two states in size $L=6$ with two free strands.}
 \label{fig:extr}
\end{figure}

Let us now study the limit of large crossing weight $w\to\infty$. In this limit the transfer matrix $T_k$ is dominated by choosing a crossing at each site, see Figure \ref{fig:row1}. It creates a loop in the horizontal direction and acts as the identity on the strands. Thus at leading order 
\begin{equation}
T_k=N w^{L} +O(w^{L-1})\,.
\end{equation}
 \begin{figure}[H]
 \begin{center}
\includegraphics[scale=0.35]{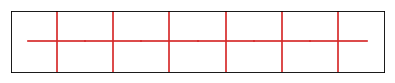}
\end{center}
\caption{The only dominant term in $T_k$ at order $w^L$: there are only tiles with crossings.}
\label{fig:row1}
\end{figure}
The next order is obtained by choosing a left or right corner among the $L$ sites, see Figure \ref{fig:row2}. It does not create any loops and acts as the identity on the other strands. Then
\begin{equation}
T_k=Nw^{L}+2L w^{L-1}+O(w^{L-2})\,.
\end{equation} 
 \begin{figure}[H]
 \begin{center}
\includegraphics[scale=0.35]{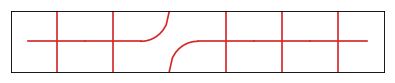}
\includegraphics[scale=0.35]{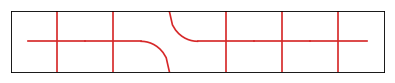}
\end{center}
\caption{Two examples of terms at order $w^{L-1}$ in $T_k$: there is only one tile without crossing.}
\label{fig:row2}
\end{figure}
Denote now $w^{L-2}R_k$ the transfer matrix that corresponds to the next order, i.e., that creates only two corners among the $L$ sites, see Figure \ref{fig:row3}. Since it commutes with the dominant order, one simply has to compute its dominant eigenvalue. Whatever is the connection between the strands, there are $4\cdot  {{L}\choose{2}}$ possibilities of placing the corners, but $2\cdot {{k}\choose{2}}$ of them (when the two corners are in opposite direction as in the two cases at the bottom in Figure \ref{fig:row3}) will connect $2$ of the $k$ free strands, which must not be counted; and among the other possibilities, $2\cdot \tfrac{L-k}{2}$ will create a loop. Thus the sum of the entries of each column of $R_k$ is always equal to $4\cdot  {{L}\choose{2}}-2\cdot {{k}\choose{2}}+(N-1)(L-k)$. Since the entries of $R_k$ are nonnegative ($N\geq 0$ case), one can conclude that the dominant eigenvalue of $R_k$ is exactly $4\cdot  {{L}\choose{2}}-2\cdot {{k}\choose{2}}+(N-1)(L-k)$. It follows that the dominant eigenvalue of $T_k$ at order $w^{L-2}$ is
\begin{equation}
\lambda_k=N w^L +2L w^{L-1}+w^{L-2}\left( 4\cdot  {{L}\choose{2}}-2\cdot {{k}\choose{2}}+(N-1)(L-k)\right)+O(w^{L-3})\,.
\end{equation}
 \begin{figure}[H]
 \begin{center}
\includegraphics[scale=0.35]{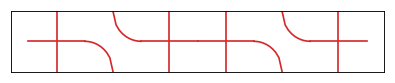}
\includegraphics[scale=0.35]{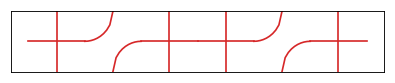}\\
\vspace{20pt}
\includegraphics[scale=0.35]{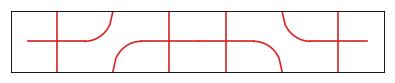}
\includegraphics[scale=0.35]{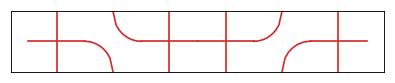}
\end{center}
\caption{Four examples of terms at order $w^{L-2}$ in $T_k$: there are two tiles without crossing among the $L$ tiles.}
\label{fig:row3}
\end{figure}
This expansion has been checked numerically. It can also be continued at order $w^{L-3}$. At order $w^{L-4}$ complications appear since at each transfer matrix step the number of possibilities depends on the state; and from order $w^{L-5}$ on, the interaction between different orders counts and  probably cannot be simply taken into account.


Assume now $N=0$. We have
\begin{equation}
-\log(\lambda_k/\lambda_0)=\frac{k(k-2)}{2}\frac{1}{wL}+O(w^{-2})\,,
\end{equation}
corresponding to a gap
\begin{equation}
\frac{L^2 \Delta e_L^k}{2\pi v_F}=\frac{k(k-2)}{2} \frac{1}{2\pi w}\,.
\end{equation}

This matches the result  equation (\ref{osp22gaps}) with $j=1/2,q=(k-1)/2$ (that corresponds to $1$ fermionic box and $k-1$ bosonic boxes in the Young tableau, equivalent to $k$ free bosonic strands on the infinite cylinder), and  the coupling constant must be put to $\kappa^{-1}g={1\over 2\pi w}$, in agreement with our earlier discussion. We conclude that, remarkably, the weak-coupling sigma model provides a very accurate description of the loop soup at large $w$ with the very simple correspondence $g\propto 1/w$. See Figure \ref{fig:logcor2} for further numerical evidence with other sectors. For reasons we do not quite understand, this simple correspondence 
 seems to be true only for $N=0$.\\

 \begin{figure}[H]
 \begin{center}
\includegraphics[scale=0.35]{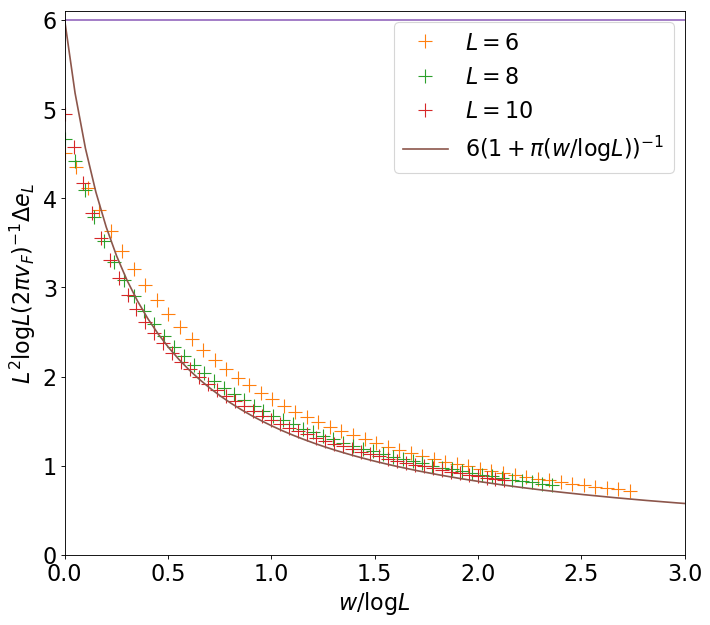} 
\includegraphics[scale=0.35]{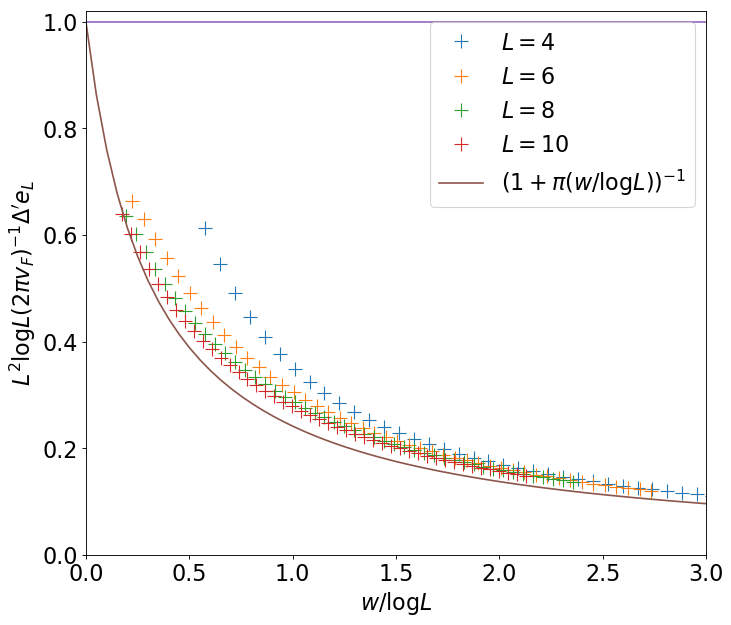} 
\end{center}
\caption{Left: measure of $\tfrac{L^2\log L}{2\pi v_F}\Delta e_L$ for the $6$-legs watermelon 2-point function with loop weight $0$, as a function of $w/\log L$ for different sizes. Right: measure of $\tfrac{L^2\log L}{2\pi v_F}\Delta' e_L$ with $\Delta ' e_L$ the difference between the $(q,j)=(1,1)$ sector and the $(q,j)=(0,1)$ sector,  with loop weight $0$, as a function of $w/\log L$ for different sizes. The limit values predicted by the supersphere sigma model are indicated in purple.}
\label{fig:logcor2}
\end{figure}

Let us nevertheless give some comments on the $N\neq 0$ case. One can compute the large $w$ regime of the eigenvalues as well:
\begin{equation}
-\log(\lambda_k/\lambda_0)=\frac{k(k+N-2)}{N}\frac{1}{w^2}+O(w^{-3})\,,
\end{equation}
which corresponds to  a gap
\begin{equation}
\label{eq:gaplargew}
\frac{L^2 \Delta e_L^k}{2\pi v_F}=\frac{k(k+N-2)}{N}\frac{L}{2\pi w^2}+O(w^{-3})\,.
\end{equation}
However the presence of the factor $L$ and the behavior in $w^{-2}$ suggest that the coupling constant $g$ at $N\neq 0$ could behave as $\propto (\log L +...+ \text{C }\times w^2/L)^{-1}$ where the dots indicate a term subdominant in $w$, and $C$ a constant. When $w\to\infty$ the $w^2/L$ part dominates and one observes the behavior \eqref{eq:gaplargew}. 
However it turns out that the finite-size corrections are not as easily captured as in the $N=0$ case. Note finally that the expansion \eqref{eq:gaplargew} for the energy gaps is valid for the singular case $N=2$ as well.

\section{Conclusion}

In a previous paper \cite{frahmmartins} the amplitude of logarithmic corrections for the states of the $osp(3|2)$ spin chain with spins $j=1/2$, $q$ (thus for only one degree of freedom) was interpreted as the value of the Casimir. This observation was  generalized in \cite{FM} for similar states and other  $osp(n|2m)$ models. It must be emphasized that the amplitude of the corrections in general is {\sl not given} by the Casimir. 

The emergence of the Casimir in this context can be traced back to the mini-superspace approach to the problem, discussed in detail in \cite{candu2}. In this approach, one neglects fluctuations of the sigma model fields along the space direction (i.e. along the spin chain, in the discretized version), and focusses only on fluctuations in the (imaginary) time direction. This corresponds to the ``particle limit'' of string theory, that is, quantum mechanics on the target---here a supersphere. In this limit, the Hamiltonian becomes proportional (with proportionality constant $g$) to the Laplacian on the target, an object that can easily be diagonalized using group theory. Using this approach, one finds easily that the scaled gaps in this approximation should be of the form
\begin{equation}
{L^2\over 2\pi v_F}\Delta e_L={l(l+r-4)\over 2(4-r)\log(L/L_0)}, L\to\infty\label{scaledgapssigma}\,,
\end{equation}
for the supersphere
 $OSp(r|2)/OSp(r-1|2)$. Of course, identical results are expected for the more general model based on $OSp(r|2s)$ provided $r-2s<2$ so the model flows to weak coupling in the IR. The combination  $r-4$ in (\ref{scaledgapssigma}) must then be replaced by $r-2s-2$:
\begin{equation}
{L^2\over 2\pi v_F}\Delta e_L={l(l+r-2s-2)\over 2(2+2s-r)\log(L/L_0)}, L\to\infty\label{scaledgapssigmai}\,.
\end{equation}
 For $s=0$, $r=3$,  that is the $O(3)$ sigma model, (\ref{scaledgapssigmai})  has been checked in great detail in \cite{balog3}. We recall that for $N>2$ the $O(N)$ models have different renormalization properties, as seen from the beta function above \eqref{scalingofg}, and that in this case it corresponds to the UV regime. 

Associated with (\ref{scaledgapssigmai}) in the case of ordinary spheres---that is, sigma models on $O(N)/O(N-1)$---are homogeneous symmetric polynomials   \cite{vilenkin,Garrettreview} in the order parameter components $\phi_i,i=1,\ldots, N$. They can be associated with fully symmetric representations  on $l$ boxes, with  $l=1$  corresponding to the fundamental  representation, i.e., the order parameter itself. Of course, from the general formula for the Casimir of $O(N)$ representations
\begin{equation}
{\cal C}=\sum_{i=1}^{[{N\over 2}]}\lambda_i^2+(N-2i)\lambda_i\,,
\end{equation}
for the Young tableau $(\lambda_1,\ldots,\lambda_{[{N\over 2}]})$, we see that the numerator of (\ref{scaledgapssigmai}) is just the Casimir of the corresponding representations.

This picture remains true in the case of superspheres with $N=r-2s$, although some care has to be taken  because of the complexity of $osp(r|2s)$ representation theory.

The point is that the minisuperspace approach is only able to give corrections to dimensions of the fields without derivatives - that is, in practice, corrections to scaled gaps which vanish in the limit $L\to\infty$. Whenever a scaled gap has a finite part, corresponding to a finite value of the critical exponent in the IR fixed point theory, the associated field must involve derivatives of the sigma model field, and the minisuperspace formula per se cannot be applied. There is   no reason in this case to expect the residue of the correction to still be given by the Casimir. 

This can be immediately checked on the three examples $osp(1|2)$, $osp(2|2)$ and $osp(3|2)$, where the logarithmic corrections given in  \eqref{osp12leadingcorr}, \eqref{osp22gaps}, and \eqref{eq:osp32logcorr2} do not match the expression of the Casimir in \eqref{eq:casimirosp12}, \eqref{eq:casimirosp22} and \eqref{eq:casimirosp32}. More precisely, we observe that only the bosonic part is given by the Casimir, the fermionic part for a spin $j$ half-integer being always given by $j^2$. Precisely one has the scaled gaps

\begin{equation}
{L^2\over 2\pi v_F}\Delta e_L^{j,q}=(h+\bar{h})+\frac{1}{(2-N)}\left(j^2-\frac{1}{2}q(q+N)+\frac{1-2N}{4} \right)\frac{1}{\log L/L_0}\,.
\end{equation}

In terms of the Young tableau describing the representation corresponding to the state, $2j$ is the number of 'fermionic' boxes (those aligned horizontally) and $q$ the number of 'bosonic' boxes (those aligned vertically).

It is intriguing to ask whether the amplitude of logarithmic corrections in the scaled gaps for more general states (with possibly different chiral and antichiral content)
\begin{equation}
{L^2\over 2\pi v_F}\Delta e_L^{\ldots}=(h+\bar{h})+{{\cal A}(\ldots)\over \log L}\label{moregengaps}\,,
\end{equation}
can be expressed simply in terms of the quantum numbers of the corresponding $osp$ representation, together with some other quantum numbers such as the orders of derivatives etc \cite{wegner1989}.  There does exist a related calculation of such an amplitude---deduced from the leading power law behavior of the two-point function of the corresponding operators---in the conformal case  \cite{cagnazzo,mitev}. However, since our sigma models are not conformal, there is no reason for these results  to apply here (except for $OSp(4|2)$). The fully symmetric case is an exception, unified by the mini-superspace approach, for which   
${\cal A}={C\over 2(2+2s-r)}$. 

We also note that the order parameter always appears with negative logarithmic corrections when compared with the group invariant singlet ground state. This is presumably related with the fact that the symmetry is spontaneously broken. Note that negative logarithmic corrections on the cylinder means correlation functions in the plane that grow logarithmically with distance. This is in agreement with the general perturbative result for the two point function of the order parameter at weak coupling for ordinary $O(N))$ models \cite{polyakov}:
\begin{equation}
\langle \phi_a(x)\phi_a(0)\rangle=\left[1-\hbox{C } g_\sigma^2(N-2)\log x\right]^{N-1\over N-2}\,.
\end{equation}
with $C$ a constant.
When $N<2$, the sign of the correction switches, showing that fluctuations do not destroy the spontaneous order---in agreement with the result that the symmetry is in fact broken. Taking this expression seriously gives long-distance correlations at weak coupling proportional to $(\log x)^{N-1\over N-2}$. This is in agreement with the result in (\ref{scaledgapssigma}) using equation (\ref{afflecklogeq}).

We also believe that  in \cite{frahmmartins} the charge associated to the continuum degree of freedom was wrongly interpreted as  the fermionic charge (the one that takes half-integer values), whereas it should be associated to the bosonic charge. We think the misunderstanding is linked to a question of highest-weight vector with respect to the $osp(3|2)$ algebra that is explained in section \ref{sec:osp32multiplet}, although the logarithmic corrections therein are numerically correct. Indeed in \cite{frahmmartins} states are studied that are said to belong to the $(0,q)$ sectors with $q\geq 1/2$, which seems to imply that the continuum is associated to the fermionic charge, see \cite{vanderjeugt} in which the only irreducible representation with a null fermionic charge is the trivial representation. In fact they belong to the $(1/2,q-1/2)$ sectors, and $q$ is indeed associated to the bosonic chargee. It is crucial to consider the charges of the highest-weight state (and not the Bethe state) to make the correspondance between the bosonic part of the logarithmic corrections and the Casimir; but also when discussing the loop configurations associated to the sector, see section \ref{sec:eigenvalues}.

A word on the spin chains with $sl(r|2s)^{(2)}$ symmetry: these chains also have the $osp(r|2s)$ symmetry and are critical, but they happen to be non-relativistic with a quadratic dispersion relation. The energy difference between the first excited states behaves as $L^{-3}$. Although at order $L^{-2}$ this gives rise, formally, to a continuum, the result cannot be related to any critical exponents because of the absence of conformal invariance.

An important aspect that is missing in our results is the calculation of the finite part of the density of states for the continuum component of the spectrum in the cases of $OSp(2|2)$ and $OSp(3|2)$. While it seems possible to determine this density for some values of the conformal weight using brute force and the Bethe-ansatz, we do not know for now how to obtain formulas in full generality, such as the ones checked (but not proved either) in \cite{ikhlef,ikhlef2}.

It is finally relevant to question the role of the periodic boundary conditions in our calculations. Different boundary conditions, e.g. open boundary conditions modify the scaled gaps and thus the critical exponents of the physical observables, and one can ask how this is translated in the field theory setup. Formula  \eqref{eq:pert}  in appendix \ref{sec:appc} for the modifications to the finite-size corrections from the Bethe equations is a first step in this direction.

\smallskip
\noindent{\bf Acknowledgments}: this work was supported in part by the Advanced ERC Grant NuQFT. We thank V. Mitev and Y. Ikhlef for discussions.

\smallskip

\appendix
\section{Change of grading \label{sec:grading}}
In the algebraic Bethe Ansatz with fermionic degrees of freedom, a choice has to be made on the grading used, i.e., to chose which index of the R-matrix is bosonic or fermionic. Different gradings lead to different Bethe equations and different expression for the eigenvalues and eigenvectors. But it turns out that in all the models considered one can pass from a grading to another by applying a transformation on the Bethe equations and  Bethe eigenvalues. It shows that \textit{for the eigenvalues}, all the gradings are equivalent, in the sense that if there is a set of Bethe roots in one grading that gives a precise eigenvalue, then there has to exist Bethe roots in all the other gradings (possibly degenerate) that give exactly the same eigenvalue. Nevertheless, nothing guarantees that the corresponding eigenvector in another grading will be non-zero. In other words, there may be some eigenvectors that we can build with the Bethe Ansatz only in particular gradings. All the gradings are equivalent for the eigenvalues, but not for the eigenvectors.\\

We present in the following the transformation in a rather general way. It is a generalization of what is presented in \cite{essler}. We assume that we have $r$ distinct families of $M_n$ Bethe roots $\lambda^n_i$, $n$ being the index of the family and $i$ the index of the root inside the family, and a set of reals $\alpha^n_{m,s}$, $n$ and $m$ indexing the families, and $s$ being an aditionnal index varying from $1$ to $u_{n,m}$ (that is needed for the Bethe equations to keep the same shape in another grading). We assume that the Bethe equations read for all $n=1,...,r$ and $i=1,...,M_n$:
\begin{equation}
\left(\prod_{s=1}^{u_{n,0}}\dfrac{\sinh(\lambda^n_i+i\gamma\alpha^n_{0,s})}{\sinh(\lambda^n_i-i\gamma\alpha^n_{0,s})}  \right)^L=\prod_{m=1}^r\prod_{j=1}^{M_m}\prod_{s=1}^{u_{n,m}}\dfrac{\sinh(\lambda^n_i-\lambda^m_j+i\gamma\alpha^n_{m,s})}{\sinh(\lambda^n_i-\lambda^m_j-i\gamma\alpha^n_{m,s})}\,,
\end{equation}
for $\gamma>0$ a parameter. Note that the case of non-$q$ deformed Bethe equations can be recovered by taking $\gamma\to 0$ and rescale $\lambda^n_i$ by $\gamma\lambda^n_i$.
We impose without making it explicit that if $n=m$ in the product then the condition $i\neq j$ has to be taken. Two important assumptions have to be made in order to do be able to do the transformation:
\begin{itemize}
\item symmetry of the Bethe equations: $\alpha^n_{m,s}=\alpha^m_{n,s}$ and $u_{n,m}=u_{m,n}$.
\item existence of a non-self-coupling family: there exists $n$ such that $\alpha^{n}_{n,s}=0$\,.
\end{itemize}
These assumptions are stable under the transformation, as it will be shown. They are satisfied for the models studied in this paper.\\

Let $n$ be such that $\alpha^{n}_{n,s}=0$. The first step is to rewrite the equation for the $n$-th family as being the root of the polynomial $P(X)$ that reads:
\begin{equation}
\begin{aligned}
P(X)=&\left( \prod_{s=1}^{u_{n,0}}  (X q^{\alpha^n_{0,s}}-q^{-\alpha^n_{0,s}})  \right)^L \prod_{m=1,\neq n}^r\prod_{j=1}^{M_m}\prod_{s=1}^{u_{n,m}} (X q^{-\alpha^n_{m,s}}-t_j^m q^{\alpha^n_{m,s}})\\
&-\left( \prod_{s=1}^{u_{n,0}}  (X q^{-\alpha^n_{0,s}}-q^{\alpha^n_{0,s}})  \right)^L \prod_{m=1,\neq n}^r\prod_{j=1}^{M_m}\prod_{s=1}^{u_{n,m}} (X q^{\alpha^n_{m,s}}-t_j^m q^{-\alpha^n_{m,s}})\,,
\end{aligned}
\end{equation}
where we set $t_j^m=\exp 2\lambda_j^m$ and $q=e^{i\gamma}$. The Bethe equation for the $n$-th family is thus $P(t_i^n)=0$. But $P$ is a ploynomial of degree $Lu_{n,0}+\sum_{m=1,\neq n}^r u_{n,m}M_m$ which is not necessarily equal to $M_n$ the number of Bethe roots in family $n$. Set $Lu_{n,0}+\sum_{m=1,\neq n}^r u_{n,m}M_m=M_n+M_n'$ and define $s_i^n=\exp 2\mu_i^n$ as the $M_n'$ other roots of $P$. These will be the Bethe roots of family $n$ in the other grading. They satisfy exactly the same Bethe equation as the former roots, but the equations for the other family are changed. All the $\lambda_j^n$ must be changed for the $\mu_j^n$. Consider first an $\alpha^m_{n,s}\neq 0$ that leads to a factor $A$ of the type:
\begin{equation}
A=\log \prod_{j=1}^{M_n}\dfrac{\sinh(\lambda^m_i-\lambda^n_j+i\gamma\alpha^m_{n,s})}{\sinh(\lambda^m_i-\lambda^n_j-i\gamma\alpha^m_{n,s})}=-2M_n\alpha^m_{n,s} \log q+ \sum_{j=1}^{M_n}\log \dfrac{t_i^m q^{2\alpha^m_{n,s}}-t_j^n}{t_i^m q^{-2\alpha^m_{n,s}}-t_j^n }\,.
\end{equation}
Define then the function $f(z)$ as:
\begin{equation}
f(z)=\log \dfrac{t_i^m q^{2\alpha^m_{n,s}}-z}{t_i^m q^{-2\alpha^m_{n,s}}-z}\,.
\end{equation}
In the complex plane, $(\log P)'(z)$ has $M_n+M_n'$ poles at $t_j^n$ and $s_j^n$, and $f(z)$ has a branch cut where the argument of the log is real negative, which is a segment from $t_i^m q^{2\alpha^m_{n,s}}$ to $t_i^m q^{-2\alpha^m_{n,s}}$. Consider a contour $\mathcal{C}$  encircling the roots $t_j^n$ and not the roots $s_j^n$, neither the branch cut. The residue theorem gives:
\begin{equation}
A=-2M_n\alpha^m_{n,s} \log q+\dfrac{1}{2i\pi}\oint_\mathcal{C}f(z)(\log P)'(z)dz\,.
\end{equation}
We now would like to deform the contour so that it encircles (in the other direction) the other poles. But then the branch cut enters in the integral, which becomes $2\pi i$ times the integral of $(\log P)'$ over the segment. It gives:
\begin{equation}
A=-2(M_n+M_n')\alpha^m_{n,s} \log q-\sum_{j=1}^{M_n'}\log \dfrac{t_i^m q^{\alpha^m_{n,s}}-s_j^nq^{-\alpha^m_{n,s}}}{t_i^m q^{-\alpha^m_{n,s}}-s_j^n q^{\alpha^m_{n,s}}}+\log\dfrac{P(t_i^m q^{2\alpha^m_{n,s}})}{P(t_i^m q^{-2\alpha^m_{n,s}})}\,.
\end{equation}
But using the symmetry $\alpha^m_{n,s}=\alpha^n_{m,s}$, one of the two terms in $P$ is zero when evaluated at these points, so that we get:
\begin{equation}
\begin{aligned}
\dfrac{P(t_i^m q^{2\alpha^m_{n,s}})}{P(t_i^m q^{-2\alpha^m_{n,s}})}=&-q^{2\alpha^m_{n,s}(Lu_{n,0}+\sum_{m'=1,\neq n}^ru_{n,m'})}\left(\prod_{s'=1}^{u_{n,0}}\dfrac{t_i^m q^{-\alpha^n_{0,s'}+\alpha^m_{n,s}}-q^{\alpha^n_{0,s'}-\alpha^m_{n,s}}}{t_i^m q^{\alpha^n_{0,s'}-\alpha^m_{n,s}}-q^{-\alpha^n_{0,s'}+\alpha^m_{n,s}}} \right)^L\\
&\times \prod_{m'=1,\neq n}^r\prod_{j=1}^{M_{m'}}\prod_{s'=1}^{u_{n,m'}} \dfrac{t_i^m q^{\alpha^m_{n,s}+\alpha^{n}_{m',s'}}-t_j^{m'} q^{-\alpha^m_{n,s}-\alpha^{n}_{m',s'}}}{t_i^m q^{-\alpha^m_{n,s}-\alpha^{n}_{m',s'}}-t_j^{m'} q^{\alpha^m_{n,s}+\alpha^{n}_{m',s'}}}\,.
\end{aligned}
\end{equation}
Avoiding the term $i=j$ for $m'=m$ for every $s'$ (which is present by symmetry of the $\alpha$'s, since $\alpha^m_{n,s}\neq 0$), the $-$ factor becomes $(-1)^{u_{n,m}-1}$. Since there are $u_{n,m}$ such multiplicative factors, they contribute to $1$. The equations can be transformed back into $\sinh$ form, so that we get the Bethe equations for all $k=1,...,r$ and $i=1,...,M_k$:

\begin{equation}
\left(\prod_{s=1}^{u_{k,0}}\dfrac{\sinh(\lambda^k_i+i\gamma\kappa^k_{0,s})}{\sinh(\lambda^k_i-i\gamma\kappa^k_{0,s})}  \right)^L=\prod_{m=1}^r\prod_{j=1}^{M_m}\prod_{s=1}^{u_{k,m}}\dfrac{\sinh(\lambda^k_i-\lambda^m_j+i\gamma\kappa^k_{m,s})}{\sinh(\lambda^k_i-\lambda^m_j-i\gamma\kappa^k_{m,s})}\,,
\end{equation}
with for $m,m'\neq n$:
\begin{equation}
\label{changegrading}
\begin{aligned}
&\{\kappa^m_{0,s}\}_s=\{ \alpha^m_{n,s}-\alpha^n_{0,s'}\}_{s',s}\cup \{ -\alpha^m_{0,s'}\}_{s'}\\
&\{\kappa^m_{n,s}\}_s=\{ \alpha^m_{n,s}\}_{s}\\
&\{\kappa^m_{m',s}\}_s=\{-\alpha^m_{m',s}\}_{s}\cup \{ -\alpha^m_{n,s}-\alpha^n_{m',s'}\}_{s,s'}\,,
\end{aligned}
\end{equation}
and for all $m$:
\begin{equation}
\kappa^n_{m,s}=\alpha^n_{m,s}\,,
\end{equation}
We recall that in these new Bethe equations, the new $M_n$ is the former $M_n'$, and the new $\lambda^n_i$ are the former $\mu^n_i$. The two assumptions are still satisfied in this new grading, since the $\alpha$'s are symmetric and since the family $n$ is still non-self-coupling (it is the only whose Bethe roots changed, but whose Bethe equations did not). We stress the fact that in the formulas \eqref{changegrading}, according to our conventions, a $\kappa$ equal to zero must not be counted, a $\kappa$ cancels a $-\kappa$, and importance must be given to the range of the $s$ and $s'$ (in particular in $\{...\}_{s,s'}$ if one of them sums over the empty set then the whole set is empty).

For example, for $osp(2|2)$ one has $\alpha^1_0=\alpha^2_0=1/2$, $\alpha^1_1=\alpha^2_2=0$, $\alpha^1_2=\alpha^2_1=1$. After the transformation it gives no source term for the second family, and $\kappa^2_1=1$, $\kappa^2_2=-2$.

\section{Degeneracies of the $osp(4|2)$ model\label{sec:osp42}}

To give numerical support to the arguments in section \ref{sec:osp42deg}, we present here the results of an exact diagonalization in finite-size $L=2,4,6$ and the root structure corresponding to each state. The degeneracy in parentheses corresponds to the degeneracies of the same eigenvalue for $su(2)$. We use the abbreviation `s' for string, `t' for theta root. $su(2)'$ means that it is a solution of  the XXX Bethe equations \eqref{BExxx} with a multiplicative $-1$ factor.


\subsubsection*{$L=2$}
\begin{center}
\begin{tabular}{|c|c|c|c|}
\hline
eigenvalue & particularity & degeneracy  & roots \\
\hline
$e_1$ & su(2) & 1 (1)& degenerate \\
$e_2$ & $-e_3$ & 18 & void \\
$e_3$ & su(2) & 17 (3) & 1 t at 0 \\
\hline
\end{tabular}
\end{center}

\subsubsection*{$L=4$}
\begin{center}
\begin{tabular}{|c|c|c|c|}
\hline
eigenvalue & particularity & degeneracy & roots \\
\hline
$e_1$ & su(2) & 1 (1) & 2 s degenerate \\
$e_2$ & $-e_3$ & 18  & 1 s at 0 \\
$e_3$ &su(2)& 17 (3) & 1 s at 0, 1 t at 0 \\
$e_4$& su(2) & 1 (1) & (s in su(2))  \\
$e_5$ &su(2)& 17  & 1 s at 0+, 1 t at 0- \\
$e_6$ &su(2)& 17 & 1 s at 0-, 1 t at 0+ \\
$e_7$ &$-e_5$& 18 (3) & 1 s at 0+ \\
$e_8$ &$-e_6$& 18 (3) & 1 s at 0- \\
$e_9$ &su(2)& 307 (5) & void\\
$e_{10}$ &$-e_9$&  306  & 1 t at 0\\
$e_{11}$ &$\approx i e_9$& 288  & 1 t $>$ 0\\
$e_{12}$ &$\approx -i e_9$& 288 & 1 t $<$ 0\\
\hline
\end{tabular}
\end{center}

\subsubsection*{$L=6$}
\begin{center}
\begin{tabular}{|c|c|c|c|}
\hline
eigenvalue & particularity & degeneracy & roots \\
\hline
$e_1$ & su(2) & 1 (1) & 3 s degenerate \\
$e_2$ & $-e_3$ &  18 & 2 s around 0 \\
$e_3$ & su(2) & 17 (3) & 2 s around 0, t at 0 \\
$e_4$ & su(2) & 1 (1) & (s in su(2))\\
$e_5$ &su(2)  &  17 (3) & 2 s $> $ 0 , t at 0 \\
$e_6$ &su(2)  & 17 (3) &2 s $< $ 0 , t at 0 \\
$e_7$ & $-e_5$  & 18 & 2 s $> $ 0 \\
$e_8$ & $-e_6$ & 18& 2 s $< $ 0\\
$e_9$ &su(2)  &  17 (3) & 1 s $> $ 0 , t at 0 \\
$e_{10}$ &su(2)  & 17 (3) & 1 s $< $ 0 , t at 0 \\
$e_{11}$ & $-e_9$  & 18 &1 s $> $ 0 \\
$e_{12}$ & $-e_{10}$ & 18& 1 s $< $ 0 \\
$e_{13}$& su(2) & 307 (5)& 1 s at 0  \\
$e_{14}$ & $-e_{13}$& 306 & 1 s at 0, 1 t at 0 \\
$e_{15}$ & su(2)'& & 1 s $>$ 0, 1 t at 0 \\
$e_{16}$ & su(2)'& & 1 s $<$ 0, 1 t at 0 \\
... & ... & ... & ... \\
\hline
\end{tabular}
\end{center}
When the Hamiltonian limit is taken, a link with the degeneracies derived in \cite{candu2} could be studied.

\section{Logarithmic corrections from the Bethe ansatz \label{sec:appc}}
\subsection{Generalities}
In this appendix we relate the previously derived logarithmic corrections to the parameters of the Bethe equations, and give a formula for the perturbation of the exponents $h+\bar{h}$ and $\alpha$ in \eqref{afflecklogeq} when the Bethe equations are perturbed by an additional source term at one site, that occurs for example when the boundary conditions are modified.

To fix the ideas, we consider the $su(2)$ or $osp(1|2)$ spin chains that can both be recast into the logarithmic form
\begin{equation}
z_L(\lambda)=s(\lambda)+\frac{t(\lambda)}{L}-\frac{1}{L}\sum_i r(\lambda-\lambda_i)\,,
\end{equation}
where the Bethe roots $\lambda_i$ satisfy $z_L(\lambda_i)=I_i/L$ with $I_i$ a Bethe number. The energy $e_L$ is then
\begin{equation}
e_L=-\frac{2\pi}{L}\sum_i s'(\lambda_i)\,,
\end{equation}
$s$ and $r$ are functions that read
\begin{equation}
\label{rs}
\begin{aligned}
s(\lambda)&=\frac{1}{\pi}\arctan 2\lambda\,, \quad \text{for $su(2)$ and $osp(1|2)$}\\
r(\lambda)&=\frac{1}{\pi}\arctan \lambda\,, \quad \text{for $su(2)$}\\
&=\frac{1}{\pi}\arctan \lambda-\frac{1}{\pi}\arctan 2\lambda\,,\quad \text{for $osp(1|2)$}\,,
\end{aligned}
\end{equation}
and $t(\lambda)$ is an additional source term that is zero for periodic boundary conditions, but that can be non-zero in case of isolated roots (such as 'strings') or open boundary conditions. It is assumed to be odd and continuous, and to satisfy the expansion \eqref{eq:decay} hereafter. We define $r_\infty$, $\eta$, $t_\infty$, $\eta_t$ by the expansion at large $\lambda>0$
\begin{equation}
\label{eq:decay}
r(\lambda)=r_\infty-\frac{\eta}{\lambda}+o(\lambda^{-1})\,,\quad t(\lambda)=t_\infty-\frac{\eta_t}{\lambda}+o(\lambda^{-1})\,.
\end{equation}
Moreover we have the Fermi velocities $v_F=\pi$ for $su(2)$ and $v_F=2\pi/3$ for $osp(1|2)$.

\subsection{Perturbation to the critical exponents}
Assume that there are $n_\pm$ vacancies in the positive/negative roots, at the outmost positions, and define $\epsilon_L$ by
\begin{equation}
e_L=e_\infty+\frac{f}{L}+\frac{2\pi v_F}{L^2}\epsilon_L\,,
\end{equation}
with a certain surface energy term $f$. Our result is that for $su(2)$ and $osp(1|2)$, for symmetric stated $n_+=n_-\equiv n$, it reads
\begin{equation}
\label{eq:pert}
\boxed{
\begin{aligned}
\epsilon_L=&-\frac{1}{12}+\frac{t_\infty^2}{1+2r_\infty}+n^2(1+2r_\infty)+2t_\infty n\\
&-\frac{v_F}{\log L}\left(\eta n^2+2\left(\eta_t-\eta\frac{t_\infty}{1+2r_\infty} \right)n -3\eta \left(\frac{t_\infty}{1+2r_\infty} \right)^2+2\eta_t \frac{t_\infty}{1+2r_\infty}\right)\,.
\end{aligned}
}
\end{equation}
For asymmetric states $n_+\neq n_-$, we conjecture the following formula
\begin{equation}
\label{eq:pertasym}
\begin{aligned}
\epsilon_L=&-\frac{1}{12}+\frac{t_\infty^2}{1+2r_\infty}+\frac{(n_++n_-)^2}{4}(1+2r_\infty)+\frac{(n_+-n_-)^2}{4}\frac{1}{1+2r_\infty}+t_\infty (n_++n_-)\\
&-\frac{v_F}{\log L}\left(\eta n_+n_-+\left(\eta_t-\eta\frac{t_\infty}{1+2r_\infty} \right)(n_++n_-) -3\eta \left(\frac{t_\infty}{1+2r_\infty} \right)^2+2\eta_t \frac{t_\infty}{1+2r_\infty}\right)\,.
\end{aligned}
\end{equation}

The physical meaning of \eqref{eq:pert} is the perturbation of the critical exponents (those of the algebraic decay, as well as those of the logarithmic decay) when a modification of the system can be recast into the perturbation of the Bethe equations by an odd function $t(\lambda)/L$.
\subsubsection*{Sketch of the derivation}

It is not the purpose of this paper to give a detailed proof to would add too many technicalities; however we give the main lines of the derivation, using the work done in \cite{granetjacobsensaleur}. 
The logarithmic corrections appear here because the function $r(\lambda)$ decays algebraically at infinity and not exponentially like in the XXZ case. Denoting $S_L(\phi)=\tfrac{1}{L}\sum_i\phi(\lambda_i)$ the sum of a function over the Bethe roots, and $w_L=S_L-S_\infty$, one first establishes that $w_L(\phi)$ can be written as
\begin{equation}
w_L(\phi)=\frac{1}{L}\int t^{\prime \rm dr}\phi+\sum_{\omega, 1+\hat{r'}(\omega)=0}A_\omega \hat{\phi}(\omega)+\sum_{n>0}B_n^\pm \{ \phi^{\rm dr}\}_n^\pm\,,
\end{equation}
with $\widehat{(f^{\rm dr})}=\hat{f}(1+\hat{r'})^{-1}$ and $\{f\}_n^\pm$ the coefficient in $x^{-n}$ in the expansion of $f(x)$ at $\pm\infty$. The terms in $B_n^\pm$ are possible even if $\phi$ does not decay algebraically, because $r$ does, contrary to the XXZ case. Assume for simplicity that $n_+=n_-\equiv n$, i.e. the Bethe roots are symmetric. The energy of the chain at order $L^{-2}(\log L)^{-1}$ is given by $e_L=e_\infty+\tfrac{f}{L}-2\pi\hat{s'}(iv_F)(A_{iv_F}+A_{-iv_F})$ with $f$ some surface energy. Denote $\bar{w}_L(\phi)=w_L(\phi)-\tfrac{1}{L}\int t^{\prime \rm dr} \phi$. Decomposing
\begin{equation}
z_L(\lambda)=z_\infty(\lambda) +\frac{t(\lambda)-(t^{\prime \rm dr}\star r)(\lambda)}{L}-\bar{w}_L(r(\lambda-\cdot))\,,
\end{equation}
one finds by computing $S_L(|z_L|)$ by two ways, first using $z_L(\lambda_i)=I_i/L$, second using $S_L=S_\infty+\bar{w}_L+\tfrac{t^{\prime \rm dr}}{L}$, that
\begin{equation}
\begin{aligned}
(A_{iv_F}+A_{-iv_F})\hat{s'}(iv_F)&=-\frac{v_F}{L^2}\left(-\frac{1}{12}+(1+2r_\infty)(n+\varphi)^2\right)\\
&-\frac{v_FB_1^\pm}{L(1+2r_\infty)}\left( \eta (n+\varphi)+2\left(\eta_t-\eta \frac{2t_\infty}{1+2r_\infty} \right)\right)\,,
\end{aligned}
\end{equation}
with $\varphi=\underset{\lambda\to\infty}{\lim} t(\lambda)-(t^{\prime \rm dr}\star r)(\lambda)=t_\infty/(1+2r_\infty)$. To determine the $B_1^\pm$, without the function $t$ it would be proportional to $n/\log L$, as it can be seen by computing $S_L(1/\log(\alpha-z_L))$ with $\alpha=\underset{\lambda\to\infty}{\lim}z_\infty(\lambda)$. But the presence of $t$ acts like an 'odd' twist $\pm \varphi$ for positive/negative roots, and it amounts to changing $n$ by $n+\varphi$. Finally
\begin{equation}
B_1^\pm=-v_F(1+2r_\infty)\frac{n+\varphi}{L\log L}\,,
\end{equation}
hence formula \eqref{eq:pert} in the case $n_+=n_-=n$.

\subsection{Numerical results \label{sec:numerics}}
In this subsection are given numerical checks of formula \eqref{eq:pert}. The Bethe equations are solved numerically for sizes up to $\approx 1500$, and the results are extrapolated to the thermodynamic limit using a ratio of two polynomials in $\log L$. Many extrapolated results only slightly change with the degrees of the polynomials or with the sizes that are used in the extrapolation; however some cases with a wilder extrapolating curve such as the rightmost case in Figure \ref{fig:su2open} do vary more, although the global shape of the curve is often stable.
\subsubsection{Periodic $su(2)$}
The periodic $su(2)$ case is obtained with $r(\lambda)=\tfrac{1}{\pi}\arctan\lambda$ and $t(\lambda)=0$, hence $v_F=\pi$, $1+2r(\infty)=2$, $t(\infty)=0$, $\eta=\tfrac{1}{\pi}$ and $\eta_t=0$. It yields
\begin{equation}
\epsilon_L=-\frac{1}{12}+2n^2-\frac{n^2}{\log L}
\end{equation}
Only this result was already known \cite{woynaeckle,hamerbatchelorbarber}.

\subsubsection{Open $su(2)$}
The open $su(2)$ case with trivial boundary matrix $K$ has the following Bethe equations
\begin{equation}
\left(\frac{\lambda_i+i/2}{\lambda_i-i/2}\right)^L=\prod_{j\neq i}\frac{\lambda_i-\lambda_j+i}{\lambda_i-\lambda_j-i}\frac{\lambda_i+\lambda_j+i}{\lambda_i+\lambda_j-i}
\end{equation}
where the roots $\lambda_i$ are strictly positive. One can rewrite it with a usual root configuration $\mu_i$ by considering only symmetric root structures, i.e. set of roots that contains $-\mu_i$ if it contains $\mu_i$, and adding the appropriate source term
\begin{equation}
\left(\frac{\mu_i+i/2}{\mu_i-i/2}\right)^L\frac{2\mu_i+i}{2\mu_i-i}\frac{\mu_i+i}{\mu_i-i}=\prod_{j\neq i}\frac{\mu_i-\mu_j+i}{\mu_i-\mu_j-i}
\end{equation}
which is obtained with $r(\lambda)=\tfrac{1}{\pi}\arctan\lambda$ and $t(\lambda)=\tfrac{1}{\pi}\arctan 2\lambda+\tfrac{1}{\pi}\arctan\lambda$, hence $1+2r(\infty)=2$, $t(\infty)=1$, $\eta=\tfrac{1}{\pi}$ and $\eta_t=\tfrac{3}{2\pi}$. For $L/4-n$ roots $\lambda_i$, there are $L/2-2n+1$ roots $\mu_i$, that yield
\begin{equation}
\epsilon_L=-\frac{1}{12}+2n^2-\frac{n(n+1)}{\log L}
\end{equation}
\begin{figure}[H]
 \begin{center}
\includegraphics[scale=0.38]{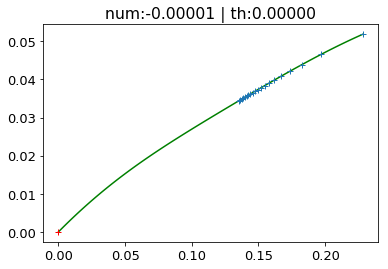}
\includegraphics[scale=0.38]{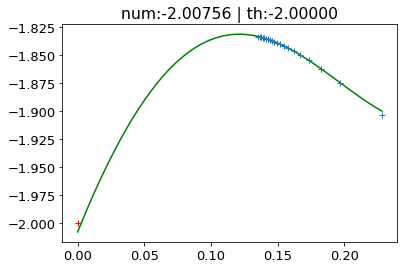}
\includegraphics[scale=0.38]{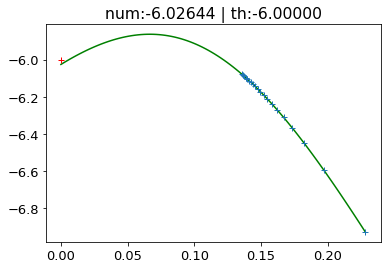}
\end{center}
\caption{Logarithmic correction to the open $su(2)$ states for $n=0,1,2$ (from left to right). The measured limit value and the theoretical values are indicated above the plots.}
 \label{fig:su2open}
\end{figure}

\subsubsection{Periodic $osp(1|2)$}
For periodic $osp(1|2)$ with $L-1-2n$ real roots, one has $v_F=\tfrac{2\pi}{3}$, $1+2r(\infty)=1$, $t(\infty)=0$, $\eta=\tfrac{1}{2\pi}$ and $\eta_t=0$, that yield
\begin{equation}
\epsilon_L=-\frac{1}{12}+(n+\tfrac{1}{2})^2-\frac{(n+\tfrac{1}{2})^2}{3\log L}
\end{equation}
\begin{figure}[H]
 \begin{center}
\includegraphics[scale=0.37]{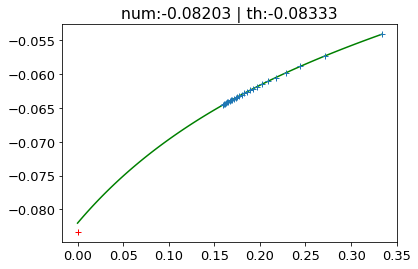}
\includegraphics[scale=0.37]{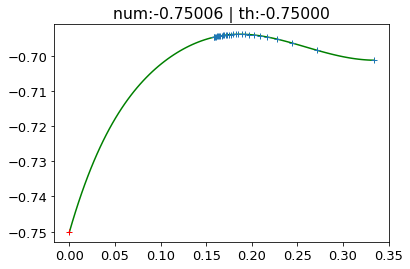}
\includegraphics[scale=0.37]{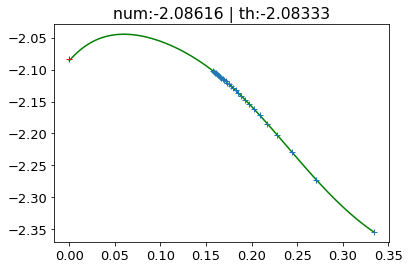}
\end{center}
\caption{Logarithmic correction to the periodic $osp(1|2)$ states for $n=0,1,2$ (from left to right). The measured limit value and the theoretical values are indicated above the plots.}
 \label{fig:su2open}
\end{figure}

\subsubsection{Periodic $osp(1|2)$ with strings}
In case of an exact $2$-string at $0$ for periodic $osp(1|2)$ with $L-2-2n$ other real roots, one has $1+2r(\infty)=1$, $\eta=\tfrac{1}{2\pi}$, $t(\lambda)=\tfrac{1}{\pi}\arctan\lambda-\tfrac{1}{\pi}\arctan 2\lambda-\tfrac{1}{\pi}\arctan 2\lambda/3$, hence $t(\infty)=-\tfrac{1}{2}$, and $\eta_t=-\tfrac{1}{\pi}$, that yield
\begin{equation}
\epsilon_L=\frac{1}{6}+n(n+1)-\frac{(n+1)(n-2)+\tfrac{5}{4}}{3\log L}
\end{equation}

\begin{figure}[H]
 \begin{center}
\includegraphics[scale=0.37]{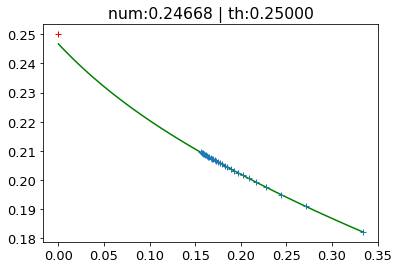}
\includegraphics[scale=0.37]{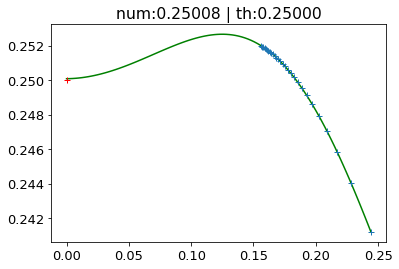}
\includegraphics[scale=0.37]{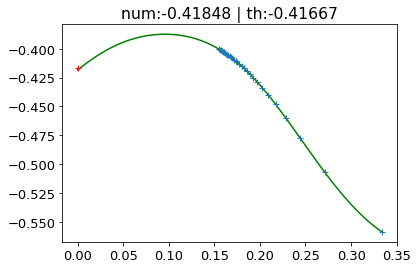}
\end{center}
\caption{Logarithmic correction to the periodic $osp(1|2)$ states with an exact string at $0$ for $n=0,1,2$ (from left to right). The measured limit value and the theoretical values are indicated above the plots.}
 \label{fig:su2open}
\end{figure}

\subsubsection{Open $osp(1|2)$}
For open $osp(1|2)$ with $L/2-n$ positive roots, hence $L-2n+1$ normal roots, one has $1+2r(\infty)=1$, $\eta=\tfrac{1}{2\pi}$, $t(\lambda)=\tfrac{1}{\pi}\arctan\lambda-\tfrac{1}{\pi}\arctan 4\lambda$, hence $t(\infty)=0$, and $\eta_t=\tfrac{3}{4\pi}$, that yield
\begin{equation}
\epsilon_L=-\frac{1}{12}+(n-\tfrac{1}{2})^2-\frac{(n-\tfrac{1}{2})(n+\tfrac{5}{2})}{3\log L}
\end{equation}

\begin{figure}[H]
 \begin{center}
\includegraphics[scale=0.37]{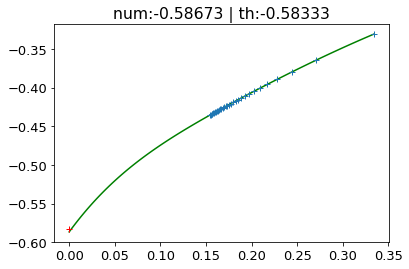}
\includegraphics[scale=0.37]{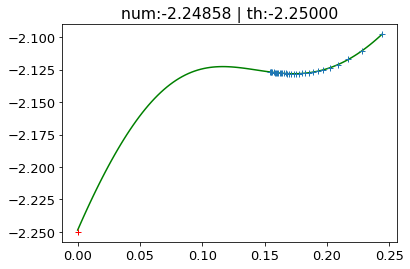}
\includegraphics[scale=0.37]{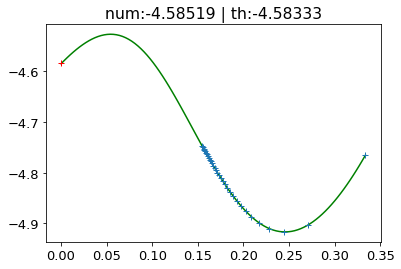}
\end{center}
\caption{Logarithmic correction to the open $osp(1|2)$ for $n=0,1,2$ (from left to right). The measured limit value and the theoretical values are indicated above the plots.}
 \label{fig:su2open}
\end{figure}

\subsection{Loop configurations for $osp(1|2)$ with open boundary conditions}
In case of open boundary conditions with a trivial reflection matrix $K=Id$ \cite{arnaudonavan}, the $osp$ spin chains can be interpreted as loop soups with crossings with periodic boundary conditions in the vertical direction, and with particular boundary conditions in the horizontal direction. It is depicted in Figure \ref{fig:loopOpen}: if the strands cross the left or right boundaries, they are folded back onto the upper or lower adjacent row.

 \begin{figure}[H]
 \begin{center}
\includegraphics[scale=0.5]{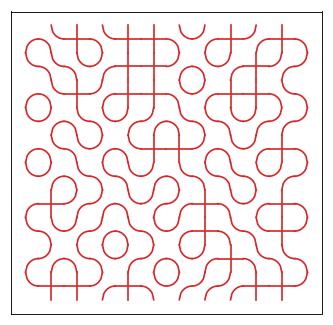}
\end{center}
\caption{A configuration of a dense loop soup with crossings, with 'open' boundary conditions.}
\label{fig:loopOpen}
\end{figure}
An explicit correspondance between the eigenvalues of the transfer matrix, and the constraint induced on the loops with these boundaries can be established like in the periodic case, see Table \ref{tab:osp12open}.

\ytableausetup{smalltableaux}
\ytableausetup{aligntableaux=center}
\begin{table}[H]
\begin{center}
\begin{tabular}{|c|c|c|c|}
\hline
Eigenvalue & Bethe roots $\{\lambda\}$ & $J_z$ & Constraint\\
\hline
\hline
35061 &$\{0.123,0.256,0.414,0.631,1.00 \}$ & $1/2$ &  
  \begin{ytableau}
   \none & *(lightgray) 1 \\
   1\\
  \end{ytableau}\\
\hline
29945 &$\{0.126,0.263,0.430,0.676,1.17+0.47i,1.17-0.47i \}$ & $0$ & no constraint\\
\hline
17050 &$\{0.122,0.254,0.410,0.627 \}$ & $1$ & \ytableaushort{1,2}\\
\hline
3656.1 &$\{0.123,0.256,0.418 \}$ & $3/2$ & 
  \begin{ytableau}
   \none & *(lightgray) 1 \\
   1\\
   2\\
   3
  \end{ytableau}\\
\hline
349.78 &$\{0.126,0.266\}$ & $2$ & \ytableaushort{1,2,3,4}\\
\hline
\end{tabular}
\end{center}
\caption{Correspondance between transfer matrix eigenvalues, Bethe roots, charges, and loop configurations for the open $osp(1|2)$ case in size $L=6$ at the isotropic integrable point.}
\label{tab:osp12open}
\end{table}
The previous results can be used then to determine the watermelon exponents of these configurations. We remark that the exponents in case of open boundaries display the same quadratic part, but a different linear part compared to the periodic boundary case.

\section{Proof of lemma \ref{lemma} \label{app:lemma}}
We give in this appendix a proof of lemma \ref{lemma}.

\begin{proof}
First note that the absolute value of the weight of the configuration is independent of the indices of the loops that compose it, so that only the sign can change. This will precisely occur because of fermionic signs and the fact that $(-1)^{p_b}=-(-1)^{p_f}$ if $b$ is bosonic ($p_b = 0$) and $f$ fermionic $(p_f = 1)$. 

Let us choose a loop $l$ in the configuration (see Figure~\ref{theloopl}) and assume first that it does not cross the boundaries and does not intersect itself (but it can intersect other loops). After the change of index, all the $(-1)^{p_a}$ with $a$ lying on the loop are multiplied by $-1$.
\begin{figure}[H]
\begin{center}
\begin{tikzpicture}[scale=0.5]
\begin{scope}[xscale=-1]
\foreach \k in {0.,1.,...,9.}
\draw[dotted] (0,\k) --++ (9,0);
\foreach \k in {0.,1.,...,9.}
\draw[dotted] (\k,0) --++ (0,9);
\draw (2,3) -- (2,7)--(4,7)--(4,6)--(6,6)--(6,7)--(7,7)--(7,4)--(6,4)--(6,3)--(7,3)--(7,2)--(4,2)--(4,3)--(2,3) ;
\end{scope}
\end{tikzpicture}
\end{center}
\caption{The loop $l$. The other loops are not drawn. 
} \label{theloopl}
\end{figure}
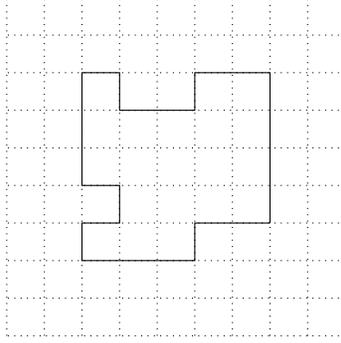
Let us first consider the term $(-1)^{p_ip_j}$ in the $I$ term in \eqref{eq:r}, that is present at each NW or SE corner. After the change of index, each NW and each SE corner thus contributes to a $-1$. In Figure~\ref{looplredcrosses} we drew a red cross at an edge around each such corners to indicate this additional sign term.
 
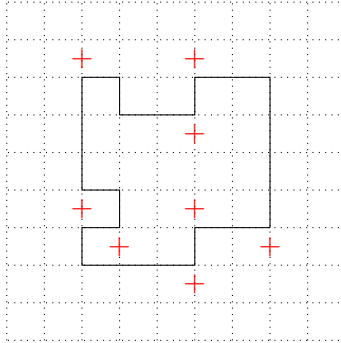
\begin{figure}[H]
\begin{center}
\begin{tikzpicture}[scale=0.5]

\begin{scope}[xshift=10cm,xscale=-1]
\foreach \k in {0.,1.,...,9.}
\draw[dotted] (0,\k) --++ (9,0);
\foreach \k in {0.,1.,...,9.}
\draw[dotted] (\k,0) --++ (0,9);
\draw (2,3) -- (2,7)--(4,7)--(4,6)--(6,6)--(6,7)--(7,7)--(7,4)--(6,4)--(6,3)--(7,3)--(7,2)--(4,2)--(4,3)--(2,3) ;
\draw (2,2.5) node{\color{red} $+$};
\draw (4,1.5) node{\color{red} $+$};
\draw (6,2.5) node{\color{red} $+$};
\draw (4,3.5) node{\color{red} $+$};
\draw (7,3.5) node{\color{red} $+$};
\draw (7,7.5) node{\color{red} $+$};
\draw (4,5.5) node{\color{red} $+$};
\draw (4,7.5) node{\color{red} $+$};
\end{scope}
\end{tikzpicture}
\end{center}
\caption{The red crosses indicate an edge around each corner for which a $(-1)^{p_a}$ appears after changing the index of the loop, taking into account the term $(-1)^{p_ip_j}$ in \eqref{eq:r}.}
\label{looplredcrosses}
\end{figure}

Let us study now the term $(-1)^{\sum_{m=1}^M\sum_{j=2}^L(p_{\alpha_j^m}+p_{\alpha^{m+1}_j})\sum_{i=1}^{j-1}p_{\alpha^m_i}}$. This one 'links' $\alpha^m_i$ and $\alpha^n_j$ whenever $m=n\pm 1$ and $i\gtrless j$, or $m=n$ and $i\neq j$. In the left panel of Figure \ref{looplbluecrosses} the straight blue lines cross all the vertical edges that this sign term 'link' to the vertical edge indicated by a black cross. Notice that a straight line that begins inside the loop and that goes out intersects it an odd number of times; a straight line that begins outside the loop and that goes inside it and comes out intersects it an even number of times. 

Consider an $\alpha^m_i$ (a vertical edge) that does not belong to the loop $l$. The parity of the number of $\alpha$'s linked to it by this sign term depends thus on the position of the points $(m+1/2,i+1/2)$ and $(m-3/2,i-1/2)$ (the blue bullets in Figure \ref{looplbluecrosses}): if they both lie \textit{inside} the loop or \textit{outside} the loop, it is even; if one is inside and the other one outside it is odd. There is thus a $(-1)^{p_a}$ for each vertical edge with index $a$ for which this number is odd, as shown in the middle panel of Figure~\ref{looplbluecrosses}. Now the vertices that correspond to an intersection gives a possibility of simplification: every couple of $(-1)^{p_a}$ that are on each side of an edge that belongs to the loop $l$ simplifies (since their index is primed and $p_a=p_{a'}$). One recovers signs only around SE or NW corners, see the right panel of Figure~\ref{looplbluecrosses}. All these signs then exactly compensate with the signs of the first sign term.

\begin{figure}[H]
\begin{center}
\begin{tikzpicture}[scale=0.5]
\begin{scope}[yshift=-10cm,xscale=-1]
\foreach \k in {0.,1.,...,9.}
\draw[dotted] (0,\k) --++ (9,0);
\foreach \k in {0.,1.,...,9.}
\draw[dotted] (\k,0) --++ (0,9);
\draw (2,3) -- (2,7)--(4,7)--(4,6)--(6,6)--(6,7)--(7,7)--(7,4)--(6,4)--(6,3)--(7,3)--(7,2)--(4,2)--(4,3)--(2,3) ;
\draw (3,5.5) node{\color{black} $+$};
\draw[blue] (2.5,5.5) -- (0,5.5);
\draw[blue] (2.5,4.5) -- (0,4.5);
\draw[blue] (3.5,5.5) -- (9,5.5);
\draw[blue] (3.5,6.5) -- (9,6.5);
\draw (3.5,6.5) node{\color{blue} $\bullet$};
\draw (2.5,4.5) node{\color{blue} $\bullet$};
\end{scope}

\begin{scope}[yshift=-10cm, xshift=10cm,xscale=-1]
\foreach \k in {0.,1.,...,9.}
\draw[dotted] (0,\k) --++ (9,0);
\foreach \k in {0.,1.,...,9.}
\draw[dotted] (\k,0) --++ (0,9);
\draw (2,3) -- (2,7)--(4,7)--(4,6)--(6,6)--(6,7)--(7,7)--(7,4)--(6,4)--(6,3)--(7,3)--(7,2)--(4,2)--(4,3)--(2,3) ;
\draw (2,2.5) node{\color{blue} $+$};
\draw (3,2.5) node{\color{blue} $+$};
\draw (3,3.5) node{\color{blue} $+$};
\draw (4,3.5) node{\color{blue} $+$};
\draw (4,1.5) node{\color{blue} $+$};
\draw (5,1.5) node{\color{blue} $+$};
\draw (5,2.5) node{\color{blue} $+$};
\draw (6,1.5) node{\color{blue} $+$};
\draw (7,3.5) node{\color{blue} $+$};
\draw (3,6.5) node{\color{blue} $+$};
\draw (3,7.5) node{\color{blue} $+$};
\draw (4,5.5) node{\color{blue} $+$};
\draw (4,7.5) node{\color{blue} $+$};
\draw (5,5.5) node{\color{blue} $+$};
\draw (5,6.5) node{\color{blue} $+$};
\draw (7,7.5) node{\color{blue} $+$};
\end{scope}

\begin{scope}[yshift=-10cm, xshift=20cm,xscale=-1]
\foreach \k in {0.,1.,...,9.}
\draw[dotted] (0,\k) --++ (9,0);
\foreach \k in {0.,1.,...,9.}
\draw[dotted] (\k,0) --++ (0,9);
\draw (2,3) -- (2,7)--(4,7)--(4,6)--(6,6)--(6,7)--(7,7)--(7,4)--(6,4)--(6,3)--(7,3)--(7,2)--(4,2)--(4,3)--(2,3) ;
\draw (2,2.5) node{\color{blue} $+$};
\draw (4,3.5) node{\color{blue} $+$};
\draw (4,1.5) node{\color{blue} $+$};
\draw (6,2.5) node{\color{blue} $+$};
\draw (7,3.5) node{\color{blue} $+$};
\draw (4,5.5) node{\color{blue} $+$};
\draw (4,7.5) node{\color{blue} $+$};
\draw (7,7.5) node{\color{blue} $+$};
\end{scope}
\end{tikzpicture}
\end{center}
\caption{The blue line intersects the vertical edges that are linked to the black cross in $(-1)^{\sum_{m=1}^M\sum_{j=2}^L(p_{\alpha_j^m}+p_{\alpha^{m+1}_j})\sum_{i=1}^{j-1}p_{\alpha^m_i}}$. Then the blue crosses indicate the edges for which a $(-1)^{p_a}$ appears, and then after simplification. }
\label{looplbluecrosses}
\end{figure}
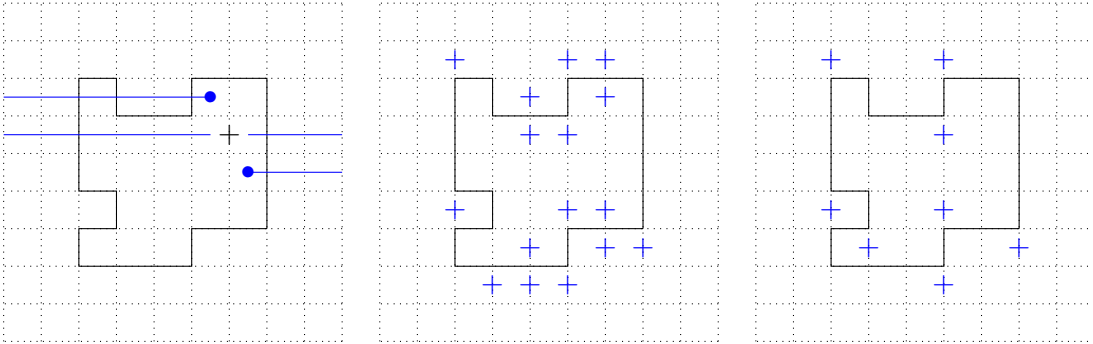
Consider now an $\alpha^m_i$ (vertical edges) that is on the loop $l$. To avoid counting twice a change of sign, only the left-going straight blue lines in the left panel of Figure \ref{looplbluecrosses} must be taken into account. There is a change of sign if and only if $\alpha^m_i$ is inside a NE corner. This gives a $-1$ for each NE corner.

There is now the sign term that comes from the third term in \eqref{eq:r}. This one contributes to $-1$ or $1$ for each SW and NE corner (according to whether the index of the loop is smaller or larger than $N/2$, by inspection of \eqref{eq:r} -- but the sign is the same for both types of corners) in the loop $l$. Together with the $-1$ for each NE corner, it comes that there is finally a $-1$ for each SW \textit{or} NE corner. But the number of NE corners (or of SW corners) is always odd for a loop that can be contracted into a point, which is the case for a non-self intersecting loop that does not cross the boundaries. Therefore the total contribution after the change of index of the loop $l$ is $-1$.

If the loop $l$ crosses only the left and right boundaries, then the previous arguments are still valid (because an horizontal line will always cross the loop an even number of times), but the number of SW or NE corners has now opposite parity as the number of times the left and right boundaries are crossed. Because of the sign term $(-1)^{\sum_{m=1}^M p_{c_1^m}}$ it cancels out and the total sign factor is still $-1$. Note the importance of taking the supertrace of the monodromy matrix to have this extra sign term.

If the loop $l$ crosses the top and bottom boundaries, then the previous arguments are slightly modified (because an horizontal line will always cross the loop a number of times that has same parity as the number of times the up and down boundaries are crossed) but still hold. The only important difference is that there is no equivalent term to $(-1)^{\sum_{m=1}^M p_{c_1^m}}$ for the up and down boundaries, so that the resulting sign factor is $(-1)^{b_v+1}$ where $b_v$ is the number of times the loop crosses the up and down boundaries.

\begin{figure}[H]
\begin{center}
\begin{tikzpicture}[scale=0.5]
\begin{scope}[xscale=-1]
\foreach \k in {0.,1.,...,9.}
\draw[dotted] (0,\k) --++ (9,0);
\foreach \k in {0.,1.,...,9.}
\draw[dotted] (\k,0) --++ (0,9);
\draw (0,3) -- (4,3)--(4,2)--(6,2)--(6,6)--(8,6)--(8,3)--(9,3) ;
\end{scope}

\begin{scope}[xshift=10cm,xscale=-1]
\foreach \k in {0.,1.,...,9.}
\draw[dotted] (0,\k) --++ (9,0);
\foreach \k in {0.,1.,...,9.}
\draw[dotted] (\k,0) --++ (0,9);
\draw (2,0) -- (2,4)--(7,4)--(7,9) ;
\draw (7,0) -- (7,1)--(6,1)--(6,0) ;
\draw (2,9) -- (2,6)--(4,6)--(4,7)--(6,7)--(6,9) ;
\end{scope}
\end{tikzpicture}
\end{center}
\caption{Examples of non-contractible loops in both directions.}
\end{figure}

Finally, the intersections of the loop $l$ with itself can be equally considered as a NW-SE couple of corners (if the indices of the two strands that intersect are the same) or a NE-SW couple of corners (if the indices of the two strands that intersect are primed) multiplied by a $-1$ in both cases. Indeed the fermionic signs stay the same during this transformation after the change of index, by inspection of \eqref{eq:r}. This transforms a loop with $n$ self-intersections into a collection of $n+1$ independent non-self-intersecting loops, but whose indices have to be collectively changed at the same time when the index of the original loop $l$ is changed. This gives an additional $(-1)^{n}$ that is compensated by the $-1$ that comes with each transformation of a self-intersection into a corner.

\end{proof}



\bibliography{sigmamodel}

\begin{thebibliography}{10}

\bibitem{denseloops}
J.~L. Jacobsen, N.~Read, and H.~Saleur, ``Dense loops, supersymmetry, and
  {G}oldstone phases in two dimensions,'' {\em Phys. Rev. Lett.}, vol.~90,
  p.~090601, 2003.

\bibitem{niedermayer}
S.~D. Katz, F.~Niedermayer, D.~N\'ogr\'adi, and C.~T\"or\"ok, ``Comparison of
  algorithms for solving the sign problem in the {O}(3) model in $1+1$
  dimensions at finite chemical potential,'' {\em Phys. Rev. D}, vol.~95,
  p.~054506, Mar 2017.

\bibitem{schomerusreview}
T.~Quella and V.~Schomerus, ``Superspace conformal field theory,'' {\em Journal
  of Physics A: Mathematical and Theoretical}, vol.~46, no.~49, p.~494010,
  2013.

\bibitem{Zirnbauer}
M.~Zirnbauer, ``Conformal field theory of the integer quantum hall plateau
  transition.'' unpublished work; arXiv:hep-th/9905054, 1999.

\bibitem{Zirnbauer1}
M.~Zirnbauer, ``The integer quantum hall plateau transition is a current
  algebra after all.'' unpublished work; arXiv:1805.12555, 2018.

\bibitem{gruzberg}
I.~A. Gruzberg, A.~W.~W. Ludwig, and N.~Read, ``Exact exponents for the spin
  quantum {H}all transition,'' {\em Phys. Rev. Lett.}, vol.~82, 1999.

\bibitem{parisisourlas}
G.~Parisi and N.~Sourlas, ``Self avoiding walk and supersymmetry,'' {\em J.
  Physique Lett.}, vol.~41, 1980.

\bibitem{readsaleursigma}
N.~Read and H.~Saleur, ``Exact spectra of conformal supersymmetric nonlinear
  sigma models in two dimensions,'' {\em Nuclear Physics B}, vol.~613, no.~3,
  pp.~409 -- 444, 2001.

\bibitem{LCFTreview}
``Special issue on logarithmic conformal field theory,'' {\em J. Phys. A},
  vol.~46, no.~49, 2013.

\bibitem{nienhuisrietman}
B.~Nienhuis and R.~Rietman, ``A solvable model for intersecting loops,'' {\em
  ITFA 92-35}, 1993.

\bibitem{martinsnienhuisrietman}
M.~J. Martins, B.~Nienhuis, and R.~Rietman, ``An intersecting loop model as a
  solvable super spin chain,'' {\em Phys. Rev. Lett.}, vol.~81, p.~504, 1998.

\bibitem{affleck}
I.~Affleck, ``Mass generation by merons in quantum spin chains and the {O}(3)
  \ensuremath{\sigma} model,'' {\em Phys. Rev. Lett.}, vol.~56, pp.~408--411,
  Feb 1986.

\bibitem{haldane}
F.~D.~M. Haldane, ``{O}(3) nonlinear $\ensuremath{\sigma}$ model and the
  topological distinction between integer- and half-integer-spin
  antiferromagnets in two dimensions,'' {\em Phys. Rev. Lett.}, vol.~61,
  pp.~1029--1032, Aug 1988.

\bibitem{birgit2}
H.~Saleur and B.~Wehefritz-Kaufmann, ``Integrable quantum field theories with
  supergroup symmetries: the {OSp}(1|2) case,'' {\em Nuclear Physics B},
  vol.~663, no.~3, pp.~443 -- 466, 2003.

\bibitem{frahmmartins}
H.~Frahm and M.~J. Martins, ``Finite-size effects in the spectrum of the
  ${O}sp(3|2)$ superspin chain,'' {\em Nucl. Phys. B.}, vol.~894, p.~665, 2015.

\bibitem{FM}
H.~Frahm and M.~J. Martins, ``The fine structure of the finite-size effects for
  the spectrum of the ${O}sp(n|2m)$ spin chain,'' {\em Nucl. Phys. B},
  vol.~930, p.~545, 2018.

\bibitem{scheunert}
M.~Scheunert, {\em The theory of {L}ie superalgebras}.
\newblock Springer, 1979.

\bibitem{wegner1989}
F.~Wegner, ``Four-loop order $\beta$-function of nonlinear $\sigma$-models in
  symmetric spaces,'' {\em Nucl. Phys. B}, vol.~316, p.~663, 1988.

\bibitem{floratospecther}
E.~G. Floratos and D.~Petcher, ``A two-loop calculation of the mass gap for the
  ${O}({N})$ model in finite volume,'' {\em Nucl. Phys. B}, vol.~252, p.~689,
  1985.

\bibitem{essler}
F.~H.~L. Essler, V.~E. Korepin, and K.~Schoutens, ``Exact solution of an
  electronic model of superconductivity in $1+1$ dimensions,'' {\em Int. J.
  Mod. Phys. B}, vol.~8, p.~3205, 1994.

\bibitem{arnaudonavan2}
D.~Arnaudon, J.~Avan, N.~Cramp\'{e}, A.~Doiku, L.~Frappat, and E.~Ragoucy,
  ``R-matrix presentation for (super)-{Y}angians ${Y}(g)$,'' {\em J. Math.
  Phys.}, vol.~44, p.~302, 2003.

\bibitem{galleasmartins}
W.~Galleas and M.~J. Martins, ``R-matrices and spectrum of vertex models based
  on superalgebras,'' {\em Nucl. Phys. B}, vol.~699, p.~455, 2004.

\bibitem{arnaudonavan}
D.~Arnaudon, J.~Avan, N.~Cramp\'{e}, A.~Doiku, L.~Frappat, and E.~Ragoucy,
  ``Bethe ansatz equations and exact {S} matrices for the $osp(m|2n)$ open
  super spin chain,'' {\em Nucl. Phys. B}, vol.~687, p.~257, 2004.

\bibitem{blotecardy}
H.~W. Bl\"ote, J.~L. Cardy, and M.~P. Nightingale, ``Conformal invariance, the
  central charge, and universal finite-size amplitudes at criticality,'' {\em
  Phys. Rev. Lett.}, vol.~56, 1986.

\bibitem{affleckc}
I.~Affleck, ``Universal term in the free energy at a critical point and the
  conformal anomaly,'' {\em Phys. Rev. Lett.}, vol.~56, 1986.

\bibitem{cardyoperator}
J.~L. Cardy, ``Operator content of two-dimensional conformally invariant
  theories,'' {\em Nucl. Phys. B}, vol.~270, 1986.

\bibitem{afflecklog}
I.~Affleck, D.~Gepner, H.~J. Schulz, and T.~Ziman, ``Critical behaviour of
  spin-s {H}eisenberg antiferromagnetic chains: analytic and numerical
  results,'' {\em J. Phys. A: Math. Gen.}, vol.~22, 1989.

\bibitem{berezintolstoy}
F.~A. Berezin and V.~N. Tolstoy, ``The group with {G}rassmann structure
  ${UOS}p(1|2)$,'' {\em Commun. Math. Phys.}, vol.~78, 1981.

\bibitem{martinsosp12}
M.~J. Martins, ``The exact solution and the finite-size behaviour of the
  ${O}sp(1|2)$-invariant spin chain,'' {\em Nucl. Phys. B}, vol.~450, p.~768,
  1995.

\bibitem{mccoy}
G.~Albertini, S.~Dasmahapatra, and B.~M. McCoy, ``Spectrum and completeness of
  the integrable $3$-state {P}otts model: a finite size study,'' {\em Int. J.
  Mod. Phys. A.}, vol.~7, 1992.

\bibitem{mccoy2}
K.~Fabricius and B.~M. McCoy, ``Bethe's equation is incomplete for the {XXZ}
  model at roots of unity,'' {\em J. Stat. Phys.}, vol.~103, 2001.

\bibitem{cardy1986}
J.~L. Cardy, ``Logarithmic corrections to finite-size scaling in strips,'' {\em
  J. Phys. A: Math. Gen.}, vol.~19, p.~L1093, 1986.

\bibitem{kausch}
H.~G. Kausch, ``Curiosities at $c=-2$,'' {\em ArXiv:hep-th/9510149}, 1995.

\bibitem{lukyanov}
S.~Lukyanov, ``Low energy effective {H}amiltonian for the {XXZ} spin chain,''
  {\em Nucl. Phys. B}, vol.~522, 1998.

\bibitem{ikhlef2}
Y.~Ikhlef, J.~L. Jacobsen, and H.~Saleur, ``An integrable spin chain for the
  ${SL}(2,{R})/{U}(1)$ black hole sigma model,'' {\em Phys. Rev. Lett.},
  vol.~108, p.~081601, 2012.

\bibitem{lu}
D.~Lu, ``On the classification of irreducible $osp(2|2)$ representations,''
  {\em Kyushu J. Math.}, vol.~66, 2012.

\bibitem{galleasmartins2}
W.~Galleas and M.~J. Martins, ``Exact solution and finite size properties of
  the ${U}_q(osp(2|2m))$ vertex model,'' {\em Nucl. Phys. B}, vol.~768, p.~219,
  2007.

\bibitem{vanderjeugt}
J.~V. der Jeugt, ``Finite and infinite dimensional representations of the
  orthosymplectic superalgebra $osp(3|2)$,'' {\em J. Math. Phys.}, vol.~25,
  p.~3334, 1984.

\bibitem{razumovstroganov}
A.~V. Razumov and Y.~G. Stroganov, ``Combinatorial nature of ground state
  vector of ${O}(1)$ loop model,'' {\em Theor. Math. Phys.}, vol.~138, 2004.

\bibitem{knutsonpzj}
A.~Knutson and P.~Zinn-Justin, ``A scheme related to the {B}rauer loop model,''
  {\em Advances in Mathematics}, vol.~214, 2007.

\bibitem{candu1}
C.~Candu, V.~Mitev, and V.~Schomerus, ``Spectra of coset sigma models,'' {\em
  Nucl. Phys. B}, vol.~877, p.~900, 2013.

\bibitem{cagnazzo}
A.~Cagnazzo, V.~Schomerus, and V.~Tlapak, ``On the spectrum of superspheres,''
  {\em Journal of High Energy Physics}, vol.~2015, p.~13, Mar 2015.

\bibitem{mitev}
V.~Mitev, T.~Quella, and V.~Schomerus, ``Principal chiral model on
  superspheres,'' {\em Journal of High Energy Physics}, vol.~2008, no.~11,
  p.~086, 2008.

\bibitem{candu2}
C.~Candu and H.~Saleur, ``A lattice approach to the conformal
  ${OS}p(2{S}+2|2{S})$ supercoset sigma model. {P}art {I}: algebraic structures
  in the spin chain. {T}he {B}rauer algebra,'' {\em Nucl. Phys. B}, vol.~808,
  p.~441, 2009.

\bibitem{candu3}
C.~Candu and H.~Saleur, ``A lattice approach to the conformal
  ${OS}p(2{S}+2|2{S})$ supercoset sigma model. {P}art ii: the boundary
  spectrum,'' {\em Nucl. Phys. B}, vol.~808, p.~487, 2009.

\bibitem{saleur1999}
H.~Saleur, ``The continuum limit of $sl({N}|{K})$ integrable super spin
  chains,'' {\em Nucl. Phys. B}, vol.~578, p.~552, 2000.

\bibitem{martinsramos}
M.~J. Martins and P.~B. Ramos, ``Solution of a supersymmetric model of
  correlated electrons,'' {\em Phys. Rev. B}, vol.~56, 1997.

\bibitem{kasteleyn}
P.~W. Kasteleyn, ``Dimer statistics and phase transitions,'' {\em J. Math.
  Phys.}, vol.~4, 1963.

\bibitem{saleurpolymer}
H.~Saleur, ``Polymers and percolation in two dimensions and twisted {N}=2
  supersymmetry,'' {\em Nucl. Phys. B}, vol.~382, 1992.

\bibitem{canduglnm}
C.~Candu, ``Continuum limit of $gl({M}|{N})$,'' {\em JHEP}, vol.~1107, 2011.

\bibitem{gourdin}
M.~Gourdin, ``Relation between the supertableaux of the supergroups
  ${OS}p(2|2)$ and ${SU}(1|2)$,'' {\em J. Math. Phys.}, vol.~27, 1986.

\bibitem{ivashkevich}
E.~V. Ivashkevich, ``Correlation functions of dense polymers and $c=-2$
  {C}onformal {F}ield {T}heory,'' {\em J. Phys. A}, vol.~32, p.~1691, 1999.

\bibitem{caracciolojacobsensaleur}
S.~Caracciolo, J.~L. Jacobsen, H.~Saleur, A.~D. Sokal, and A.~Sportiello,
  ``Fermionic field theory for trees and forests,'' {\em Phys. Rev. Lett.},
  vol.~93, p.~080601, 2004.

\bibitem{nienhuis82}
B.~Nienhuis, ``Exact critical point and critical exponents of {$O(n)$} models
  in two dimensions,'' {\em Phys. Rev. Lett.}, vol.~49, 1982.

\bibitem{nienhuis84}
B.~Nienhuis, ``Critical behavior of two-dimensional spin models and charge
  asymmetry in the {C}oulomb gas,'' {\em J. Stat. Phys.}, 1984.

\bibitem{nienhuis87}
B.~Nienhuis, ``Coulomb gas formulations of two-dimensional phase transitions,''
  {\em Phase transitions and critical phenomena}, vol.~11, 1987.

\bibitem{saleurduplantier2}
B.~Duplantier and H.~Saleur, ``Exact critical properties of two-dimensional
  dense self-avoiding walks,'' {\em Nucl. Phys. B}, vol.~290, p.~291, 1987.

\bibitem{nahum}
A.~Nahum, P.~Serna, A.~M. Somoza, and M.~Ortu{\~n}o, ``Loop models with
  crossings,'' {\em Phys. Rev. B}, vol.~87, p.~184204, 2013.

\bibitem{balog3}
J.~Balog and A.~Heged\'us, ``The finite size spectrum of the 2-dimensional
  ${O}(3)$ non-linear sigma-model,'' {\em Nucl. Phys. B}, vol.~829, 2010.

\bibitem{vilenkin}
N.~J. Vilenkin, {\em Special functions and the theory of group
  representations}.
\newblock American Mathematical Society, 1968.

\bibitem{Garrettreview}
P.~Garrett, ``Harmonic analysis on spheres,'' 2010.

\bibitem{polyakov}
A.~Polyakov, ``Interaction of goldstone particles in two dimensions.
  {A}pplications to ferromagnets and massive {Y}ang-{M}ills fields,'' {\em
  Physics Letters B}, vol.~59, no.~1, pp.~79 -- 81, 1975.

\bibitem{ikhlef}
Y.~Ikhlef, J.~L. Jacobsen, and H.~Saleur, ``A staggered six-vertex model with
  non-compact continuum limit,'' {\em Nucl. Phys. B}, vol.~789, p.~483, 2008.

\bibitem{granetjacobsensaleur}
E.~Granet, J.~L. Jacobsen, and H.~Saleur, ``A distribution approach to
  finite-size corrections in {B}ethe ansatz solvable models,'' {\em Nucl. Phys.
  B.}, vol.~934, p.~96, 2018.

\bibitem{woynaeckle}
F.~Woynarovich and H.-P. Eckle, ``Finite-size corrections and numerical
  calculations for long spin $1/2$ {H}eisenberg chains in the critical
  region,'' {\em J. Phys. A: Math. Gen.}, vol.~20, p.~L97, 1987.

\bibitem{hamerbatchelorbarber}
C.~J. Hamer, M.~T. Batchelor, and M.~N. Barber, ``Logarithmic corrections to
  finite-size scaling in the four-state {P}otts model,'' {\em J. Stat. Phys.},
  vol.~52, 1988.

\end{thebibliography}
\bibliographystyle{ieeetr}

\end{document}